%% file: paper.tex
\pdfoutput=1
\documentclass[onecolumn,12pt]{IEEEtran}

\usepackage{amsmath,amssymb,amsthm,mathtools,dsfont}
\usepackage{graphicx}
\usepackage{booktabs}
\usepackage{tikz}
\usetikzlibrary{arrows.meta,positioning}
\usepackage[colorlinks=true,allcolors=blue]{hyperref}
\usepackage[capitalise]{cleveref}

\theoremstyle{plain}
\newtheorem{theorem}{Theorem}
\newtheorem{lemma}[theorem]{Lemma}
\newtheorem{corollary}[theorem]{Corollary}
\newtheorem{proposition}[theorem]{Proposition}
\theoremstyle{definition}
\newtheorem{definition}[theorem]{Definition}
\newtheorem{remark}[theorem]{Remark}
\newtheorem{example}[theorem]{Example}

\DeclareMathOperator*{\argmax}{arg\,max}
\DeclareMathOperator*{\argmin}{arg\,min}

\newcommand{\cA}{{\mathcal A}}\newcommand{\cB}{{\mathcal B}}\newcommand{\cC}{{\mathcal C}}
\newcommand{\cD}{{\mathcal D}}\newcommand{\cE}{{\mathcal E}}\newcommand{\cF}{{\mathcal F}}
\newcommand{\cH}{{\mathcal H}}
\newcommand{\cM}{{\mathcal M}}
\newcommand{\cN}{{\mathcal N}}\newcommand{\cP}{{\mathcal P}}
\newcommand{\cS}{{\mathcal S}}\newcommand{\cT}{{\mathcal T}}
\newcommand{\cU}{{\mathcal U}}
\newcommand{\cX}{{\mathcal X}}\newcommand{\cY}{{\mathcal Y}}

\newcommand{\mR}{{\mathbb R}}

\newcommand{\bbm}{\mathbf{m}}\newcommand{\bbW}{\mathbf{W}}\newcommand{\bbC}{\mathbf{C}}
\newcommand{\bbQ}{\mathbf{Q}}\newcommand{\bbU}{\mathbf{U}}

\newcommand{\BRA}[1]{\left( #1 \right)}

\newcommand{\BRAi}[1]{\left\langle #1 \right\rangle}
\newcommand{\BRAs}[1]{\left\{ #1 \right\}}
\newcommand{\abs}[1]{\left| #1 \right|}
\newcommand{\norm}[1]{\lVert #1 \rVert}

\newcommand{\PR}[1]{\mathbb{P}\left\{ #1 \right\}}
\newcommand{\PRs}[2]{\mathbb{P}_{#1}\left\{ #2 \right\}}
\newcommand{\E}[1]{\mathbb{E}\left( #1 \right)}
\newcommand{\Es}[2]{\mathbb{E}_{#1}\left( #2 \right)}

\newcommand{\Unif}[1]{\cU\left( #1 \right)}
\newcommand{\uU}{\Unif{[0,1]}}
\newcommand{\Ind}[1]{\mathds{1}_{\BRAs{#1}}}

\newcommand{\PEP}[2]{p_e\BRA{#1 \mid #2}}

\newcommand{\PEPU}[3]{p_e\BRA{#1 \mid #2; #3}}

\newcommand{\Fspec}[2]{F\!\BRA{#1\,;\,#2}}

\newcommand{\RCUp}{\mathrm{RCU}^{+}}

\newcommand{\MMI}{\mathrm{MMI}}

\begin{document}

\title{Universal Decoding via the Pairwise Error Probability}

\author{Nir~Elkayam~and~Meir~Feder%
\thanks{The authors are with the Department of Electrical Engineering--Systems,
Tel Aviv University, Tel Aviv, Israel (e-mail: nir.elkayam@gmail.com; meir@eng.tau.ac.il).}%
\thanks{This paper builds on the pairwise-error-probability framework
of the companion paper~\cite{elkayampep1}.}}

\maketitle

\begin{abstract}
We develop a theory of universal decoding on the pairwise-error-probability (PEP)
primitive of a companion paper. The starting point is that the PEP and its
error spectrum are defined for an \emph{arbitrary} decoding metric, and
that, while raw metric values across a family share no common scale, the PEP
supplies one. Universal
decoding --- decoding well simultaneously against a whole family of metrics or
channels --- therefore becomes the problem of \emph{merging} a family of
per-metric spectra into a single decoding rule. We construct such a rule from a
clipped inverse-PEP statistic. The construction rests on a Kraft-type
inequality for decoding, valid for any input prior: per output, every metric
canonicalizes into a conditional probability assignment on the codewords, and
the merge is the normalized-maximum-likelihood envelope of the induced family.
We prove the rule is random-coding universal (it loses
only a vanishing rate relative to the best metric in the family, against every
channel), show it is an asymptotic minimax/equalizer rule, and derandomize it:
over a subexponential channel family a single deterministic code inherits the
guarantee. Specialized
to discrete memoryless channels the construction reduces to types and recovers
the maximum-mutual-information decoder, its tilted variant for non-uniform
memoryless input, and finite-state universal decoders. It extends to a
decoding rule with an erasure option (deterministically, uniformly over erasure
margins), and, via a discretization of separable
metric families, to continuous-alphabet channels: for AWGN with deterministic
interference it attains the matched-ML exponent up to an explicit typical-set
cap, and for ISI channels it attains the best exponent in a family of
equalize-and-decode rules, under explicit assumptions on the equalizers and
the channel spectrum and the same cap. Throughout, universality is
obtained as a corollary of the PEP analysis rather than as a separate theory.
\end{abstract}


\input{sections/01-intro}
\input{sections/02-framework}
\input{sections/03-universal-metric}
\input{sections/04-dmc}
\input{sections/05-erasure}

\input{sections/06-awgn}

\input{sections/07-conclusion}
\input{sections/08-appendix}

\bibliographystyle{IEEEtran}
\bibliography{refs}

\end{document}

%% file: sections/01-intro.tex
\section{Introduction}\label{sec:intro}

A universal decoder must decode well without knowing the channel, or without
committing to a single decoding metric: it competes, at every blocklength,
against the best decoder in a family. The landmarks of this program are the
maximum-mutual-information (MMI) decoder of
Goppa~\cite{goppa1975nonprobabilistic}, Csisz\'ar~\cite{csiszar1974extremum},
and Csisz\'ar--K\"orner~\cite{csiszar1981information} for the
family of discrete memoryless channels, the finite-grid constructions of Feder
and Lapidoth~\cite{feder1998universal} for separable channel families, the
competitive-minimax approach of Feder and
Merhav~\cite{feder2002competitive,akirav2007competitive},
and their relatives. This paper develops universal decoding directly on the
pairwise-error-probability (PEP) primitive of~\cite{elkayampep1}.

A decoder is invariant to monotone transformations of its metric: only the
order of the scores matters, and their numerical values carry no meaning of
their own. Raw metric values from \emph{different} family members are
therefore mutually incommensurable, which is exactly why the classical route
of merging by weighted likelihood mixtures is sensitive to each member's
arbitrary scale and needs weighting or grid conditions to control it. The PEP
canonicalizes each metric first --- order-preservingly, so member-by-member
decoding is unchanged --- into a probability whose value is intrinsic: the
probability that a random competitor outranks the candidate at the observed
output. The merge then compares numbers that mean the same thing for every
member, and the only price is a subexponential union term.

This program runs on the observation that the PEP and its error spectrum are
defined for an \emph{arbitrary} decoding metric (\cref{sec:recap}). A family
of metrics $\cM=\{m_\theta\}$ therefore induces a family of spectra; a
universal decoder is a single metric whose spectrum is, up to a vanishing rate
loss, as good as the best $m_\theta$ for every channel in the family. We
construct such a metric by carrying out the canonicalize-then-merge step
above: a clipped inverse-PEP statistic pools the per-metric pairwise
comparisons into one decision rule (\cref{sec:univ-metric}).
Underlying the construction is a Kraft-type inequality for decoding
(\cref{rem:uc-kraft}): relative to any input prior --- discrete or continuous
--- every metric canonicalizes, per output, into a conditional probability
assignment on the codewords. Universal decoding thereby becomes universal
prediction over the induced assignments, with the merge as their
normalized-maximum-likelihood envelope.

The merge metric is shown to be \emph{random-coding universal}: its
random-coding exponent matches that of the best metric in the family,
simultaneously for every channel. It is also an asymptotic minimax/equalizer
rule. By expurgation the guarantee derandomizes, so that over a subexponential
channel family a single deterministic code inherits it (\cref{sec:framework}).

Specialized to discrete memoryless channels the construction reduces to the
method of types and recovers the MMI decoder, its tilted variant, and
finite-state universal decoders (\cref{sec:dmc}). It extends to a decoder with an
erasure option (\cref{sec:erasure}). Through a discretization of
\emph{separable} metric families, it reaches continuous-alphabet channels ---
matched-ML universality up to a typical-set cap for AWGN with deterministic
interference, and, for ISI, the best exponent in a family of
equalize-and-decode rules, under explicit assumptions on the equalizers and
the channel spectrum and the same cap (\cref{sec:awgn}). In every case universality is a
corollary of the PEP analysis: the same spectrum that drove achievability and
converse in~\cite{elkayampep1} now drives the comparison across a family.

\section{Recap of the PEP Primitive}\label{sec:recap}

We work at blocklength $n$ with alphabets $\cX,\cY$; sequences are
$x^n\in\cX^n$, $y^n\in\cY^n$, and a (per-blocklength) decoding metric is
an extended-real-valued $m^n:\cX^n\times\cY^n\to[-\infty,\infty]$ (the
convention is fixed in \cref{subsec:uc-notation}). Codes have $M_n=e^{nR}$ messages,
with codewords drawn i.i.d.\ from a prior $Q_X^n$. For a fixed $y^n$ and candidate $x^n$, with ties
broken lexicographically by an i.i.d.\ dither $U\sim\uU$, the dithered PEP is
$\PEPU{x^n}{y^n}{u}\triangleq Q_X^n\{m^n(\cdot,y^n)>m^n(x^n,y^n)\}+u\,
Q_X^n\{m^n(\cdot,y^n)=m^n(x^n,y^n)\}$, and $\PEP{x^n}{y^n}\triangleq\PEPU{x^n}{y^n}{U}$.
We use two facts, valid for \emph{any} metric: fixed-output
uniformity of the PEP, proved in~\cite{elkayampep1} via a randomized
probability integral transform, and the spectral exponent reading of the
random-coding functional, proved by a short slicing argument in
Appendix~\ref{app:exponent-reading}.

\begin{proposition}[PEP uniformity~{\cite{elkayampep1}}]\label{prop:pep-uniform}
For every fixed $y^n$, if $(X^n,U)\sim Q_X^n\times\uU$ then
\begin{equation*}
  \PEP{X^n}{y^n}\sim\uU.
\end{equation*}
\end{proposition}

Under the channel-induced joint law $(X^n,Y^n)\sim Q_X^n\cdot W^n$, define the
error spectrum
\begin{equation}\label{eq:recap-spectrum}
\Fspec{Q_X^n}{z}\triangleq\PR{-\tfrac1n\log\PEP{X^n}{Y^n}\le z}
\end{equation}
(the CDF of the normalized negative log-PEP). The random-coding error probability
is governed by this spectrum through an elementary identity of the companion
paper~\cite{elkayampep1}: the $\RCUp$ bound is the kernel integral
$\bar P_e\le \Fspec{Q_X^n}{R}+e^{nR}\int_R^\infty e^{-nz}\,d\Fspec{Q_X^n}{z}$ (the below-rate
spectrum mass plus the kernel integral over the tail). (\Cref{prop:uc-erc-pep} below
records the exponent-level form of the random-coding functional that this
paper uses.)

\begin{proposition}[Exponent reading]\label{prop:exponent}
For every $n$, rate $R>0$, and $\rho\ge0$,
\begin{equation}\label{eq:recap-sandwich}
  e^{-n\rho}\,\Fspec{Q_X^n}{R+\rho}
  \;\le\;
  \E{\min\BRAs{1,\,e^{nR}\,\PEP{X^n}{Y^n}}}
  \;\le\;
  e\,(n^2+1)\sup_{\rho'\ge0}e^{-n\rho'}\Fspec{Q_X^n}{R+\rho'}\;+\;e^{-n^2}.
\end{equation}
Consequently, with $E_1(z)\triangleq-\tfrac1n\log\Fspec{Q_X^n}{z}$ the
per-blocklength spectrum exponent,
\begin{equation*}
  -\tfrac1n\log\E{\min\BRAs{1,\,e^{nR}\,\PEP{X^n}{Y^n}}}
  \;=\;\inf_{\rho\ge0}\bigl\{E_1(R+\rho)+\rho\bigr\}
  \;+\;O\Bigl(\tfrac{\log n}{n}\Bigr)
\end{equation*}
whenever $\inf_{\rho\ge0}\{E_1(R+\rho)+\rho\}\le n$ --- in particular whenever
it stays bounded in $n$. Thus comparing decoders reduces to comparing their
spectra (equivalently, the exponents $E_1$) on the channel of interest.
\end{proposition}

\begin{IEEEproof}[Proof sketch]
The lower bound restricts to the event $\{-\tfrac1n\log p_e\le R+\rho\}$; the
upper bound slices the log-PEP range at spacing $1/n$, each slice dominated by
the supremum with an $e\,(n^2+1)$ count and the remainder super-exponentially
negligible. The short computation is in Appendix~\ref{app:exponent-reading}.
\end{IEEEproof}

This is the exponent-scale form of the kernel-integral identity
of~\cite{elkayampep1}.

The universal-decoding problem is now stated in spectrum terms: given a family of
metrics (or channels), find one metric whose spectrum exponent is, for every
member of the family, no worse than that member's own --- up to $o(1)$ in the
rate. The rest of the paper constructs and analyzes such a metric.

%% file: sections/02-framework.tex

\providecommand{\Erc}{E_{\mathrm{rc}}}
\providecommand{\Edet}{E_{\det}}
\providecommand{\Eera}{E_{\mathrm{era}}}
\providecommand{\Eund}{E_{\mathrm{und}}}
\providecommand{\Mset}{\cM}            
\providecommand{\Cset}{\cC}            
\providecommand{\Md}{\cM_{d}}          
\providecommand{\PEPm}[3]{p_e\!\BRA{#1 \mid #2}_{#3}}  
\providecommand{\Ehat}{\widehat{E}}    
\providecommand{\Ihat}{\hat I}         

\section{Framework and Definitions}\label{sec:framework}

This section assembles the universal-decoding machinery used in the rest of
the paper. We work at blocklength $n$ with finite or measurable alphabets
$\cX,\cY$; sequences live in $\cX^n,\cY^n$, codes have rate $R>0$ and
$M_n=e^{nR}$ messages (up to subexponential factors). Throughout we reuse,
without re-derivation, the pairwise-error-probability (PEP) primitive of
the companion paper~\cite{elkayampep1} recalled in \cref{sec:recap}: in particular the
fixed-output uniformity of the PEP (\cref{prop:pep-uniform}) and the
spectrum--exponent reading (\cref{prop:exponent}). The universal-decoding
problem is to find a \emph{single} metric whose PEP spectrum is, against every
member of a given family of metrics or channels, no worse than that member's
own --- up to an $o(1)$ rate loss.

\subsection{Notation: sequences, families, and exponents}
\label{subsec:uc-notation}

Boldface denotes a sequence indexed by the blocklength:
$\bbm=\{m^n\}_{n\ge1}$ (metric), $\bbW=\{W^n\}_{n\ge1}$ (channel),
$\bbQ_X=\{Q_X^n\}_{n\ge1}$ (input prior), $\bbC=\{C^n\}_{n\ge1}$ (codebook),
and $\bbU=\{U^n\}_{n\ge1}$ (the universal/merge metric constructed in
\cref{sec:univ-metric}). Calligraphic letters $\Mset,\Cset$ denote
\emph{families} of such sequences; the \emph{section} at blocklength $n$ is
$\Mset^n\triangleq\{m^n:\bbm\in\Mset\}$, and likewise $\Cset^n$. A metric at
blocklength $n$ is a measurable function
$m^n:\cX^n\times\cY^n\to[-\infty,\infty]$, extended-real-valued throughout:
the values $\pm\infty$ are admitted --- the tilted type metrics of
\cref{sec:dmc} take the value $-\infty$ --- since maximum-metric decoding
uses only the order of the scores. The
decoding rule is maximum-metric decoding
$\hat\imath(y^n)=\argmax_i m^n(x_i^n,y^n)$ with ties counted as errors.

For a nonnegative sequence $\{a_n\}$ we write its error exponent as
\begin{equation}\label{eq:uc-exp-def}
  \Ehat(a_n)\;\triangleq\;\liminf_{n\to\infty}-\tfrac1n\log a_n\;\in\;[0,\infty],
\end{equation}
and say $a_n$ has \emph{no worse exponent than} $b_n$ if
$\Ehat(a_n)\ge\Ehat(b_n)$.

\paragraph{Tie-as-error PEP.}
The universal-decoding analysis uses the \emph{tie-as-error} variant of the
PEP of \cref{sec:recap}: for a metric $m^n$ and a pair
$(x^n,y^n)$, with a competitor $\tilde X^n\sim Q_X^n$ drawn independently,
\begin{equation}\label{eq:uc-pep-tie}
  \PEPm{x^n}{y^n}{m^n}\;\triangleq\;
  Q_X^n\BRAs{\tilde X^n:\,m^n(\tilde X^n,y^n)\ge m^n(x^n,y^n)}.
\end{equation}
This differs from the dithered PEP of \cref{sec:recap} only in the treatment
of metric ties --- here ties are charged in full rather than broken by an
auxiliary uniform dither. The two coincide whenever the PEP is atomless (the
typical case under continuous noise or generic priors); when ties carry
positive mass, \eqref{eq:uc-pep-tie} dominates the dithered PEP pointwise and
yields an everywhere-valid upper bound on the random-coding error
probability. The development below uses only this monotonicity together with
the fixed-output uniformity of \cref{prop:pep-uniform}; it does not require
any sharper property of the tie-as-error form. We record the consequence
used repeatedly (in the merge bounds and throughout the erasure analysis).

\begin{lemma}[Tie-vs-dither transfer]
  \label{lem:uc-tie-dither}
  Fix $y^n$ and a metric $m^n$, and let $\tilde X^n\sim Q_X^n$. Then for
  every $t\ge0$,
  $\PR{-\log\PEPm{\tilde X^n}{y^n}{m^n}\ge t}\le e^{-t}$, with equality for
  the dithered PEP.
\end{lemma}

\begin{IEEEproof}
  Charging ties in full only increases the PEP, so
  $\{\PEPm{\cdot}{y^n}{m^n}\le e^{-t}\}$ is contained in the corresponding
  dithered event, whose probability is exactly $e^{-t}$ by
  \cref{prop:pep-uniform}.
\end{IEEEproof}

\subsection{Random-coding and deterministic error exponents}
\label{subsec:uc-error-exponents}

Let $P_e(C^n,W^n,m^n)$ be the average error probability of code $C^n$ under
channel $W^n$ and metric $m^n$, with a uniform message and tie-as-error
decoding.

\begin{definition}[Deterministic and random-coding exponents]
  \label{def:uc-error-exponents}
  For a rate $R>0$, channel sequence $\bbW$, prior $\bbQ_X$, and metric
  sequence $\bbm$,
  \begin{align}
    \Edet(R;\bbW,\bbC,\bbm) &\triangleq \Ehat\bigl(P_e(C^n,W^n,m^n)\bigr),
      \label{eq:uc-e-det}\\
    \Erc(R;\bbW,\bbQ_X,\bbm) &\triangleq \Ehat\bigl(\E{P_e(C^n,W^n,m^n)}\bigr),
      \label{eq:uc-e-rc}
  \end{align}
  where in \eqref{eq:uc-e-rc} the codebook $C^n$ has $M_n=e^{nR}$ codewords
  drawn i.i.d.\ $\sim Q_X^n$ and the expectation is over the codebook, the
  uniform message, and the channel output.
\end{definition}

For the i.i.d.\ random-coding ensemble, $\Erc$ has a representation purely
through the PEP --- the single functional that drives the entire analysis.

\begin{proposition}[Random-coding exponent via the PEP]
  \label{prop:uc-erc-pep}
  For any $\bbW$, $\bbQ_X$, $\bbm$, and rate $R>0$,
  \begin{equation}\label{eq:uc-erc-pep}
    \Erc(R;\bbW,\bbQ_X,\bbm)
    \;=\;
    \Ehat\Bigl(\E{\min\BRAs{1,\,e^{nR}\,\PEPm{X^n}{Y^n}{m^n}}}\Bigr),
  \end{equation}
  with $(X^n,Y^n)\sim Q_X^n\cdot W^n$.
\end{proposition}

\begin{IEEEproof}
  Conditioned on $(X^n,Y^n)=(x,y)$, each of the $M_n-1$ i.i.d.\ competitors
  ties-or-beats $x$ under $m^n$ with probability exactly $\PEPm{x}{y}{m^n}$
  by \eqref{eq:uc-pep-tie}, so the conditional error probability is
  $1-(1-\PEPm{x}{y}{m^n})^{M_n-1}$. The elementary inequality
  $\tfrac12\min\{1,Np\}\le1-(1-p)^N\le\min\{1,Np\}$, with
  $N=M_n-1\in[\tfrac12e^{nR},e^{nR}]$ for $n$ large, sandwiches this by
  $\min\{1,e^{nR}\PEPm{x}{y}{m^n}\}$ up to constant factors. The constants
  vanish under $-\tfrac1n\log$ and the $\liminf$, giving
  \eqref{eq:uc-erc-pep}.
\end{IEEEproof}

All universality statements below are phrased through this functional.

\subsection{Two notions of universality}
\label{subsec:uc-universality}

We compare a candidate metric simultaneously against a family of metrics
$\Mset$ and across a family of channels $\Cset$. The two families are
\emph{independent} inputs: $\Mset$ says \emph{which decoders we compete
against}, $\Cset$ says \emph{for which channels universality must hold}.

\begin{definition}[Subexponential growth, (H-sub)]
  \label{def:uc-hsub}
  A family $\Mset$ (resp.\ $\Cset$) has \emph{subexponential growth} if its
  sections are finite and
  $\tfrac1n\log|\Mset^n|\to0$ (resp.\ $\tfrac1n\log|\Cset^n|\to0$). We write
  $|\Mset^n|=e^{o(n)}$.
\end{definition}

\begin{definition}[Universality, random-coding and deterministic]
  \label{def:uc-universality}
  Fix a prior $\bbQ_X$, families $\Mset,\Cset$, and a rate $R>0$.
  \begin{enumerate}
  \item[\textup{(rc)}] A metric sequence $\bbU$ is \emph{random-coding
    universal at rate $R$ w.r.t.\ $(\Mset,\Cset)$} if
    \begin{equation}\label{eq:uc-huniv}
      \Erc(R;\bbW,\bbQ_X,\bbU)\;\ge\;\Erc(R;\bbW,\bbQ_X,\bbm)
      \quad\text{for every }\bbW\in\Cset,\ \bbm\in\Mset .
    \end{equation}
    The \emph{fixed-channel} case is $\Cset=\{\bbW\}$. We abbreviate
    \eqref{eq:uc-huniv} as (H-univ($R$)).
  \item[\textup{(det)}] A pair $(\bbC,\bbU)$ is \emph{deterministically
    universal at rate $R$ w.r.t.\ $(\Mset,\Cset)$} if
    $\Edet(R;\bbW,\bbC,\bbU)\ge\Erc(R;\bbW,\bbQ_X,\bbm)$ for every
    $\bbW\in\Cset$ and $\bbm\in\Mset$.
  \end{enumerate}
\end{definition}

\begin{remark}[Role of the prior]
  \label{rem:uc-role-prior}
  The universal metric $\bbU$ is \emph{allowed to depend on $\bbQ_X$} --- and
  in our construction does, since $U^n$ is built from the PEPs
  $\PEPm{\cdot}{\cdot}{m^n}$, which are functionals of $Q_X^n$. The
  \emph{same} prior is used across the whole channel family $\Cset$: this is
  the natural universal setting, as the encoder does not know $\bbW$, and a
  per-channel prior would require channel knowledge at the encoder. No
  structural assumption (product, symmetry, finite support) is placed on
  $Q_X^n$; the clipping device of \cref{sec:univ-metric} removes the
  ``normal prior'' condition of the earlier merge~\cite{elkayam2014universal}.
\end{remark}

\subsection{From random-coding to deterministic universality}
\label{subsec:uc-derand}

The deterministic notion follows from the random-coding one by expurgation,
at no exponent cost, whenever the channel family is subexponential.

\begin{theorem}[Derandomization]
  \label{thm:uc-rc-to-det}
  Suppose $\Cset$ satisfies (H-sub) and $\bbU$ satisfies (H-univ($R$)).
  Then there is a code sequence $\bbC$ of rate $R$ such that $(\bbC,\bbU)$ is
  deterministically universal at rate $R$ w.r.t.\ $(\Mset,\Cset)$.
\end{theorem}

\begin{IEEEproof}
  For each $W^n\in\Cset^n$, Markov's inequality bounds the fraction of
  i.i.d.\ codebooks whose error probability exceeds $t$ times the
  random-coding mean by $1/t$. Take $t=|\Cset^n|^2$ and union-bound over
  $\Cset^n$; (H-sub) makes the $|\Cset^n|^2$ penalty subexponential, so a
  positive fraction of codebooks meets every channel-specific bound
  simultaneously. Picking one such $C^n$ for each $n$ and assembling over $n$
  gives $\bbC$; only the codebook is derandomized, the metric $\bbU$ is
  unchanged.
\end{IEEEproof}

Note that no growth condition on $\Mset$ is needed in
\cref{thm:uc-rc-to-det}: the expurgated decoder is $\bbU$ and does not
depend on the metric family, so $\Mset$ enters no expurgation constraint.

When the channel family of interest is a continuum, the expurgation above
certifies the code only on a subexponential representative test family
$\Cset_d$. For the Gaussian families of \cref{sec:awgn} the guarantee extends
from $\Cset_d$ to the full parameter continuum through a change-of-measure
device of Feder and Lapidoth, which we use by citation and do not reprove.

\begin{lemma}[Change of measure for a fixed code, Feder--Lapidoth
  {\cite[Lem.~8]{feder1998universal}}]
  \label{lem:uc-transfer}
  Let $W^n,W'^n$ be channels and $C^n$ a codebook such that for every codeword
  $x^n\in C^n$ there is a measurable $B_{x^n}\subset\cY^n$ with
  $W^n(y^n\mid x^n)\le e^{n\delta}\,W'^n(y^n\mid x^n)$ for all
  $y^n\notin B_{x^n}$, and $W^n(B_{x^n}\mid x^n)\le e^{-nE'}$. Then for
  \emph{every} decoder $\varphi$,
  \[
    P_e(C^n,W^n,\varphi)\;\le\;e^{n\delta}\,P_e(C^n,W'^n,\varphi)+e^{-nE'}.
  \]
\end{lemma}

The finite-alphabet \cite[Lem.~5]{feder1998universal} is the special case
$B_{x^n}=\BRAs{y^n:W^n(y^n\mid x^n)\le e^{-n(E'+\log|\cY|)}}$.

\begin{remark}[From the test family to the continuum]
  \label{rem:uc-cd-to-c}
  The hypothesis of \cref{lem:uc-transfer} compares channel \emph{laws}
  codeword by codeword, which is why the conclusion holds uniformly over
  decoders --- exactly what pointwise PEP domination cannot supply, since the
  latter compares decoders under a fixed law and says nothing about a fixed
  codebook's error under a different law. Applied twice per channel, the
  lemma upgrades the deterministic guarantee from the test family $\Cset_d$
  to the full continuum, at the cost of an arbitrarily small $\delta$ and an
  exponent-negligible rare term per application, provided $E'$ exceeds the
  target exponent. The two applications run in opposite orientations: one
  carries the expurgated code's achieved error from $\bbW$ to its
  representative $\bbW_d$; the other, with the roles exchanged, carries the
  benchmark exponent from $\bbW_d$ back to $\bbW$. Each orientation needs
  its own rare set under its own dominated law (the two-sided form of
  Feder--Lapidoth's strong separability); both Gaussian verifications are
  two-sided because the underlying log-likelihood-ratio bounds are absolute
  values. For the
  Gaussian ISI family the law-closeness holds on the $O(\sqrt n)$ output
  shell, off the shell-escape event; this is Feder--Lapidoth's own
  verification~\cite{feder1998universal}, which we do not restate. For the
  finite-dimensional interference family the verification is a shift bound: a $1/n$-net
  keeps the log-likelihood gap $O(1)$ off a shell-escape event whose
  probability is $e^{-nE'}$ with $E'$ arbitrarily large in the shell
  constant (Appendix~\ref{app:rem-interf-shift}). \emph{Absent this lemma, the
  deterministic claims of this paper for continuous channel families are read
  over the subexponential test family $\Cset_d$}, as the hypotheses of
  \cref{thm:uc-rc-to-det} already require.
\end{remark}

\subsection{Two orthogonal axes}
\label{subsec:uc-axes}

Two largely independent reductions organize the development.

\paragraph{Axis 1 (derandomization).}
The random-coding notion delivers one metric $\bbU$ that, against \emph{any}
channel, attains the random-coding exponent of every $\bbm\in\Mset$
simultaneously. ``Random'' refers to the codebook, not the channel:
$\bbW$ is arbitrary, possibly adversarial. \Cref{thm:uc-rc-to-det} replaces
the random codebook by a single deterministic code $\bbC$, paying only the
subexponential test list $\Cset^n$.

\paragraph{Axis 2 (separability).}
Both the merge construction (which builds $\bbU$ from $\Mset$) and the
expurgation step (which builds $\bbC$ from $\Cset$) require their input family
to be subexponential. Many natural families --- continuous parameter spaces,
ISI tap vectors --- are not. \emph{Separability} replaces a continuous family
by a subexponential representative subfamily, in two flavors: a
\emph{type-based} reduction for discrete-structure families
(\cref{sec:dmc}), and a Lipschitz $\eta$-net for continuous parametric
families. The representative subfamily then feeds the merge and the
expurgation, both exponent-tightly. On the metric side the feed is
unconditional: the merge dominates every member of $\Mset$. On the channel
side it certifies $\Cset_d$ directly, extending to the full family where
the law-closeness hypothesis of \cref{lem:uc-transfer} is verified
(\cref{rem:uc-cd-to-c}; both Gaussian families of \cref{sec:awgn}). We
record the hypothesis used in \cref{sec:dmc,sec:erasure}.

\begin{definition}[Separability, (H-sep)]
  \label{def:uc-hsep}
  $\Mset$ has a subfamily $\Md$ satisfying (H-sub) such that, for every
  $\bbm\in\Mset$ and every $n$,
  \[
    \min_{m_d^n\in\Md^n}\PEPm{x^n}{y^n}{m_d^n}\;\le\;s_n\,\PEPm{x^n}{y^n}{m^n}
    \qquad\text{at every }(x^n,y^n),
  \]
  with slack $s_n=e^{o(n)}$ uniform over $\Mset$. (The type-based families of
  \cref{sec:dmc} satisfy this with $s_n$ polynomial or even $s_n=1$.)
\end{definition}

\begin{definition}[Approximate separability, (H-asep)]
  \label{def:uc-hasep}
  As (H-sep), but the pointwise domination is required only off a per-channel
  exceptional set $P^n\subset\cX^n\times\cY^n$ with
  $\PRs{Q_X^n\cdot W^n}{P^n}\le e^{-n(E^\star+\delta)}$ for some $\delta>0$,
  where $E^\star=\sup_{\bbm\in\Mset}\Erc(R;\bbW,\bbQ_X,\bbm)$ is the target
  exponent at the operating rate $R$. The exceptional set then contributes
  an exponent-negligible term to every random-coding bound.
\end{definition}

The two axes are orthogonal: separability is a property of the input family
alone, derandomization of how the codebook is drawn. \Cref{sec:awgn} works
with continuous-alphabet families for which even (H-asep) holds only in a
typical-set-restricted form; there the continuum is handled directly, with a
capped guarantee.

With these notions fixed, \cref{sec:univ-metric} constructs the merge metric
and proves it random-coding universal and asymptotically minimax;
\cref{sec:dmc} instantiates the construction for DMC and finite-state
families; and \cref{sec:erasure} carries it to a decoder with an erasure
option.

%% file: sections/03-universal-metric.tex

\section{The Universal Merge Metric}\label{sec:univ-metric}

This is the core of the paper. Given a subexponential family of decoding
metrics, we build a single metric --- the \emph{merge} --- from their
pairwise error probabilities, and show it is random-coding universal
(\cref{thm:univ-rc}) and asymptotically minimax (\cref{thm:univ-minimax}).
The construction canonicalizes before it merges: each metric is first
replaced by its own PEP --- order-preservingly, so member-by-member decoding
is unchanged --- and the family is then pooled on the common probability
scale (\cref{fig:uc-merge-pipeline}).
The only technical input beyond \cref{sec:framework} is a clipping lemma for
the inverse PEP, which removes the ``normal prior'' condition of the earlier
merge~\cite{elkayam2014universal}: clipping the inverse PEP makes its mean
subexponential for \emph{every} prior, channel, and metric.

Throughout, fix a prior $\bbQ_X$, a channel $\bbW$, and a rate $R>0$; unless
stated otherwise, expectations are under $(X^n,Y^n)\sim Q_X^n\cdot W^n$. A few
lemmas are stated in the \emph{fixed-$y^n$ / product} law (competitor
$\tilde X^n\sim Q_X^n$ drawn independently of $y^n$); the uniformity of
\cref{prop:pep-uniform} lives in that law only --- under the channel-joint law
$Q_X^n\cdot W^n$ the PEP is generally \emph{not} uniform but distributed as
the error spectrum \eqref{eq:recap-spectrum}. Each downstream use of these
lemmas instantiates them at the $y^n$ produced by the outer channel-joint
expectation, so no statement changes truth value.

\subsection{Clipping the inverse PEP}
\label{subsec:uc-clipping}

By the fixed-output uniformity of the PEP (\cref{prop:pep-uniform}), the
randomized PEP is uniform on $[0,1]$ at fixed $y^n$. Hence the reciprocal
$1/p$ at $p\sim\Unif{[0,1]}$ is heavy-tailed with infinite mean. Clipping the
reciprocal at any level $\alpha_n>1$ collapses the divergent tail to a finite
expectation that grows only logarithmically: the harmonic estimate
$\int_{1/\alpha_n}^{1}\mathrm dw/w=\log\alpha_n$ controls it. This is the only
fact about the inverse PEP that the merge construction uses.

\begin{lemma}[Clipped inverse PEP is subexponential]
  \label{lem:uc-clipped-inv-pep}
  For any metric $m^n$, any fixed $y^n\in\cY^n$, and any clip level
  $\alpha_n>1$, with $\tilde X^n\sim Q_X^n$,
  \begin{equation}\label{eq:uc-clipped-inv-pep}
    \Es{\tilde X^n}{\min\BRAs{\bigl(\PEPm{\tilde X^n}{y^n}{m^n}\bigr)^{-1},\,
      \alpha_n}}\;\le\;1+\log\alpha_n .
  \end{equation}
  In particular $\alpha_n=e^{n^2}$ gives a bound of $1+n^2$, subexponential
  in $n$.
\end{lemma}

\begin{IEEEproof}
  For the randomized (dithered) PEP of \cref{sec:recap} computed with the
  metric $m^n$, $\PEP{\tilde X^n}{y^n}\sim\Unif{[0,1]}$ at fixed $y^n$ by
  \cref{prop:pep-uniform}, so writing $P\sim\Unif{[0,1]}$,
  \[
    \E{\min\{P^{-1},\alpha_n\}}
    =\int_0^1\min\{p^{-1},\alpha_n\}\,\mathrm dp
    =\int_0^{1/\alpha_n}\!\alpha_n\,\mathrm dp+\int_{1/\alpha_n}^1\!\frac{\mathrm dp}{p}
    =1+\log\alpha_n .
  \]
  The tie-as-error PEP \eqref{eq:uc-pep-tie} dominates the randomized PEP
  pointwise, so its reciprocal is no larger and the same expectation is
  $\le 1+\log\alpha_n$.
\end{IEEEproof}

\begin{remark}[Removal of the normal-prior condition]
  \label{rem:uc-normal-relaxed}
  \Cref{lem:uc-clipped-inv-pep} holds for \emph{any} prior $Q_X^n$ ---
  product, symmetric, finitely supported, or continuous (e.g.\ uniform on a
  power shell) --- and is the only nontrivial input the merge draws from
  PEP uniformity.
\end{remark}

\begin{remark}[A Kraft inequality for decoding]
  \label{rem:uc-kraft}
  \Cref{lem:uc-clipped-inv-pep} is best read as the decoding analogue of
  the Kraft inequality, in an \emph{information} form. The classical
  inequality lives on a discrete alphabet: integer codeword lengths
  $\ell$ satisfy $\sum_x2^{-\ell(x)}\le1$ precisely when they are, in
  essence, a (sub-)probability assignment --- to compress is to assign a
  probability to each element.

  \emph{To decode is to assign a conditional probability to each codeword
  given the output, relative to the prior.} The budget here is relative to
  an arbitrary
  reference prior $Q_X^n$ --- discrete or continuous --- and conditional on
  the output: writing $u(x)=\min\{-\log\PEPm{x}{y^n}{m^n},\,n^2\}$ for the
  clipped score of \emph{any} metric at a fixed $y^n$,
  \eqref{eq:uc-clipped-inv-pep} reads
  $\Es{\tilde X^n}{e^{u(\tilde X^n)}}\le1+n^2$: up to a subexponential
  normalization, the score is the log-density $\log(dP/dQ_X^n)$ of a
  conditional probability assignment $P(\cdot\mid y^n)$ on the codewords.
  The PEP transform plays the
  role of the Kraft normalization $\ell\mapsto2^{-\ell}/Z$, converting an
  arbitrary metric into its assignment --- order-preservingly, and
  idempotently (the fixed-point property of
  \cref{subsec:uc-merge-def}). For the matched-ML metric the induced
  assignment agrees at exponent scale with the true posterior-to-prior
  ratio $\log\bigl(P(x^n\mid y^n)/Q_X^n(x^n)\bigr)$; the tilt identity
  \eqref{eq:uc-tilt-identity} of \cref{sec:dmc} is its empirical form.
  Universal decoding over a family of metrics thereby becomes universal
  prediction over the induced conditional
  assignments~\cite{merhav1998universal}, and the
  constructions below inherit their coding-theory counterparts: the merge
  is the normalized-maximum-likelihood (envelope) assignment of the
  family, and the minimax normalizer $\gamma_n$ of
  \cref{subsec:uc-minimax} is its Kraft budget.
\end{remark}

\subsection{The merge metric}
\label{subsec:uc-merge-def}

Let $\Md^n=\{m_1^n,\dots,m_{K_n}^n\}$ be a finite family of metrics at
blocklength $n$, with $\Md=\{\Md^n\}_{n\ge1}$ of subexponential growth
($\tfrac1n\log K_n\to0$).

\begin{definition}[PEP-based universal merge metric]
  \label{def:merge-metric}
  For the finite family $\Md^n=\{m_1^n,\dots,m_{K_n}^n\}$ define
  \begin{equation}\label{eq:uc-U-def}
    U^n(x^n,y^n)\;\triangleq\;
    \min\Bigl\{-\log\min_{1\le k\le K_n}\PEPm{x^n}{y^n}{m_k^n},\;\;n^2\Bigr\}.
  \end{equation}
  The sequence $\bbU=\{U^n\}_{n\ge1}$ is the \emph{merge metric} for $\Md$.
\end{definition}

Operationally, $U^n$ is the negative log of the \emph{best-in-class} PEP: a
codeword ranked well by \emph{any} single $m_k$ is ranked well by $U^n$. The
$n^2$ clip floors the merged PEP at $e^{-n^2}$, capping the inverse PEP at
$e^{n^2}$; its role is to keep
the moment $\E{e^{U^n}}$ finite --- the quantity on which the Markov proof
below and, more importantly, the minimax normalizer of
\cref{subsec:uc-minimax} both run. The level $n^2$ is a canonical
representative, not a tuned choice: any clip at $e^{c_n}$ with $c_n=\omega(n)$
and $\tfrac1n\log(1+c_n)\to0$ serves identically
(\cref{rem:uc-clip-calibration} records the calibration). Note $U^n$
depends on $(x^n,y^n)$ only through the PEP values of the family, and in
particular is independent of the rate $R$.

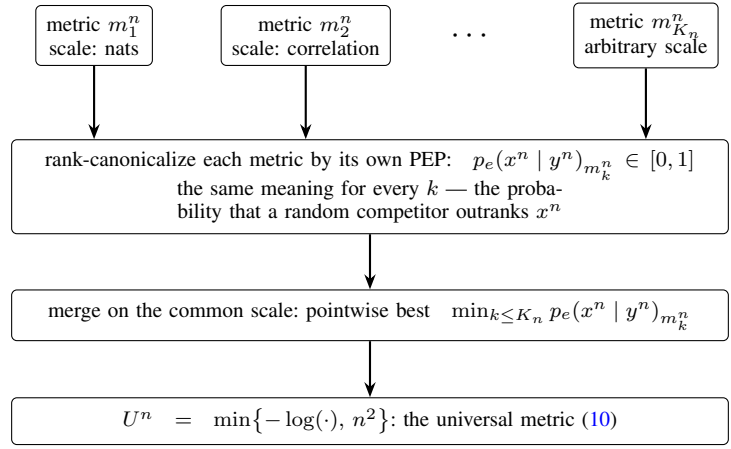
\begin{figure}[tbp]
\centering
\begin{tikzpicture}[
  align=center,
  box/.style={draw, rounded corners=2pt, inner sep=4pt, font=\scriptsize},
  arr/.style={draw, -{Stealth[length=2mm]}, thick}
]
\node[box] (m1) at (-4.4,0) {metric $m^n_1$\\ scale: nats};
\node[box] (m2) at (-1.6,0) {metric $m^n_2$\\ scale: correlation};
\node at (0.6,0) {$\cdots$};
\node[box] (mK) at (2.9,0) {metric $m^n_{K_n}$\\ arbitrary scale};
\node[box, text width=9.2cm] (pep) at (-0.75,-2.0)
  {rank-canonicalize each metric by its own PEP:\quad $\PEPm{x^n}{y^n}{m^n_k}\in[0,1]$\\
   the same meaning for every $k$ --- the probability that a random competitor outranks $x^n$};
\node[box, text width=9.2cm] (mrg) at (-0.75,-3.7)
  {merge on the common scale: pointwise best\quad $\min_{k\le K_n}\PEPm{x^n}{y^n}{m^n_k}$};
\node[box, text width=9.2cm] (uu) at (-0.75,-5.1)
  {$U^n=\min\bigl\{-\log(\cdot),\,n^2\bigr\}$: the universal metric \eqref{eq:uc-U-def}};
\draw[arr] (m1.south) -- (pep.north -| m1.south);
\draw[arr] (m2.south) -- (pep.north -| m2.south);
\draw[arr] (mK.south) -- (pep.north -| mK.south);
\draw[arr] (pep.south) -- (mrg.north);
\draw[arr] (mrg.south) -- (uu.north);
\end{tikzpicture}
\caption{\textbf{Canonicalize, then merge.} The PEP canonicalizes each metric
first (order-preserving, so member-by-member decoding is unchanged), and only
then does the merge compare numbers that mean the same thing for every member.
The construction is scale-free; the clip at $n^2$ is a technical floor
(it keeps the merge moments finite) and costs nothing at exponent level.}
\label{fig:uc-merge-pipeline}
\end{figure}

The merge metric is, up to a subexponential factor, a fixed point of the map
\begin{equation*}
  m^n\;\mapsto\;-\log\PEPm{\cdot}{\cdot}{m^n}:
\end{equation*}
a single $-\log$PEP metric already
equals its own negative-log-PEP (its values are unit-exponential under
$\tilde X^n\sim Q_X^n$, by \cref{prop:pep-uniform}), and the merge is the
clipped pointwise maximum of such fixed points.

In the Kraft reading of
\cref{rem:uc-kraft}, the merge is the upper \emph{envelope} of the family's
induced sub-densities, renormalized --- the normalized-maximum-likelihood
(Shtarkov) assignment of the canonicalized family; its Kraft budget, the
envelope mass $\Es{\tilde X^n}{e^{U^n(\tilde X^n,y^n)}}$, is the quantity
the minimax analysis names $e^{n\gamma_n}$.

The clip is what keeps the
moment finite: without it, $\Es{\tilde X^n}{e^{U^n(\tilde X^n,y^n)}}$ can
diverge (for continuous priors it does) and the Markov step below is vacuous;
\cref{rem:uc-clip-free} records precisely what the clip does and does not
enable.

\begin{lemma}[Pointwise PEP bound for the merge]
  \label{lem:uc-pep-bound-u}
  With $U^n$ from \eqref{eq:uc-U-def}, for every $(x^n,y^n)$ and every
  $j\in\{1,\dots,K_n\}$,
  \begin{equation}\label{eq:uc-pep-u-pointwise}
    \PEPm{x^n}{y^n}{U^n}\;\le\;
    K_n(1+n^2)\,\max\BRAs{\PEPm{x^n}{y^n}{m_j^n},\,e^{-n^2}},
  \end{equation}
  and consequently
  \begin{equation}\label{eq:uc-pep-u-min}
    \PEPm{x^n}{y^n}{U^n}\;\le\;
    K_n(1+n^2)\,\min_k\max\BRAs{\PEPm{x^n}{y^n}{m_k^n},\,e^{-n^2}}.
  \end{equation}
\end{lemma}

\begin{IEEEproof}
  Fix $(x^n,y^n)$. By Markov's inequality applied to $e^{U^n(\cdot,y^n)}$,
  with $\tilde X^n\sim Q_X^n$ the fixed-$y^n$ competitor,
  \[
    \PEPm{x^n}{y^n}{U^n}
    =\PR{e^{U^n(\tilde X^n,y^n)}\ge e^{U^n(x^n,y^n)}}
    \le\Es{\tilde X^n}{e^{U^n(\tilde X^n,y^n)}}\,e^{-U^n(x^n,y^n)} .
  \]
  By \eqref{eq:uc-U-def},
  $e^{-U^n(x^n,y^n)}=\max\{\min_k\PEPm{x^n}{y^n}{m_k^n},\,e^{-n^2}\}$. For the
  normalizer, using $\min(\max_k a_k,b)=\max_k\min(a_k,b)$ for nonnegative
  $a_k,b$ and then the union bound,
  \begin{align*}
    \Es{\tilde X^n}{e^{U^n(\tilde X^n,y^n)}}
    &=\Es{\tilde X^n}{\max_k\min\BRAs{\PEPm{\tilde X^n}{y^n}{m_k^n}^{-1},\,e^{n^2}}}\\
    &\le\sum_{k=1}^{K_n}\Es{\tilde X^n}{\min\BRAs{\PEPm{\tilde X^n}{y^n}{m_k^n}^{-1},\,e^{n^2}}}\\
    &\le K_n(1+n^2),
  \end{align*}
  the last step by \cref{lem:uc-clipped-inv-pep} with $\alpha_n=e^{n^2}$.
  Combining the two displays gives \eqref{eq:uc-pep-u-pointwise};
  \eqref{eq:uc-pep-u-min} follows by minimizing over $j$.
\end{IEEEproof}

Up to the polynomial factor $K_n(1+n^2)$ and the negligible floor $e^{-n^2}$,
the merge is never worse, in pointwise PEP, than any member of $\Md^n$.
Equivalently,
\begin{equation}\label{eq:uc-U-vs-logpep}
  -\log\PEPm{x^n}{y^n}{U^n}\;\ge\;
  U^n(x^n,y^n)-\bigl(\log K_n+\log(1+n^2)\bigr),
\end{equation}
so the merge agrees with its own negative-log-PEP up to the same
subexponential factor that governs the universality cost.

\begin{remark}[Calibration of the clip level]
  \label{rem:uc-clip-calibration}
  $c_n=\omega(n)$ drives the merge proof's residual $K_n(1+c_n)e^{nR-c_n}$
  to zero at every rate (a merely linear clip $c_n=cn$ fails at $R\ge c$)
  and keeps the ceded mass $K_ne^{-c_n}$ of the minimax converse below
  $\tfrac12$; $\tfrac1n\log(1+c_n)\to0$ keeps the moment price $1+c_n$ of
  \cref{lem:uc-clipped-inv-pep} subexponential.
\end{remark}

\begin{remark}[A clip-free proof, and what the clip actually buys]
  \label{rem:uc-clip-free}
  The moment route above is not the only proof of
  \cref{lem:uc-pep-bound-u}, and the clip is not what makes the merge bound
  true. Direct union bound: fix $(x^n,y^n)$, write
  $p^\star\triangleq\min_k\PEPm{x^n}{y^n}{m_k^n}$ and
  $t\triangleq\max\{p^\star,e^{-n^2}\}$. A competitor $\tilde x^n$ satisfies
  $U^n(\tilde x^n,y^n)\ge U^n(x^n,y^n)$ iff
  $\max\{\min_k\PEPm{\tilde x^n}{y^n}{m_k^n},\,e^{-n^2}\}\le t$, which ---
  since $t\ge e^{-n^2}$ --- holds iff
  $\PEPm{\tilde x^n}{y^n}{m_k^n}\le t$ for \emph{some} $k$. For each fixed
  $k$, the tie-vs-dither transfer (\cref{lem:uc-tie-dither}, with $t\le1$)
  gives
  $Q_X^n\{\PEPm{\cdot}{y^n}{m_k^n}\le t\}\le t$; a union bound over the $K_n$
  members yields
  $\PEPm{x^n}{y^n}{U^n}\le K_n\max\{p^\star,e^{-n^2}\}$ ---
  \eqref{eq:uc-pep-u-min} with the strictly better prefactor $K_n$, no
  clipping lemma, no moment control. This is precisely the PEP/prior-mass
  generalization of the merging-of-rankings lemma of Feder and
  Lapidoth~\cite{feder1998universal}: under a uniform prior the PEP \emph{is}
  the normalized rank, and the display above is their union bound over merged
  ranking lists. The moment route is kept because it doubles as the minimax
  machinery of \cref{subsec:uc-minimax}, and \emph{there} the clip is
  load-bearing: $\gamma_n$ is finite only because of it, and the minimax
  converse applies the clipped-moment bound to an \emph{arbitrary} competing
  metric --- a step with no union-bound analogue.
\end{remark}

\subsection{Random-coding universality}
\label{subsec:uc-rc}

\begin{theorem}[Random-coding universality of the merge]
  \label{thm:univ-rc}
  Fix $\bbW$, $\bbQ_X$, $R>0$, and a finite metric family
  $\Md=\{\Md^n\}_{n\ge1}$ of subexponential growth. Let $\bbU=\{U^n\}$ be
  the merge metric
  $U^n(x^n,y^n)=\min\bigl\{-\log\min_{1\le k\le K_n}\PEPm{x^n}{y^n}{m_k^n},\,n^2\bigr\}$
  of \cref{def:merge-metric}. Then
  \begin{equation}\label{eq:uc-rc-univ}
    \Erc(R;\bbW,\bbQ_X,\bbU)\;\ge\;\sup_{\bbm\in\Md}\Erc(R;\bbW,\bbQ_X,\bbm).
  \end{equation}
  Consequently, if $(\Mset,\Cset)$ satisfies (H-sep) or (H-asep) with
  representative $\Md$, the merge metric $\bbU$ satisfies (H-univ($R$)) w.r.t.\
  $(\Mset,\Cset)$; combined with \cref{thm:uc-rc-to-det} (under (H-sub) on
  $\Cset$) it yields a deterministically universal pair $(\bbC,\bbU)$.
\end{theorem}

\begin{IEEEproof}
  Fix $\bbm\in\Md$. By \cref{lem:uc-pep-bound-u} and
  \cref{prop:uc-erc-pep},
  \begin{align*}
    &\E{\min\BRAs{1,e^{nR}\PEPm{X^n}{Y^n}{U^n}}}\\
    &\quad\le\E{\min\BRAs{1,e^{nR}K_n(1+n^2)\max\{\PEPm{X^n}{Y^n}{m^n},e^{-n^2}\}}}\\
    &\quad\le K_n(1+n^2)\,\E{\min\BRAs{1,e^{nR}\PEPm{X^n}{Y^n}{m^n}}}
      +K_n(1+n^2)e^{nR-n^2}.
  \end{align*}
  The second term is super-exponentially small and $K_n(1+n^2)=e^{o(n)}$.
  Taking $-\tfrac1n\log$, the $\liminf$, and using \cref{prop:uc-erc-pep}
  gives $\Erc(R;\bbU)\ge\Erc(R;\bbm)$. Since $\bbm\in\Md$ was arbitrary,
  \eqref{eq:uc-rc-univ} follows. Under (H-sep), chaining
  \eqref{eq:uc-pep-u-min} with the per-$n$ domination
  $\min_{m_d^n}\PEPm{x^n}{y^n}{m_d^n}\le s_n\PEPm{x^n}{y^n}{m^n}$ of
  \cref{def:uc-hsep} gives, at every $(x^n,y^n)$,
  \begin{equation*}
    \PEPm{x^n}{y^n}{U^n}\le K_n(1+n^2)\max\{s_n\PEPm{x^n}{y^n}{m^n},e^{-n^2}\},
  \end{equation*}
  and the same computation with the enlarged --- still subexponential ---
  prefactor $K_n(1+n^2)s_n$ yields (H-univ($R$)). Under (H-asep) the
  domination is available off the exceptional set $P^n$, whose contribution to
  the random-coding error probability is at most
  $\PRs{Q_X^n\cdot W^n}{P^n}\le e^{-n(E^\star+\delta)}$, exponent-negligible
  against every target $\Erc\le E^\star$. The deterministic claim is
  \cref{thm:uc-rc-to-det}.
\end{IEEEproof}

The cost of universality is the factor $K_n(1+n^2)$, i.e.\
$\tfrac1n\log K_n+O(\log n/n)$ in rate --- vanishing for any subexponential
family. This matches the cost in the type-class universal-decoding line of
Goppa~\cite{goppa1975nonprobabilistic}, Csisz\'ar~\cite{csiszar1974extremum},
and Feder--Lapidoth~\cite{feder1998universal} when specialized to the
polynomial DMC family (\cref{sec:dmc}); but the construction here is
channel-structure-free and applies directly to any subexponential family.

The decoder of \cref{thm:univ-rc} is built on the representative set $\Md$.
Separability in fact certifies more: the decoder that merges over the
\emph{entire} family --- holding no representative set at all --- is universal
too, with the representatives appearing only inside the proof.

\begin{proposition}[Merging over the full family]
  \label{prop:uc-continuum-merge}
  Let $\Mset$ satisfy (H-sep) with representative $\Md$, slack
  $s_n\ge1$, and $K_n=|\Md^n|$ ($s_n\ge1$ is a harmless normalization:
  enlarging the slack only weakens the hypothesis). Assume further that a
  \emph{single} countable subfamily $\Mset_0\subset\Mset$ satisfies
  $\inf_{\bbm\in\Mset}\PEPm{x^n}{y^n}{m^n}
   =\inf_{\bbm\in\Mset_0}\PEPm{x^n}{y^n}{m^n}$ at every $(x^n,y^n)$ --- as
  holds for type- and ranking-based families, where the PEP takes finitely
  many values, and for compactly parametrized families with
  parameter-continuous PEP. Define the \emph{continuum merge metric}
  \begin{equation}\label{eq:uc-continuum-merge}
    U_\infty^n(x^n,y^n)\;\triangleq\;
    \min\Bigl\{-\log\inf_{\bbm\in\Mset}\PEPm{x^n}{y^n}{m^n},\;n^2\Bigr\}.
  \end{equation}
  Then $\bbU_\infty$ satisfies (H-univ($R$)) w.r.t.\ $(\Mset,\Cset)$ at every
  rate, with the same subexponential prefactor $s_nK_n(1+n^2)$ as the merge
  over $\Md$.
\end{proposition}

\begin{IEEEproof}
  Measurability of $U_\infty^n$ holds by the countable-subfamily hypothesis.
  \emph{Domination, with zero slack:} for every $\bbm\in\Mset$,
  $e^{-U_\infty^n(x^n,y^n)}
   =\max\{\inf_{\bbm'\in\Mset}\PEPm{x^n}{y^n}{m'^n},e^{-n^2}\}
   \le\max\{\PEPm{x^n}{y^n}{m^n},e^{-n^2}\}$, since the infimum runs over a
  superset of $\{\bbm\}$. \emph{Moment transfer:} the left-hand side of the
  (H-sep) display does not depend on $\bbm$, so taking the infimum over
  $\bbm\in\Mset$ on its right side gives
  $\min_{m_d^n\in\Md^n}\PEPm{x^n}{y^n}{m_d^n}
   \le s_n\inf_{\bbm\in\Mset}\PEPm{x^n}{y^n}{m^n}$ at every $(x^n,y^n)$; with
  $s_n\ge1$ this yields $e^{U_\infty^n}\le s_n\,e^{U^n}$ pointwise for the
  merge metric $U^n$ of \cref{def:merge-metric}, whence, by the expectation
  display in the proof of \cref{lem:uc-pep-bound-u},
  $\Es{X^n\sim Q_X^n}{e^{U_\infty^n(X^n,y^n)}}\le s_nK_n(1+n^2)$ at every
  fixed $y^n$. Markov's inequality then bounds the PEP of $U_\infty^n$ exactly
  as in \cref{lem:uc-pep-bound-u}, and the exponent computation of
  \cref{thm:univ-rc} goes through verbatim with the prefactor
  $s_nK_n(1+n^2)$.
\end{IEEEproof}

\begin{remark}[The representative set is proof scaffolding]
  \label{rem:uc-continuum-merge}
  For the continuum merge, domination over every member is automatic --- the
  infimum runs over a superset --- so separability is consumed \emph{only}
  through the moment bound: the single (H-sep) inequality, read in one
  direction, makes the grid merge dominate the family, and read in the other,
  caps how far the continuum infimum can fall below the grid minimum. In
  the Kraft reading (\cref{rem:uc-kraft}), separability is exactly a
  certificate that the family's per-$y$ envelope mass is subexponential:
  the envelope sees, at each $x$, only the members achieving the maximum
  --- the \emph{extremal rankings} --- and (H-sep) says that
  subexponentially many of them realize it, up to slack, uniformly in $y$.
  The instantiations count precisely these extreme points: joint types
  (\cref{sec:dmc}), faces of a hyperplane arrangement
  (\cref{lem:uc-fsm-polycount}), an $\eta$-net (\cref{sec:awgn}). In the
  instantiations: for finite-state families the slack is $s_n=1$ and
  $\Md\subset\Mset$, so the two decoders coincide; for the DMC families the
  infimum is attained and, under a memoryless prior, the continuum merge
  \emph{is} the tilted-MMI decoder up to $O(\log n)$ (\cref{sec:dmc}); for the
  AWGN families the domination holds only on a typical event whose
  product-law slice (H-asep) does not control, and the continuum merge
  requires a decoder-known typical event and a floor ---
  \cref{sec:awgn}'s GLRT theorem is the typical-set analogue of
  \cref{prop:uc-continuum-merge}, with a capped guarantee as the price of the
  carve-out.
\end{remark}

\subsection{Asymptotic minimax interpretation}
\label{subsec:uc-minimax}

\Cref{thm:univ-rc} is one-sided: the merge is at least as good as any family
member. We now show that, in a minimax sense, no universal decoder does
significantly better. Following the source-coding framework of Xie and
Barron~\cite{xie2000asymptotic}, we measure performance by \emph{regret}
--- the per-letter log-ratio of the universal PEP to the best-in-class PEP.

Work at fixed $n$ and fixed $y\in\cY^n$, with merge metric $u_n=U^n$. For each
$(x,y)$ let $\hat\theta(x,y)\triangleq\argmin_k\PEPm{x}{y}{m_k^n}$ denote the
in-hindsight best metric (ties broken arbitrarily), and write
$m_{\hat\theta}=m_{\hat\theta(x,y)}^n$.

\begin{definition}[Regret, equalizer, minimax]
  \label{def:uc-regret}
  The \emph{regret} of a universal metric $u_n$ at $(x,y)$ is
  \begin{equation}\label{eq:uc-regret}
    r_{y,n}(u_n,x)\;\triangleq\;
    \tfrac1n\log\!\BRA{\frac{\PEPm{x}{y}{u_n}}{\PEPm{x}{y}{m_{\hat\theta}}}} .
  \end{equation}
  The \emph{minimax regret} is
  $\bar r_{y,n}\triangleq\min_{u_n}\max_x r_{y,n}(u_n,x)$. A decoder is
  \emph{minimax} if $\max_x r_{y,n}(u_n,x)=\bar r_{y,n}$, an \emph{equalizer}
  if $r_{y,n}(u_n,x)=\bar r_{y,n}$ for all $x$, and \emph{asymptotically}
  so if these hold up to $o(1)$.
\end{definition}

The regret is non-negative pointwise for any decoder built from $\Md^n$, since
$m_{\hat\theta(x,y)}$ is the best-in-class at $(x,y)$. Define the normalizer
\begin{equation}\label{eq:uc-gamma-n}
  \gamma_n\;\triangleq\;\tfrac1n\log\Es{\tilde X^n\sim Q_X^n}{e^{U^n(\tilde X^n,y)}},
\end{equation}
which is finite for \emph{every} prior thanks to the clip, and which
\cref{lem:uc-clipped-inv-pep} bounds by
$\gamma_n\le\tfrac1n\log K_n+O(\log n/n)=o(1)$. In the Kraft reading
(\cref{rem:uc-kraft}), $e^{n\gamma_n}$ is the \emph{envelope mass}: the
probability budget the merge over-reserves to dominate every member's
assignment simultaneously. It is maximal ($\approx K_n$) when the members
place disjoint bets --- each concentrating its mass on different codewords
--- and collapses to $\approx1$ when they rank alike; the cost of
universality is a disagreement measure of the family's rankings at $y$.

\begin{proposition}[Equalizer bound for the merge]
  \label{prop:uc-equalizer}
  Let $\cF_n\triangleq\{x\in\cX^n:\PEPm{x}{y}{m_{\hat\theta(x,y)}}\ge e^{-n^2}\}$
  be the finite-exponent set. The merge $u_n=U^n$ satisfies
  $r_{y,n}(u_n,x)\le\gamma_n$ for every $x\in\cF_n$, and
  $Q_X^n(\cX^n\setminus\cF_n)\le K_n e^{-n^2}$.
\end{proposition}

\begin{IEEEproof}
  By Markov and \eqref{eq:uc-gamma-n},
  $\PEPm{x}{y}{u_n}\le e^{n\gamma_n}e^{-U^n(x,y)}$. For $x\in\cF_n$ the clip
  is inactive, so $e^{-U^n(x,y)}=\PEPm{x}{y}{m_{\hat\theta(x,y)}}$ exactly;
  dividing and taking $\tfrac1n\log$ gives $r_{y,n}(u_n,x)\le\gamma_n$. For
  the complement: the tie-vs-dither transfer (\cref{lem:uc-tie-dither}) gives
  $Q_X^n\{\PEPm{\cdot}{y}{m_k}<e^{-n^2}\}\le e^{-n^2}$ for each $k$; a union
  bound over $\Md^n$ gives $Q_X^n(\cX^n\setminus\cF_n)\le K_n e^{-n^2}$.
\end{IEEEproof}

The converse requires an averaged regret non-negativity that holds for
\emph{every} metric, in or out of the family. Pointwise regret may be negative
for an out-of-family metric (which can rank a chosen $x$ first), but the
$Q_X^n$-averaged \emph{randomized} PEP at fixed $y$ is $1/2$ (a consequence
of \cref{prop:pep-uniform}), so no metric beats the family's best in the
average.

\begin{lemma}[Averaged regret non-negativity]
  \label{lem:uc-avg-regret}
  Fix $y\in\cY^n$, a finite family $\Md^n$ with $K_n$ members, and a constant
  $\beta\ge0$. If a metric $u$ satisfies
  $\PEPm{x}{y}{u}\le e^{-n\beta}\PEPm{x}{y}{m_{\hat\theta(x,y)}}$ for all
  $x\in\cF_n$, then
  \begin{equation}\label{eq:uc-avg-regret}
    \beta\;\le\;\tfrac1n\log\frac{1}{1/2-K_n e^{-n^2}}\;=\;O(1/n)
  \end{equation}
  whenever $1/2-K_n e^{-n^2}>0$; the subexponential growth (H-sub) ensures
  this for all large $n$ and drives the bound to $\tfrac{\log2}{n}+o(1/n)$.
  In particular, no metric can have strictly negative regret of constant
  order simultaneously at every $x\in\cF_n$.
\end{lemma}

\begin{IEEEproof}[Proof sketch]
  By \cref{prop:pep-uniform} the averaged randomized PEP at fixed $y$ is
  $1/2$; the tie-as-error PEP dominates it, giving a lower bound
  $\tfrac12-K_n e^{-n^2}$ on the $\cF_n$-restricted averaged PEP. The
  hypothesis caps the same average by $e^{-n\beta}$. Pinching the two yields
  \eqref{eq:uc-avg-regret}. The full proof is in
  Appendix~\ref{app:avg-regret}.
\end{IEEEproof}

\begin{theorem}[Asymptotic minimaxity of the merge]
  \label{thm:univ-minimax}
  Assume $\Md$ has subexponential growth (H-sub), and let
  $\gamma_n^{\cF_n}\triangleq\tfrac1n\log\Es{\tilde X^n}{e^{U^n(\tilde X^n,y)}\Ind{\tilde X^n\in\cF_n}}$.
  Then $\gamma_n^{\cF_n}\le\gamma_n\in[0,\tfrac1n\log K_n+O(\log n/n)]$, and on
  the finite-exponent set $\cF_n$ the merge $u_n=U^n$ is asymptotically
  minimax:
  \begin{itemize}
  \item \emph{(achievability)} $\max_{x\in\cF_n}r_{y,n}(u_n,x)\le\gamma_n$;
  \item \emph{(converse)} for every $\epsilon>0$ and all large $n$, no metric
    $u_n^\star$ satisfies $r_{y,n}(u_n^\star,x)\le\gamma_n^{\cF_n}-\epsilon$
    for all $x\in\cF_n$.
  \end{itemize}
\end{theorem}

The restriction to $\cF_n$ is intrinsic: the $n^2$-clip deliberately cedes
the super-exponentially-rare set $\cX^n\setminus\cF_n$ where the
best-in-class metric already exceeds every finite exponent.

\begin{IEEEproof}[Proof sketch]
  Achievability is \cref{prop:uc-equalizer}. Converse: a two-case dichotomy on
  $\gamma_n^{\cF_n}$ versus $\epsilon$. If $\gamma_n^{\cF_n}<\epsilon$, the
  regret hypothesis matches \cref{lem:uc-avg-regret}, which pinches $\beta$ to
  $O(1/n)$ and contradicts a fixed $\epsilon$. If $\gamma_n^{\cF_n}\ge\epsilon$,
  Markov on $e^{U^n}$ plus reciprocation reduces the hypothesis to a
  clipped-inverse-PEP expectation that \cref{lem:uc-clipped-inv-pep} caps at
  $1+n^2$, while the hypothesis forces it to be $\ge e^{n\epsilon}$ ---
  contradiction for large $n$. The full proof is in
  Appendix~\ref{app:univ-minimax}. In the Kraft reading both cases are a
  mass budget: any competing metric is \emph{itself} Kraft-admissible
  (\cref{lem:uc-clipped-inv-pep} holds for every metric), so beating the
  envelope by $e^{n\epsilon}$ throughout $\cF_n$ would require total mass
  $e^{n\epsilon}$ against a budget of $1+n^2$ (Case B); Case A is the
  mean-$\tfrac12$ budget of the PEP account --- one cannot assign more
  probability than one has.
\end{IEEEproof}

\begin{remark}[Where the two normalizers live]
  \label{rem:uc-minimax-normalizer}
  The unrestricted normalizer satisfies
  $\gamma_n\in[0,\allowbreak\,\tfrac1n\log K_n+O(\log n/n)]$: the upper end is the
  clipped-inverse-PEP bound (\cref{lem:uc-clipped-inv-pep}), the lower end
  follows from $e^{U^n}\ge1$. The restricted normalizer $\gamma_n^{\cF_n}$
  lies in the same interval up to a super-exponentially small slack at the
  lower endpoint: the indicator truncates an $e^{U^n}\ge1$ integrand, so the
  lower bound is $\tfrac1n\log Q_X^n(\cF_n)\ge
  \tfrac1n\log(1-K_ne^{-n^2})=-O(e^{-n^2}/n)$ rather than strictly $0$.
\end{remark}

\begin{remark}[The scale of the minimax claim, and a fine-scale converse]
  \label{rem:uc-minimax-scale}
  \emph{What $o(1)$-minimaxity asserts.}
  At the $o(1)$ scale of \cref{def:uc-regret} the $\cF_n$-restricted minimax
  regret is itself degenerate: it tends to $0$, sandwiched between
  $\gamma_n=o(1)$ from above (\cref{rem:uc-minimax-gap}) and $-O(1/n)$ from
  below (\cref{lem:uc-avg-regret}). Every decoder with vanishing worst-case
  regret on $\cF_n$ is thus ``asymptotically minimax'' in the $o(1)$ sense;
  the content of \cref{thm:univ-minimax} lies in its converse clause, which
  pins the value at $\gamma_n^{\cF_n}$ against constant margins --- weaker
  than the source-coding template of Xie and
  Barron~\cite{xie2000asymptotic}, where the minimax value is matched at its
  natural $(\log n)/n$ scale, sharp constant included.

  \emph{The fine-scale converse.}
  The proof, however, delivers a converse at that finer scale, which we
  record: \emph{for every $\delta>0$ and all $n$ large, every metric $u$
  satisfies}
  \[
    \max_{x\in\cF_n}r_{y,n}(u,x)\;\ge\;\gamma_n^{\cF_n}\,-\,(2+\delta)\,
    \frac{\log n}{n}.
  \]
  The derivation is margin bookkeeping over the machinery already in place ---
  the Case-B reciprocate-and-average argument with a vanishing margin
  $\epsilon_n$ in place of the constant $\epsilon$, plus a
  clipped-moment/averaged-regret argument covering negative worst-case
  regret --- and is given self-contained in
  Appendix~\ref{app:minimax-scale}. Paired with the achievability
  $\max_{x\in\cF_n}r_{y,n}(U^n,x)\le\gamma_n$, this sandwiches the minimax
  regret on $\cF_n$ between $\gamma_n^{\cF_n}-O(\log n/n)$ and $\gamma_n$: on
  the finite-exponent set, no universal decoder --- the merge included --- can
  push its worst-case regret more than $O(\log n/n)$ below
  $\gamma_n^{\cF_n}$.

  \emph{What is not claimed.} The sandwich matches at the $(\log n)/n$ scale
  only up to the gap between the two normalizers
  ($\gamma_n^{\cF_n}\le\gamma_n$, bracketed individually in
  \cref{rem:uc-minimax-normalizer} but with their difference unbounded below
  the trivial $\tfrac1n\log K_n+O(\log n/n)$); pinning the minimax regret at
  the $(\log n)/n$ scale with sharp constants, Xie--Barron style, remains
  open.

  \emph{Game vs.\ operational.}
  A further limitation is directional rather than scalar: both clauses of
  \cref{thm:univ-minimax} are statements about the fixed-$y$ regret
  \emph{game}, not about the channel-averaged, operational cost of
  universality. The asymmetry is elementary. Achievability is a uniform
  upper bound, and uniform bounds survive averaging:
  $r_{y,n}(U^n,x)\le\gamma_n$ at every $(x,y)$ with $x\in\cF_n$ integrates,
  under any channel-joint law, into the global guarantee of
  \cref{thm:univ-rc} (the ceded set carries mass at most $K_ne^{-n^2}$,
  exponent-neutral). The
  converse lower-bounds a \emph{maximum} over $x$ at each $y$, and a lower
  bound on a maximum does not survive averaging: a channel routes
  probability mass and need not visit the adversarial $(x,y)$ pairs, so no
  operational statement follows. The operational minimax price of
  universality --- the $\inf_u\sup_{\bbW\in\Cset}$ gap between achieved and
  benchmark performance --- therefore lies somewhere in $[0,\gamma_n]$, and
  the minimax--maximin gap of \cref{rem:uc-minimax-gap} suggests it can sit
  strictly below $\gamma_n$. In the language of universal
  prediction~\cite{merhav1998universal},
  \cref{thm:univ-minimax} is the Shtarkov, individual-sequence half of the
  theory~\cite{shtarkov1987universal}; the redundancy--capacity half ---
  the operational converse --- is open (\cref{sec:conc}).
\end{remark}

\begin{remark}[Minimax--maximin gap]
  \label{rem:uc-minimax-gap}
  Restricting the regret game to reverse channels $W(\cdot\mid y)$ supported
  on $\cF_n$, weak duality and the equalizer bound give
  $\sup_W\inf_{u_n}\sum_xW(x\mid y)\,r_{y,n}(u_n,x)
  \le\inf_{u_n}\sup_W\sum_xW(x\mid y)\,r_{y,n}(u_n,x)\le\gamma_n$
  (substitute the merge and apply \cref{prop:uc-equalizer}). In the
  source-coding setting~\cite{xie2000asymptotic} the minimax and
  maximin regrets coincide. In channel decoding they need not: the merge is an
  asymptotic equalizer only as an \emph{upper} bound on the regret, and on
  specific inputs it can outperform $\gamma_n$ substantially. The worst-case
  regret saturates at $\gamma_n$ while the average under a specific reverse
  channel may fall strictly below it --- a structural asymmetry reflecting the
  directional nature of decoding.
\end{remark}

\begin{example}[DMC type-based metrics]
  \label{ex:uc-dmc-regret}
  For a DMC family with metrics depending only on joint types, the family size
  is polynomial, $K_n=O(n^{|\cX||\cY|-1})$, so
  $\gamma_n\le\tfrac{|\cX||\cY|-1}{n}\log n+O(\log n/n)$, vanishing at rate
  $\log n/n$. The merge metric --- which \cref{sec:dmc} shows reduces to MMI
  --- has worst-case regret on $\cF_n$ bounded by the same rate.
\end{example}

\begin{remark}[Competitive-minimax reading, and a Shtarkov identity]
  \label{rem:uc-fm-competitive}
  In the competitive-minimax program of Feder and
  Merhav~\cite{feder2002competitive}, universality over a channel family
  $\{W_\theta\}$ is measured by the largest $\xi\in[0,1]$ --- the
  universally achievable fraction of the matched benchmark exponent ---
  keeping the minimax ratio
  $\min_{\varphi}\max_\theta P_e(\varphi;\theta)/[P_e^\star(\theta)]^{\xi}$
  subexponential; $\xi^\star$ can be strictly smaller than $1$, and Akirav
  and Merhav~\cite{akirav2007competitive} establish $\xi=1$ at the
  \emph{random-coding} benchmark for several ensembles. The present
  results give the general form of that phenomenon: for any decoupled pair
  $(\Mset,\Cset)$ under (H-sep) (Feder--Merhav's pairing is the tied
  special case $\Mset=\{\log W_\theta\}$, $\Cset=\{W_\theta\}$), with
  \[
    \Theta_n(R)\;\triangleq\;\inf_{u^n}\;
    \sup_{\bbW\in\Cset,\;\bbm\in\Mset}
    \frac{\E{P_e(C^n,W^n,u^n)}}{\E{P_e(C^n,W^n,m^n)}}
  \]
  (supremum over benchmarks with $\E{P_e(C^n,W^n,m^n)}\ge e^{-n^2/2}$,
  say), \cref{lem:uc-pep-bound-u} gives
  $\Theta_n(R)\le s_nK_n(1+n^2)(1+o(1))=e^{o(n)}$: $\xi=1$ by a single
  rate-free decoder, for every prior, with no benchmark exponents in the
  rule --- the merge normalizes each member by the exact conditional law of
  its own score at the observed $y$ (\cref{prop:pep-uniform}), where the
  Feder--Merhav weight $e^{n\xi E^\star(\theta)}$ is a global,
  rate-dependent constant. The $\xi^\star<1$ phenomenon lives precisely
  where (H-sep) fails. Finally, in the Kraft reading of \cref{rem:uc-kraft}
  this minimax analysis is Shtarkov's theory transplanted: the merge is the
  normalized-maximum-likelihood rule of the canonicalized family,
  $\gamma_n$ its complexity (the log envelope mass), and
  \cref{thm:univ-minimax} a decoding analogue of Shtarkov's minimax-regret
  theorem~\cite{shtarkov1987universal}, cf.~\cite{xie2000asymptotic}.
\end{remark}

%% file: sections/04-dmc.tex

\section{Discrete Memoryless and Finite-State Channels}\label{sec:dmc}

We instantiate the merge construction of \cref{sec:univ-metric} on two
families with discrete structure: the discrete memoryless channels (DMCs) and
the finite-state metrics. Both admit a direct polynomial-cardinality
reduction through the method of types~\cite{csiszar1998method} --- no
Lipschitz/covering argument is needed. As specializations the merge recovers the classical maximum-mutual-information
(MMI) decoder of Goppa~\cite{goppa1975nonprobabilistic},
Csisz\'ar~\cite{csiszar1974extremum}, and
Csisz\'ar--K\"orner~\cite{csiszar1981information}
(\cref{thm:uc-mmi}), the channel-universal Feder--Lapidoth decoder~\cite{feder1998universal}
(\cref{cor:uc-fl}), and the finite-state/Markov generalizations
of~\cite{elkayam2014universal} (cf.\ \cref{rem:uc-normal-relaxed}).

\paragraph{Setup.}
Fix finite alphabets $\cX,\cY$. The DMC family is
$\{\bbW:W^n=W^{\otimes n},\,W\in\cP(\cY\mid\cX)\}$, with matched-ML metric
\begin{equation}\label{eq:uc-dmc-metric}
  m_W(x^n,y^n)=\log W^n(y^n\mid x^n)
  =n\sum_{a\in\cX,\,b\in\cY}T(a,b)\,\log W(b\mid a),
\end{equation}
which depends on $(x^n,y^n)$ only through the joint type $T=T_{x^n y^n}$. We
write $T_X,T_Y$ for its marginals and recall the empirical mutual information
\begin{equation}\label{eq:uc-emp-mi}
  \Ihat(T)\;\triangleq\;H(T_X)+H(T_Y)-H(T)
  \;=\;\sum_{a,b}T(a,b)\log\frac{T(a,b)}{T_X(a)T_Y(b)} .
\end{equation}

\subsection{Reduction to types}
\label{subsec:uc-dmc-types}

The matched-ML family $\{m_W:W\in\cP(\cY\mid\cX)\}$ is a continuum, but its
value \eqref{eq:uc-dmc-metric} is a linear functional of $\log W$ applied to
$T$. Since length-$n$ joint types form a polynomial-cardinality set, the family
has only polynomially many distinct rankings of joint types --- and the merge
sees a metric only through such rankings.

\begin{lemma}[Polynomial type-based family]
  \label{lem:uc-dmc-polytypes}
  The number of joint types over $\cX\times\cY$ is
  $\binom{n+|\cX||\cY|-1}{|\cX||\cY|-1}=O(n^{|\cX||\cY|-1})$. For each joint
  type $T$, let $W_T\in\cP(\cY\mid\cX)$ be its empirical channel,
  $W_T(b\mid a)=T(a,b)/T_X(a)$ when $T_X(a)>0$ (arbitrary otherwise). The
  finite family $\Md^n\triangleq\{m_{W_T}:T\in\cT_n(\cX\times\cY)\}$ has
  cardinality $K_n=O(n^{|\cX||\cY|-1})$.
\end{lemma}

\begin{IEEEproof}
  Standard type counting~\cite{csiszar1981information}.
\end{IEEEproof}

Applying \cref{def:merge-metric} to $\Md^n$ gives a universal metric with merge
penalty $\tfrac1n\log K_n=O(\log n/n)$. Different $W$ inducing the same ranking
of joint types are equivalent for the merge; the polynomial cardinality is a
feature of the type structure, not of any continuity of the family.

\subsection{Uniform prior on a type class: recovery of MMI}
\label{subsec:uc-dmc-mmi}

Fix an input type $T_X^\star\in\cT_n(\cX)$ and take the prior
$Q_X^n=\Unif{\cT(T_X^\star)}$, so every codeword has type exactly $T_X^\star$.

\begin{theorem}[The merge metric reduces to MMI]
  \label{thm:uc-mmi}
  For $Q_X^n=\Unif{\cT(T_X^\star)}$ and the type-based family $\Md^n$ of
  \cref{lem:uc-dmc-polytypes}, the merge metric $U^n$ of
  \cref{def:merge-metric} satisfies
  \begin{equation}\label{eq:uc-mmi}
    U^n(x^n,y^n)\;=\;n\,\Ihat(T_{x^n y^n})+o(n)
  \end{equation}
  uniformly in $(x^n,y^n)\in\cT(T_X^\star)\times\cY^n$. Hence the merge decoder
  $\hat\imath(y^n)=\argmax_i U^n(x_i^n,y^n)$ is exponent-equivalent to the MMI
  decoder $\hat\imath_{\MMI}(y^n)=\argmax_i\Ihat(T_{x_i^n y^n})$: an additive
  $o(n)$ shift of a joint-type-based metric moves the type-counting threshold
  by $o(n)$, so the two PEPs agree within $e^{\pm o(n)}$ at every $(x^n,y^n)$
  and every random-coding exponent coincides.
\end{theorem}

\begin{IEEEproof}
  Fix $(x^n,y^n)\in\cT(T_X^\star)\times\cY^n$ with joint type
  $T=T_{x^n y^n}$; the $n^2$ clip is inactive for the exponent claim.

  \emph{Term $T'=T$.} The candidate $\tilde x^n$ contributes to
  $\PEPm{x^n}{y^n}{m_{W_T}}$ iff $m_{W_T}(\tilde x^n,y^n)\ge m_{W_T}(x^n,y^n)$,
  a linear inequality in $T_{\tilde x^n y^n}$. Type counting~\cite[Lemma~2.6]{csiszar1981information}
  gives $\PEPm{x^n}{y^n}{m_{W_T}}=e^{-n\Ihat(T)+o(n)}$: the lower bound from the
  single type $T'=T$ (sequences of joint type $T$ tie under $m_{W_T}$ and are
  charged in full by the tie-as-error rule), the upper bound from the
  conditional-type estimate $|\cT(T_{X\mid Y}\mid y^n)|\le e^{nH(T_{X\mid Y}\mid T_Y)}$
  together with the splitting-identity step in the proof of
  \cref{lem:uc-typemass-dom}, which shows that every joint type contributing
  to the PEP event has empirical information at least $\Ihat(T)$.

  \emph{Minimum over $T'$.} Taking $T'=T$ gives
  $\min_{T'}\PEPm{x^n}{y^n}{m_{W_{T'}}}\le e^{-n\Ihat(T)+o(n)}$. For the matching
  lower bound, $m_{W_{T'}}$ depends on $(\tilde x^n,y^n)$ only through the joint
  type \eqref{eq:uc-dmc-metric}, so every $\tilde x^n\in\cT(T_X^\star)$ of joint
  type $T$ with $y^n$ ties with $x^n$ under \emph{every} $T'$ and is counted in
  the PEP. These contribute mass
  $|\cT(T\mid y^n)|/|\cT(T_X^\star)|=e^{-n\Ihat(T)+o(n)}$, using
  $H(T_X^\star)-H(T_{X\mid Y})=\Ihat(T)$. Hence
  $\min_{T'}\PEPm{x^n}{y^n}{m_{W_{T'}}}\ge e^{-n\Ihat(T)+o(n)}$, and
  $U^n=-\log\min_{T'}\PEPm{x^n}{y^n}{m_{W_{T'}}}=n\Ihat(T)+o(n)$.
\end{IEEEproof}

\begin{lemma}[Type-mass domination]
  \label{lem:uc-typemass-dom}
  Let $Q_X^n=\Unif{\cT(T_X^\star)}$. For every metric $m^n$ that depends on
  $(x^n,y^n)$ only through
  the joint type $T_{x^ny^n}$ --- in particular the matched-ML metric $m_W$ of
  every DMC $W$, by \eqref{eq:uc-dmc-metric} --- at every
  $(x^n,y^n)\in\cT(T_X^\star)\times\cY^n$,
  \begin{equation}\label{eq:uc-mmi-domination}
    \min_{T'\in\cT_n(\cX\times\cY)}\PEPm{x^n}{y^n}{m_{W_{T'}}}
    \;\le\;(n+1)^{2|\cX||\cY|+|\cX|}\,\PEPm{x^n}{y^n}{m^n},
  \end{equation}
  the uniform-on-type-class analogue of the domination clause of
  \cref{thm:uc-tilted}(ii) below.
\end{lemma}

\begin{IEEEproof}[Proof sketch]
  The tie mass at the observed joint type $T$ lower-bounds every
  joint-type-based PEP by the type-class mass $e^{-n\Ihat(T)}$, up to
  polynomial factors; the family member $m_{W_T}$ matches this bound, because
  on its PEP event the splitting identity
  $\Ihat(\tilde T)=\langle\tilde T,g_T\rangle+D(\tilde T\|T)$ forces
  $\Ihat(\tilde T)\ge\Ihat(T)$. Type counting supplies the polynomial
  factors; the computation is in Appendix~\ref{app:dmc-typemass}.
\end{IEEEproof}

\begin{corollary}[MMI is universal for fixed-type input]
  \label{cor:uc-mmi-universal}
  For every DMC $W\in\cP(\cY\mid\cX)$ and every input type $T_X^\star$, with
  $Q_X^n=\Unif{\cT(T_X^\star)}$ the MMI decoder achieves the matched-ML
  random-coding exponent at every rate $R$:
  $\Erc(R;\bbW,Q_X^n,U^n)\ge\Erc(R;\bbW,Q_X^n,m_W)$.
\end{corollary}

\begin{IEEEproof}[Proof sketch]
  Step 1 substitutes the domination \eqref{eq:uc-mmi-domination} into the
  merge bound of \cref{lem:uc-pep-bound-u} and repeats the exponent
  computation of \cref{thm:univ-rc} with the enlarged --- still
  subexponential --- prefactor; neither $\Md^n$ nor $U^n$ depends on $W$, so
  a single decoder serves every DMC simultaneously. Step 2 transfers the
  guarantee to the MMI decoder proper: by \cref{thm:uc-mmi} the two metrics
  differ by an exponent-neutral $o(n)$ shift, so their PEPs agree within
  $e^{\pm o(n)}$ pointwise. No step of the classical Csisz\'ar--K\"orner
  analysis is imported. Full proof in Appendix~\ref{app:dmc-typemass}.
\end{IEEEproof}

\begin{remark}[Recovery of Goppa's MMI]
  \label{rem:uc-goppa}
  \Cref{cor:uc-mmi-universal} is Goppa's universal decoder~\cite{goppa1975nonprobabilistic},
  recovered here as a direct corollary of the merge theorem on the type-based
  family. The recovery is self-contained: the guarantee rests on the merge
  bound and type counting alone (the domination clause
  \eqref{eq:uc-mmi-domination}, mirroring the tilted clause of
  \cref{thm:uc-tilted}(ii) below), with the matched-exponent identity of
  Csisz\'ar~\cite{csiszar1974extremum} and
  Csisz\'ar--K\"orner~\cite{csiszar1981information}
  entering only to name the common value.
\end{remark}

\subsection{Memoryless input prior: tilted MMI}
\label{subsec:uc-dmc-memoryless}

For a memoryless prior $Q_X^n=Q_X^{\otimes n}$, the codeword type $T_{X^n}$ is
random (multinomial, mean $Q_X$). The prior decomposes over input types as
$Q_X^{\otimes n}=\sum_{T_X}\pi_n(T_X)\Unif{\cT(T_X)}$ with
$\pi_n(T_X)=e^{-nD(T_X\|Q_X)+o(n)}$; we assume without loss of generality
$Q_X(a)>0$ for every $a\in\cX$. The natural decoder tilts MMI by the
input-type divergence,
\begin{equation}\label{eq:uc-tilted-mmi}
  U^n_{\mathrm{tilt}}(x^n,y^n)\;\triangleq\;
  n\bigl[\Ihat(T_{x^n y^n})+D(T_{x^n}\|Q_X)\bigr],
\end{equation}
the \emph{$Q_X$-tilted MMI} metric. The direction of the tilt is dictated by
the elementary identity
\begin{equation}\label{eq:uc-tilt-identity}
  \Ihat(T)+D(T_X\|Q_X)\;=\;D\bigl(T\,\big\|\,Q_X\otimes T_Y\bigr),
\end{equation}
valid for every joint type $T$: the metric scores $(x^n,y^n)$ by how
improbable it is for an independent competitor $\tilde X^n\sim Q_X^{\otimes n}$
to reproduce the observed joint type. An atypical input type makes the joint
type \emph{harder} to reach by chance, raising the score rather than lowering
it. Equivalently, $D(T\|Q_X\otimes T_Y)$ is the empirical rate of the
posterior-to-prior ratio $\log\bigl(P(x^n\mid y^n)/Q_X^n(x^n)\bigr)$ ---
the matched instance of the Kraft dictionary of \cref{rem:uc-kraft}. We derive \eqref{eq:uc-tilted-mmi} from the merge construction and prove
its universality.

\paragraph{The tilted type-based family.}
For each joint type $T'\in\cT_n(\cX\times\cY)$ define the \emph{tilted type
metric}
\begin{equation}\label{eq:uc-tilted-member}
  \tilde m_{T'}(x^n,y^n)\;\triangleq\;
  n\sum_{a\in\cX,\,b\in\cY}T_{x^ny^n}(a,b)\,
  \log\frac{T'(a,b)}{Q_X(a)\,T'_Y(b)},
\end{equation}
with the convention $\log0=-\infty$, and let
$\tilde\Md^n\triangleq\{\tilde m_{T'}:T'\in\cT_n(\cX\times\cY)\}$, of
cardinality $\tilde K_n\le(n+1)^{|\cX||\cY|}$. Decomposing
$\log\tfrac{T'(a,b)}{Q_X(a)T'_Y(b)}
 =\log W_{T'}(b\mid a)+\log\tfrac{T'_X(a)}{Q_X(a)}-\log T'_Y(b)$, each member
is the matched-ML metric of the empirical channel $W_{T'}$ of
\cref{lem:uc-dmc-polytypes}, shifted per input letter by the Stirling weight
$\log(T'_X(a)/Q_X(a))$ of its own input type (the per-letter reading of the
multinomial estimate above); the last term depends on $y^n$ alone and never
affects a ranking.

\begin{theorem}[Merge metric reduces to tilted MMI]
  \label{thm:uc-tilted}
  For the memoryless prior $Q_X^n=Q_X^{\otimes n}$ with $Q_X(a)>0$ for every
  $a\in\cX$, and the tilted type-based
  family $\tilde\Md^n$, the merge metric $U^n$ of \cref{def:merge-metric}
  satisfies:
  \begin{enumerate}
  \item[(i)] uniformly in $(x^n,y^n)\in\cX^n\times\cY^n$,
    $U^n(x^n,y^n)=n[\Ihat(T_{x^ny^n})+D(T_{x^n}\|Q_X)]+O(\log n)
     =U^n_{\mathrm{tilt}}(x^n,y^n)+O(\log n)$;
  \item[(ii)] \emph{(pointwise domination)} for every metric $m^n$ that
    depends on $(x^n,y^n)$ only through the joint type $T_{x^ny^n}$ --- in
    particular the matched-ML metric $m_W$ of every DMC $W$ ---
    \[
      \min_{\tilde m\in\tilde\Md^n}\PEPm{x^n}{y^n}{\tilde m}
      \;\le\;(n+1)^{2|\cX||\cY|}\,\PEPm{x^n}{y^n}{m^n}
      \qquad\text{at every }(x^n,y^n).
    \]
  \end{enumerate}
\end{theorem}

\begin{IEEEproof}[Proof sketch]
  Under the memoryless prior the mass of a conditional type class carries the
  tilted weight: $P_n(\tilde T)=e^{-nD(\tilde T\|Q_X\otimes T_Y)}$ up to
  polynomial factors, by \eqref{eq:uc-tilt-identity}. With this sandwich the
  two halves of the proof of \cref{lem:uc-typemass-dom} apply verbatim with
  $g_T=\log\tfrac{T}{Q_X\otimes T_Y}$: the tie mass gives the lower bound,
  the member $\tilde m_T$ matches it via the splitting identity, and the
  uniform bound $D(T\|Q_X\otimes T_Y)\le\log|\cX|+\log(1/\min_aQ_X(a))$
  keeps the $n^2$ clip inactive. Full proof in
  Appendix~\ref{app:dmc-typemass}.
\end{IEEEproof}

\begin{corollary}[Tilted MMI is universal for memoryless input]
  \label{cor:uc-tilted-universal}
  For every DMC $W\in\cP(\cY\mid\cX)$ and memoryless input
  $Q_X^n=Q_X^{\otimes n}$ with full-support $Q_X$, the merge metric $U^n$ over $\tilde\Md^n$ ---
  equal to $U^n_{\mathrm{tilt}}$ up to $O(\log n)$ by
  \cref{thm:uc-tilted}(i) --- achieves the matched-ML random-coding exponent
  at every rate $R$:
  $\Erc(R;\bbW,Q_X^n,U^n)\ge\Erc(R;\bbW,Q_X^n,m_W)$.
\end{corollary}

\begin{IEEEproof}
  The matched-ML metric $m_W$ depends on $(x^n,y^n)$ only through the joint
  type \eqref{eq:uc-dmc-metric}, so \cref{thm:uc-tilted}(ii) applies to it.
  Substituting the domination bound into \eqref{eq:uc-pep-u-min} of
  \cref{lem:uc-pep-bound-u} gives, at every $(x^n,y^n)$,
  $\PEPm{x^n}{y^n}{U^n}\le\tilde K_n(1+n^2)(n+1)^{2|\cX||\cY|}
   \max\{\PEPm{x^n}{y^n}{m_W},e^{-n^2}\}$, and repeating the proof of
  \cref{thm:univ-rc} with the enlarged --- still subexponential --- prefactor
  yields the exponent inequality. Neither $\tilde\Md^n$ nor $U^n$ depends on
  $W$, so a single decoder serves every DMC simultaneously.
\end{IEEEproof}

\begin{remark}[Relation to Merhav's metric-class universality; sign of the tilt]
  \label{rem:uc-tilted-sign}
  Merhav~\cite{merhav:universal} constructs a universal decoder relative to a
  given \emph{class of decoding metrics}: the decoder scores $(x^n,y^n)$ by
  $-\log$ of the prior mass of the set of inputs metric-equivalent to $x^n$
  under every metric in the class. For the class of additive metrics over
  finite alphabets under an i.i.d.\ prior, that score becomes
  $n[\Ihat(T)+D(T_X\|Q_X)]+o(n)$ --- the tilted-MMI metric
  \eqref{eq:uc-tilted-mmi} --- and his Theorem~1 delivers the matched
  random-coding exponent for every DMC. \Cref{thm:uc-tilted} recovers the same
  metric and guarantee within the merge framework: the merge scores
  $(x^n,y^n)$ by the best-in-family PEP, and the PEP of the best tilted member
  is, up to polynomial factors, exactly the prior mass of the observed
  joint-type class.

  The sign of the tilt deserves emphasis. Under the memoryless prior,
  atypical input types are exponentially unlikely,
  $\pi_n(T_X)=e^{-nD(T_X\|Q_X)+o(n)}$, and one might be tempted to
  \emph{downweight} atypical codewords by this multinomial weight, i.e.\ to
  use $n[\Ihat-D]$. The merge construction forces the opposite sign: the
  metric measures the improbability that an independent competitor reproduces
  the observed joint type, and the rarity of the input type makes that event
  harder, not easier. The $-D$-tilted rule is a genuinely different
  decoder --- the two agree within a fixed input type but rank codewords of
  different input types differently --- and no matched-exponent guarantee
  for it appears in the literature we know. For \emph{uniform} $Q_X$, identity
  \eqref{eq:uc-tilt-identity} collapses the metric to
  $n[\log|\cX|-H_T(X\mid Y)]$: tilted MMI is the minimum conditional
  empirical-entropy decoder, coinciding with MMI on constant-composition
  codebooks; for non-uniform memoryless input the tilted and plain MMI
  decoders genuinely differ on atypical codewords, and tilted MMI is the form
  that respects the prior.
\end{remark}

\subsection{Channel universality and Feder--Lapidoth}
\label{subsec:uc-dmc-fl}

The merge over the type-based family is a \emph{single} decoder, independent of
the channel parameter $W$. This is the channel-universal setting of Feder and
Lapidoth~\cite{feder1998universal}, recovered as an immediate corollary.

\begin{corollary}[Channel-universal MMI / Feder--Lapidoth]
  \label{cor:uc-fl}
  Let $\Cset_{\mathrm{DMC}}=\{\bbW:W\in\cP(\cY\mid\cX)\}$ be the family of all
  DMCs over $\cX,\cY$. For $Q_X^n=\Unif{\cT(T_X^\star)}$ the MMI decoder is a
  single decoder achieving, for every $W\in\cP(\cY\mid\cX)$, the matched-ML
  random-coding exponent. For memoryless input $Q_X^n=Q_X^{\otimes n}$ the
  same is realized by the tilted-MMI decoder \eqref{eq:uc-tilted-mmi}
  (\cref{cor:uc-tilted-universal}; cf.\ Merhav~\cite{merhav:universal} and
  \cref{rem:uc-tilted-sign}).
\end{corollary}

\begin{IEEEproof}
  Immediate from \cref{cor:uc-mmi-universal,cor:uc-tilted-universal}: the
  families $\Md^n$ and $\tilde\Md^n$ are independent of $W$, so a single
  decoder serves every $W\in\cP(\cY\mid\cX)$.
\end{IEEEproof}

\begin{remark}[Comparison with Feder--Lapidoth]
  \label{rem:uc-fl}
  Feder--Lapidoth~\cite{feder1998universal} reach the same conclusion via a
  finite-grid approximation (strong separability) of the DMC family; the
  merge route substitutes the polynomial cardinality of joint types
  (\cref{lem:uc-dmc-polytypes}) for the grid, and the clipping lemma for the
  normal-prior moment control of~\cite{elkayam2014universal}.
  Merhav~\cite{merhav:universal} treats universality over a class of
  decoding metrics --- closer to our formulation --- and the merge realizes
  that framework whenever the family admits a polynomial-cardinality
  reduction.
\end{remark}

\subsection{Finite-state and \texorpdfstring{$k$}{k}-th order Markov metrics}
\label{subsec:uc-dmc-fsc}

The type-based reduction extends, with one extra layer, to metrics computed by
a finite-state machine on $(x^n,y^n)$ --- the setting of Ziv's universal
decoder~\cite{ziv1985universal} and of~\cite{elkayam2014universal}. A
finite-state metric (FSM) is specified by a next-state map
$g:\cX\times\cY\times\cS_n\to\cS_n$, an output map
$q:\cX\times\cY\times\cS_n\to\mR$, and an initial state $s_0$; with
$s_i=g(x_i,y_i,s_{i-1})$,
\begin{equation}\label{eq:uc-fsm-metric}
  m_{(g,q,s_0)}(x^n,y^n)\;\triangleq\;\sum_{i=1}^n q(s_{i-1},x_i,y_i).
\end{equation}
Let $\cM_{\mathrm{FSM}}^n$ denote all such metrics at blocklength $n$.

\begin{lemma}[Subexponential representative subfamily for FSM]
  \label{lem:uc-fsm-polycount}
  There is a subfamily $\Md^n\subset\cM_{\mathrm{FSM}}^n$ such that for every
  $m\in\cM_{\mathrm{FSM}}^n$, $\min_{m'\in\Md^n}\PEPm{x^n}{y^n}{m'}\le\PEPm{x^n}{y^n}{m}$
  at every $(x^n,y^n)$, with
  $|\Md^n|\le|\cS_n|^{D+1}(n+1)^{4D^2+2D}$ where
  $D\triangleq|\cS_n||\cX||\cY|$. If $|\cS_n|=n^\alpha$ with
  $\alpha<\tfrac12$, then $\tfrac1n\log|\Md^n|\to0$.
\end{lemma}

\begin{IEEEproof}[Proof sketch]
  The mechanism is a ranking collapse. Fix $g,s_0$: there are
  $\le|\cS_n|^{D}$ next-state maps and $\le|\cS_n|$ initial states. Given $g$,
  the metric \eqref{eq:uc-fsm-metric} is a linear functional
  $\langle q,nT^g_{x^n y^n}\rangle$ of the empirical triple type
  $T^g\in\mR^D$, and the tie-as-error PEP depends on $q$ only through the
  \emph{ranking} (total preorder, ties included) it induces on the
  $\le(n+1)^D$ achievable triple types --- a \emph{face} of the arrangement
  of the $\le(n+1)^{2D}$ difference hyperplanes in $\mR^D$. Applying
  Zaslavsky's region count~\cite{buck1943partition,zaslavsky1975facing} flat
  by flat bounds the number of faces by $(n+1)^{4D^2+2D}$; one representative
  $q$ per face has PEP \emph{equal} to every metric in its face at every
  $(x^n,y^n)$, giving the pointwise-domination clause. Summing rates,
  $\tfrac1n\log|\Md^n|=O(n^{2\alpha-1}\log n)\to0$ for $\alpha<\tfrac12$. The
  flat-by-flat count is given in Appendix~\ref{app:fsm-polycount}; a sharper
  construction attaining
  $\alpha<1$ is given in~\cite[Sec.~V.C]{elkayam2014universal}.
\end{IEEEproof}

\begin{corollary}[Universal decoding over FSM and Markov metrics]
  \label{cor:uc-fsm}
  For $|\cS_n|=n^\alpha$, $\alpha<\tfrac12$, let $U_{\mathrm{FSM}}^n$ be the
  merge over $\Md^n$ of \cref{lem:uc-fsm-polycount}. Then for every
  $m\in\cM_{\mathrm{FSM}}^n$ and every channel $W$,
  $\Erc(R;\bbW,Q_X^n,U_{\mathrm{FSM}}^n)\ge\Erc(R;\bbW,Q_X^n,m)$. In
  particular, a $k$-th order Markov metric is an FSM metric whose state holds
  the last $k$ symbols of $(x,y)$, $|\cS_n|=(|\cX||\cY|)^{k_n}$, so the merge
  over the $k$-th order Markov family is universal whenever, for some fixed
  $\delta>0$,
  \begin{equation}\label{eq:uc-markov-order}
    k_n\;\le\;\bigl(\tfrac12-\delta\bigr)\frac{\log n}{\log(|\cX||\cY|)},
  \end{equation}
  recovering~\cite[Sec.~V.D]{elkayam2014universal}.
\end{corollary}

\begin{IEEEproof}
  \Cref{thm:univ-rc} gives $\Erc(R;U_{\mathrm{FSM}}^n)\ge\Erc(R;m')$ for every
  representative $m'\in\Md^n$, and the pointwise PEP-domination of
  \cref{lem:uc-fsm-polycount} supplies, for each $m$, a representative with no
  larger PEP, whence no smaller exponent. The Markov case follows since
  \eqref{eq:uc-markov-order} gives $|\cS_n|\le n^{1/2-\delta}$.
\end{IEEEproof}

\begin{remark}[The boundary $\alpha=1$ is genuine]
  \label{rem:uc-fsm-boundary}
  At $\alpha=1$ ($|\cS_n|=n$) the family is too rich: an adversarial next-state
  map can designate one ``best'' codeword regardless of $y^n$, so no decoder is
  universal~\cite[Ex.~V.C]{elkayam2014universal}. With the $\alpha<1$
  achievability of~\cite{elkayam2014universal}, the boundary is tight; the
  self-contained merge of \cref{lem:uc-fsm-polycount} reaches $\alpha<\tfrac12$.
\end{remark}

\paragraph{Finite-state \emph{channels}.}
For \emph{deterministic} next-state maps the state trajectory is
reconstructible from $(x,y)$, so the matched log-likelihood is itself a
triple-type FSM metric at fixed $|\cS|$, and \cref{cor:uc-fsm} already attains
its random-coding exponent --- recovering Ziv~\cite{ziv1985universal} inside
this machinery. For \emph{hidden stochastic} state the likelihood sums over
state trajectories and is no longer of FSM form; there, Lapidoth and
Ziv~\cite{lapidoth1998universality} show that the conditional Lempel--Ziv
decoder attains the matched-ML random-coding exponent. Whether the
Markov-merge family also suffices for the hidden-state case --- via
state-restart accounting at window boundaries, a subexponential cost once
$k_n\to\infty$ --- we leave open.

%% file: sections/05-erasure.tex

\section{Universal Decoding with an Erasure Option}\label{sec:erasure}

We extend the universal merge of \cref{sec:univ-metric} to decoders with an
erasure option. The setting builds on Forney's erasure / undetected-error
exponents~\cite{forney1968exponential} and the type-class treatment of
universal erasure decoders by Csisz\'ar--K\"orner~\cite{csiszar1981information},
refined by Merhav~\cite{merhav2008error}; universal erasure decoding itself
has been studied by Merhav and Feder~\cite{merhav2007minimax} in the
competitive-minimax framework and by Moulin~\cite{moulin2009neyman} via
Neyman--Pearson classes. The main message here is structural:
\emph{once a metric is random-coding universal, the same metric is automatically
universal for both the erasure and the undetected-error exponents} when paired
with a natural PEP-based erasure rule, at no exponent cost. The exponents
delivered are the Gallager-style weakened forms of Forney's --- the pair
$\Erc(R+\gamma)$ and $\gamma+\Erc(R)$; Forney's optimally tuned exponents
$E_1,E_2$~\cite{forney1968exponential} are in general strictly larger and are
not recovered --- consistent with the impossibility result
of~\cite{huleihel2016erasure}, which shows that the erasure/list
random-coding exponents are not universally achievable in general. As before, all statements are at the random-coding level;
deterministic universality follows by a derandomization adapted to the margin
continuum (\cref{thm:uc-det-erasure}).

\subsection{Erasure preliminaries}
\label{subsec:uc-erasure-prelim}

\begin{definition}[Decoder with erasure option]
  \label{def:uc-decoder-erasure}
  For a blocklength-$n$ code of size $M_n$, a decoder with erasure option is a
  map $\cD:\cY^n\to\{1,\dots,M_n\}\cup\{\bot\}$, where $\bot$ denotes erasure.
  Its erasure and undetected-error probabilities are
  \begin{align}
    P_{\mathrm{era}}(C^n,W^n,\cD) &\triangleq\PR{\cD(Y^n)=\bot},
      \label{eq:uc-p-era}\\
    P_{\mathrm{und}}(C^n,W^n,\cD) &\triangleq\PR{\cD(Y^n)\ne I,\,\cD(Y^n)\ne\bot},
      \label{eq:uc-p-und}
  \end{align}
  where $I\sim\Unif{\{1,\dots,M_n\}}$ is the sent index, $X^n=C^n(I)$, and
  $Y^n\sim W^n(\cdot\mid X^n)$.
\end{definition}

\begin{definition}[PEP-based erasure decoder]
  \label{def:uc-pep-erasure}
  (The two-threshold rule of the companion
  paper~\cite[Def.~21]{elkayampep2}, on the block scale.)
  Given a metric $\bbm$, codebook $C^n=\{x_1^n,\dots,x_{M_n}^n\}$, confidence
  margin $\gamma\ge0$ (a free operating parameter, not to be confused with the
  merge normalizer $\gamma_n$ of \cref{sec:univ-metric}), and output $y^n$, let
  $i^\star=\argmin_i\PEPm{x_i^n}{y^n}{m^n}$. The decoder erases ($=\bot$) if
  \begin{equation}\label{eq:uc-erase-cond}
    \begin{aligned}
      &-\log\PEPm{x_{i^\star}^n}{y^n}{m^n}<n(R+\gamma)
      \qquad\text{or}\\
      &\max_{j\ne i^\star}\bigl[-\log\PEPm{x_j^n}{y^n}{m^n}\bigr]+n\gamma\\
      &\qquad>-\log\PEPm{x_{i^\star}^n}{y^n}{m^n},
    \end{aligned}
  \end{equation}
  and outputs $i^\star$ otherwise.
\end{definition}

The rule erases whenever the best candidate's PEP is not small enough or is not
separated by margin $\gamma$ from all others; $\gamma$ trades erasure against
undetected error. With $\cD_n$ this decoder and $\bullet\in\{\mathrm{era},\mathrm{und}\}$,
write $E_\bullet(R,\gamma;\bbW,\bbC,\bbm)\triangleq\Ehat(P_\bullet(C^n,W^n,\cD_n))$
(deterministic) and $E_{\bullet,\mathrm{rc}}(R,\gamma;\bbW,\bbQ_X,\bbm)$ for the
random-coding version (with $C^n$ i.i.d.\ $\sim Q_X^n$), in parallel with
\cref{def:uc-error-exponents}.

The analysis rests on a one-shot bound for the PEP-based erasure rule under
an \emph{arbitrary} metric --- the two-threshold bound of the companion
paper~\cite[Th.~22]{elkayampep2}, restated here on the block scale. With the
dithered PEP of \cref{sec:recap}, write
$\tilde P_e(R;Q_X^n,m^n)\triangleq
 \E{\min\{1,e^{nR}\,\PEP{X^n}{Y^n}\}}$ for the clipped-union
($\RCUp$) functional of the metric $m^n$ under $(X^n,Y^n)\sim Q_X^n\cdot W^n$.

\begin{lemma}[One-shot erasure and undetected-error bounds]
  \label{lem:uc-era-oneshot}
  Fix a blocklength $n$, a metric $m^n$, a prior $Q_X^n$, a rate $R>0$, and a
  margin $\gamma\ge0$. For the rate-$R$ i.i.d.\ ensemble decoded with the
  erasure rule of \cref{def:uc-pep-erasure} applied to the \emph{dithered}
  PEP scores (thresholds $n(R+\gamma)$ and margin $n\gamma$ on the per-block
  scale),
  \begin{align}
    \E{P_{\mathrm{era}}(C^n,W^n,\cD_n)}
      &\le\tilde P_e(R+\gamma;Q_X^n,m^n), \label{eq:uc-era-oneshot}\\
    \E{P_{\mathrm{und}}(C^n,W^n,\cD_n)}
      &\le e^{-n\gamma}\,\tilde P_e(R;Q_X^n,m^n). \label{eq:uc-und-oneshot}
  \end{align}
\end{lemma}

\begin{IEEEproof}[Proof sketch]
  This is \cite[Th.~22]{elkayampep2} read on the block scale. The mechanism:
  condition on the transmitted index and write $Z$ for its dithered score.
  Given the transmitted pair, the $M_n-1$ competitor scores
  $C_i\triangleq-\log\PEPU{X_i^n}{Y^n}{U_i}$ are conditionally i.i.d.\ with
  the \emph{exact} exponential tail $\PR{C_i\ge t}=e^{-t}$, by the
  fixed-output uniformity of \cref{prop:pep-uniform}. An erasure requires
  $Z$ to miss the confidence bar $n(R+\gamma)$ or a competitor within the
  margin $n\gamma$ --- a union bound that assembles to
  $\E{\min\{1,e^{n(R+\gamma)}e^{-Z}\}}=\tilde P_e(R+\gamma)$. An undetected
  error requires a competitor clearing the \emph{larger} of the confidence
  bar and the margin bar $Z+n\gamma$; since
  $e^{nR}e^{-\max\{nR,Z\}}=\min\{1,e^{nR}e^{-Z}\}$ (the split is at $Z=nR$,
  not $n(R+\gamma)$), the union bound collapses to
  $e^{-n\gamma}\min\{1,e^{nR}e^{-Z}\}$, whose expectation is
  \eqref{eq:uc-und-oneshot}. This is why the undetected-error bound reads
  the spectrum at the operating rate $R$, suppressed by $e^{-n\gamma}$,
  while the erasure bound reads it at the boosted rate $R+\gamma$.
\end{IEEEproof}

\begin{theorem}[Erasure exponents via the random-coding exponent]
  \label{thm:uc-erasure-via-rc}
  For any $\bbW$, $\bbQ_X$, $\bbm$, rate $R>0$, and margin $\gamma\ge0$,
  \begin{align}
    E_{\mathrm{era},\mathrm{rc}}(R,\gamma;\bbW,\bbQ_X,\bbm)
      &\ge\Erc(R+\gamma;\bbW,\bbQ_X,\bbm), \label{eq:uc-era-via-rc}\\
    E_{\mathrm{und},\mathrm{rc}}(R,\gamma;\bbW,\bbQ_X,\bbm)
      &\ge\gamma+\Erc(R;\bbW,\bbQ_X,\bbm). \label{eq:uc-und-via-rc}
  \end{align}
\end{theorem}

\begin{IEEEproof}
  \Cref{lem:uc-era-oneshot} bounds the two failure modes for an arbitrary
  metric, and its bounds hold for the tie-as-error PEP used throughout this
  paper: the proof controls each competing codeword through the tail
  estimate of \cref{lem:uc-tie-dither}, which holds for the tie-as-error
  form, while the transmitted-word spectrum is by definition the
  tie-as-error spectrum that the right-hand $\Erc$ is built from; every
  union-bound step therefore carries through. Taking $-\tfrac1n\log$, the
  $\liminf$, using \cref{prop:uc-erc-pep} and
  $\Ehat(e^{-n\gamma}a_n)=\gamma+\Ehat(a_n)$, yields
  \eqref{eq:uc-era-via-rc}--\eqref{eq:uc-und-via-rc}.
\end{IEEEproof}

\subsection{Random-coding universal erasure decoding}
\label{subsec:uc-erasure-main}

\begin{corollary}[Random-coding universal erasure decoder]
  \label{cor:uc-rc-univ-erasure}
  Fix $\gamma_{\max}\ge0$ and suppose $\bbU$ satisfies (H-univ($R'$)) for every
  $R'\in[R,R+\gamma_{\max}]$.
  Then for every $\gamma\in[0,\gamma_{\max}]$, $\bbW\in\Cset$, and $\bbm\in\Mset$,
  \begin{align}
    E_{\mathrm{era},\mathrm{rc}}(R,\gamma;\bbW,\bbQ_X,\bbU)
      &\ge\Erc(R+\gamma;\bbW,\bbQ_X,\bbm), \label{eq:uc-univ-era}\\
    E_{\mathrm{und},\mathrm{rc}}(R,\gamma;\bbW,\bbQ_X,\bbU)
      &\ge\gamma+\Erc(R;\bbW,\bbQ_X,\bbm). \label{eq:uc-univ-und}
  \end{align}
\end{corollary}

\begin{IEEEproof}
  Chain \cref{thm:uc-erasure-via-rc} applied to $\bbU$ with (H-univ($R'$))
  at the two rates $R+\gamma$ and $R$.
\end{IEEEproof}

\begin{remark}[Why the merge metric qualifies]
  \label{rem:uc-erasure-u-via-pep}
  \Cref{thm:uc-erasure-via-rc} applies to $\bbU$ because $U^n$ is itself a
  metric: its randomized PEP is fixed-$y^n$ uniform (\cref{prop:pep-uniform}),
  so the competitor-tail estimate holds for $U^n$ verbatim. The only
  $\bbU$-specific input is the universality of $\Erc(\cdot;\bbU)$ on
  $[R,R+\gamma_{\max}]$, supplied by the merge theorem (\cref{thm:univ-rc});
  the merge metric is rate-independent, so a \emph{single} $\bbU$ serves every
  operating point on the erasure--undetected-error tradeoff. No correspondence
  between the metric value $U^n$ and $-\log\PEPm{x^n}{y^n}{U^n}$ is needed,
  though \eqref{eq:uc-U-vs-logpep} provides one up to the merge cost when the
  PEP-based thresholds of \cref{def:uc-pep-erasure} are read against the
  metric value.
\end{remark}

\begin{corollary}[Universal erasure decoder for separable families]
  \label{cor:uc-erasure-separable}
  If $(\Mset,\Cset)$ satisfies (H-sep) or (H-asep) with representative $\Md$,
  let $\bbU$ be the merge metric of \cref{def:merge-metric} for $\Md$. Since
  $\bbU$ is rate-independent, \cref{thm:univ-rc} supplies (H-univ($R'$)) for
  every $R'$, and \cref{cor:uc-rc-univ-erasure} gives
  \eqref{eq:uc-univ-era}--\eqref{eq:uc-univ-und} for every $\bbW\in\Cset$,
  $\bbm\in\Mset$, and $\gamma\in[0,\gamma_{\max}]$.
\end{corollary}

\begin{IEEEproof}
  Under (H-sep) the
  hypothesis is rate-free and available at every boosted rate verbatim.
  Under (H-asep) one point needs checking: the exceptional-set condition of
  \cref{def:uc-hasep} is calibrated to the target exponent at the operating
  rate. Rate-monotonicity supplies it. $\Erc(R';\cdot)$ is
  non-increasing in $R'$ (in \cref{prop:uc-erc-pep} the rate enters only
  through the increasing factor $e^{nR'}$), so the target at any boosted rate
  satisfies $E^\star(R')\le E^\star(R)$. Neither the representative subfamily
  nor the exceptional sets depend on the rate, so the sets and slack
  witnessing (H-asep) at the base rate witness it at every
  $R'\in[R,R+\gamma_{\max}]$: the approximate-separability slack
  is rate-uniform.
\end{IEEEproof}

The hypotheses of \cref{cor:uc-erasure-separable} are supplied by the families
of this paper as follows.
\begin{itemize}
  \item \emph{Type-based metric families} (DMCs and finite-state channels,
  \cref{sec:dmc}): (H-sep) holds, with no caveat.
  \item \emph{Linear-in-$x$ Gaussian metric families} (discretized in
  \cref{sec:awgn}): the typical-set-restricted form of \cref{thm:awgn-disc}
  yields (H-asep) only when the typicality rate $\delta$ exceeds the target
  exponent $E^\star(R')$ across $[R,R+\gamma_{\max}]$; otherwise the erasure
  guarantees carry the $\min\{\cdot,\delta\}$-capped exponents of
  \cref{thm:awgn-glrt}.
  \item \emph{Channel-side separability} (the Gaussian ISI and
  finite-dimensional interference families of \cref{sec:awgn}): these enter
  not through the hypothesis of \cref{cor:uc-erasure-separable}, which is
  metric-side, but through the derandomization of
  \cref{subsec:uc-erasure-det}, where they supply the subexponential channel
  test family $\Cset_d$.
\end{itemize}
The $\Cset_d\!\to\!\Cset$ extension of \cref{rem:uc-cd-to-c} applies to the
erasure setting as well: the change-of-measure bound of \cref{lem:uc-transfer}
holds verbatim for any decoding-region event, hence for the erasure and
undetected-error probabilities separately. We do not restate the verification;
absent it, the deterministic erasure guarantees below are read over $\Cset_d$.

\subsection{Deterministic universal erasure decoding}
\label{subsec:uc-erasure-det}

\begin{theorem}[Deterministic universal erasure decoder]
  \label{thm:uc-det-erasure}
  Suppose $\Cset$ satisfies (H-sub) and $\bbU$ satisfies (H-univ($R'$)) on
  $[R,R+\gamma_{\max}]$. Then there is a code sequence $\bbC$ of rate $R$ such
  that for every $\gamma\in[0,\gamma_{\max}]$, $\bbW\in\Cset$, $\bbm\in\Mset$,
  \begin{align}
    \Eera(R,\gamma;\bbW,\bbC,\bbU) &\ge\Erc(R+\gamma;\bbW,\bbQ_X,\bbm),
      \label{eq:uc-det-era}\\
    \Eund(R,\gamma;\bbW,\bbC,\bbU) &\ge\gamma+\Erc(R;\bbW,\bbQ_X,\bbm).
      \label{eq:uc-det-und}
  \end{align}
\end{theorem}

\begin{IEEEproof}[Proof sketch]
  Expurgate over a polynomial margin grid, then pass to the continuum of
  margins by monotonicity. A grid $G_n$ of spacing $1/n$ has $O(n)$ points.
  Markov plus a union bound over the $2|\Cset^n||G_n|$ constraints (channel,
  grid margin, erasure or undetected event; the expurgated decoder is $\bbU$
  and does not depend on $\bbm$, so $\Mset$ enters no constraint) leaves a
  positive fraction of codebooks within a subexponential factor of every
  random-coding bound
  simultaneously. Raising $\gamma$ enlarges the erase region and shrinks the
  accept region, so between grid points the two probabilities move
  monotonically; the $1/n$ mismatch costs a factor $e$ per block,
  exponent-neutral. The full five-step proof is in
  Appendix~\ref{app:det-erasure}.
\end{IEEEproof}

\begin{corollary}[Deterministic PEP-based universal erasure decoder]
  \label{cor:uc-det-erasure}
  Under the assumptions of \cref{cor:uc-erasure-separable} with (H-sub) also on
  $\Cset$, for any fixed $\gamma_{\max}\ge0$ there is a code sequence $\bbC$ of
  rate $R$ such that $(\bbC,\bbU)$, with $\bbU$ the merge metric, is a
  deterministically universal erasure decoder: \eqref{eq:uc-det-era}--\eqref{eq:uc-det-und}
  hold for all $\gamma\in[0,\gamma_{\max}]$, $\bbW\in\Cset$, $\bbm\in\Mset$.
\end{corollary}

\begin{IEEEproof}
  \Cref{cor:uc-erasure-separable} supplies (H-univ($R'$)) on
  $[R,R+\gamma_{\max}]$; with (H-sub) on $\Cset$,
  \cref{thm:uc-det-erasure} applies.
\end{IEEEproof}

%% file: sections/06-awgn.tex

\providecommand{\mS}{{\mathbb S}}        
\providecommand{\SNR}{\mathrm{SNR}}

\section{Continuous Alphabets: Separable Families and the AWGN Channel}
\label{sec:awgn}

The merge construction of \cref{sec:univ-metric} takes a finite family of
decoding metrics and produces a single metric that attains, up to a vanishing
penalty, the random-coding exponent of every member. For channels indexed by a
\emph{continuum} of unknown parameters --- the situation in most analog
settings --- a finite family is not given a priori. This section closes that gap
for continuous-alphabet AWGN channels. The bridge from the parameter continuum
to the merge theorem is a \emph{discretization lemma}
(\cref{thm:awgn-disc}): for the class of metrics that are \emph{linear in the
codeword} on the power shell, a Lipschitz analysis of the angular-correlation
spectrum collapses a compact $d$-dimensional parameter set into a
polynomial-cardinality grid that is exponent-equivalent on a typical set. The
merge of that grid then yields a universal continuous generalized-likelihood
decoder (\cref{thm:awgn-glrt}). We illustrate with two applications. For AWGN
with deterministic interference, the construction attains the matched-ML
exponent uniformly in the interference parameter up to the typical-set cap of
\cref{thm:awgn-glrt} (the \emph{$\delta$-cap}). For the Gaussian ISI channel,
it attains, under the same cap, the best exponent inside a fixed
equalize-and-decode family --- in general strictly below matched-ML.

The class boundary is drawn by computability, and deliberately so. On the
power shell the PEP of a linear-in-$x$ metric is a closed-form function of a
single cosine (\cref{eq:awgn-pep-via-f}), so the universal decoder the merge
produces is \emph{explicit} --- the normalized-correlation GLRT
\eqref{eq:awgn-glrt-rule}, implementable as it stands. This is what the two
applications showcase: the PEP framework turning an abstract universality
guarantee into a simple, computable rule. For the quadratic matched-ML ISI
metric the same framework marks the path --- a level-smoothness statement
for the shell prior in place of the closed form --- but there the PEP is a
genuine large-deviations object, and we record that step as an open lemma
rather than develop it (\cref{rem:awgn-isi-open}).

Throughout, the input distribution is the uniform measure on the power shell,
\begin{equation*}
  Q_X^n = \Unif{\sqrt{nP}\,\mS^{n-1}}
\end{equation*}
--- the canonical sphere ensemble going
back to Shannon~\cite{shannon1959probability} --- and the PEP framework and
its uniformity are those of the companion paper~\cite{elkayampep1}
(\cref{prop:pep-uniform}), read at the exponent scale through
\cref{prop:exponent}.

\subsection{Separable families (recall)}
\label{sec:awgn-sep}

The separability machinery is \cref{def:uc-hsep,def:uc-hasep} of
\cref{sec:framework}: a subexponential representative subfamily whose members
dominate the PEP of the full family, exactly ((H-sep)) or off an
exponent-negligible exceptional set ((H-asep)). Domination is required at the
level of \emph{PEPs}, not of the metric values --- and the implication
``metric-Lipschitz $\Rightarrow$ PEP-Lipschitz'' needs its own argument.
Establishing it for the linear-in-$x$ metric class on the AWGN power shell,
via the cosine chain $\theta\mapsto v_\theta(y)\mapsto t_\theta(x,y)\mapsto
\log F(t)$, is the technical content of this section
(\cref{fig:uc-lip-chain}).

Finite-alphabet DMC families and finite-state metrics are separable by the
method of types (\cref{sec:dmc}), with representative cardinality polynomial
in $n$. The Gaussian ISI family is channel-side separable, in the
exceptional-set sense, by Feder and Lapidoth~\cite{feder1998universal}: a
fixed-dimension tap net plus Gaussian smoothness of the log-likelihood give
likelihood closeness off an exponentially rare output event. The remainder of this section gives a
quantitative, geometry-driven route for the \emph{linear-in-$x$} metric class
that covers both AWGN applications; for these families the domination holds
only on a typical event, so the guarantees below carry a cap.

\subsection{Linear-in-\texorpdfstring{$x$}{x} metrics for AWGN}
\label{sec:awgn-linear}

We work with metrics that are linear in the codeword,
\begin{equation}\label{eq:awgn-linear-metric}
  m_\theta(x,y) \;=\; \BRAi{x,\,v_\theta(y)},
  \qquad \theta\in\Theta\subset\mR^d,
\end{equation}
on the power shell $\norm{x}^2=nP$ under uniform-shell input. On the shell three
reductions place a wide range of AWGN decoders into this form:
(i) even powers $\norm{x}^{2k}$ are constant and drop;
(ii) a positive prefactor $\alpha>0$ does not change the decision rule, so the
noise variance $\sigma^2$ and any input scale $h>0$ vanish from the
rank-equivalent metric;
(iii) any linear preprocessing $G$ of $y$ yields $\BRAi{x,Gy}$, covering
equalize-and-decode metrics.

\paragraph{Interference (matched-ML lands in the class).}
For $Y=hX+S(\theta)+Z$ with $Z\sim\cN(0,\sigma^2 I_n)$, expanding
$-\tfrac{1}{2\sigma^2}\norm{y-hx-S(\theta)}^2$ on the shell gives a term constant
in $x$ plus $\tfrac{h}{\sigma^2}\BRAi{x,\,y-S(\theta)}$. Dropping the positive
prefactor reduces matched-ML to~\eqref{eq:awgn-linear-metric} with
$v_\theta(y)=y-S(\theta)$, parametrized by $\theta$ alone; the pair $(h,\sigma)$
survives only inside regularity (typical-set) conditions.

\paragraph{ISI (matched-ML leaves the class).}
For $Y=H_h X+Z$ with $H_h$ the Toeplitz convolution operator,
$-\norm{y-H_h x}^2 = -\norm{y}^2 + 2\BRAi{x,\,H_h^\top y} - \norm{H_h x}^2$, and
the quadratic $\norm{H_h x}^2$ is not constant on the shell unless $H_h$ is
unitary. Matched-ML is therefore outside~\eqref{eq:awgn-linear-metric}. What
\emph{is} in the class is the equalize-and-decode (E\&NN) family: fix an
equalizer map $h\mapsto G_h$ and use $m_h(x,y)=\BRAi{x,\,G_h y}$.

\subsection{The discretization lemma}
\label{sec:awgn-disc}

Fix a metric $m(x,y)=\BRAi{x,v(y)}$ with $v(y)\neq 0$ and set the cosine
\begin{equation}\label{eq:awgn-cosine}
  t(x,y)\;\triangleq\;\frac{\BRAi{x,v(y)}}{\norm{x}\,\norm{v(y)}}.
\end{equation}
Under shell input, rotational invariance makes the PEP a function of the cosine
alone. Writing $\hat\rho = \BRAi{X^n,u}/(\norm{X^n}\norm{u})$ for the angular
correlation between a uniform shell vector and a fixed direction $u$ (its law is
the same for every $u$ by symmetry), and $F(t)\triangleq\PR{\hat\rho\ge t}$, one
has the closed form
\begin{equation}\label{eq:awgn-pep-via-f}
  \PEP{x}{y}_{m}
  \;=\; \PRs{\tilde X^n\sim Q_X^n}{\BRAi{\tilde X^n,v(y)}\ge\BRAi{x,v(y)}}
  \;=\; F\!\bigl(t(x,y)\bigr),
\end{equation}
where, with $\beta=(n-1)/2$,
\begin{equation}\label{eq:awgn-f-form}
  F(t)\;=\;\frac{1}{\pi}\,(1-t^2)^{\beta}\,I_n(t),
  \qquad
  I_n(t)\;\triangleq\;\int_0^{\pi/2}\!\bigl(1+t^2\tan^2\phi\bigr)^{-\beta}\,d\phi .
\end{equation}
The identity~\eqref{eq:awgn-pep-via-f} reduces comparing the PEP across $\theta$
to comparing $F$ at the cosines $t_\theta(x,y)$; the closed form is proved in
Appendix~\ref{app:shell-f} (\cref{lem:awgn-shell-f}). The discretization
result rests on two Lipschitz estimates --- one for $\log F$ in the cosine,
one for the cosine in the parameter --- combined through a covering argument.
\Cref{fig:uc-lip-chain} previews the chain and its constants.

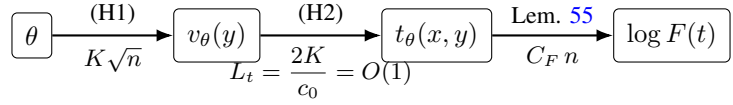
\begin{figure}[tbp]
\centering
\begin{tikzpicture}[node distance=8mm and 16mm,
	bx/.style={draw, rounded corners=2pt, inner sep=5pt, font=\small},
	arr/.style={-{Latex[length=2.2mm]}, thick}]
\node[bx] (th) {$\theta$};
\node[bx, right=of th] (v) {$v_\theta(y)$};
\node[bx, right=of v] (t) {$t_\theta(x,y)$};
\node[bx, right=of t] (F) {$\log F(t)$};
\draw[arr] (th) -- node[above, font=\footnotesize]{(H1)} node[below, font=\footnotesize]{$K\sqrt n$} (v);
\draw[arr] (v) -- node[above, font=\footnotesize]{(H2)} node[below, font=\footnotesize]{$L_t=\dfrac{2K}{c_0}=O(1)$} (t);
\draw[arr] (t) -- node[above, font=\footnotesize]{Lem.~\ref{lem:awgn-logf}} node[below, font=\footnotesize]{$C_F\,n$} (F);
\end{tikzpicture}
\caption{\textbf{The cosine chain and its constants.} The direction map moves
at scale $K\sqrt n$ in $\theta$ (H1); normalizing by the direction's typical
length $c_0\sqrt n$ (H2) cancels the $\sqrt n$'s, leaving the order-one cosine
constant $L_t=2K/c_0$; the tail $\log F$ then moves at slope $C_Fn$ in the
cosine. A $1/n$ mesh in $\theta$ therefore perturbs the log-PEP by only
$C=C_FL_t=O(1)$ --- the single surviving factor of $n$ is the exponent scale
itself (\cref{thm:awgn-disc}).}
\label{fig:uc-lip-chain}
\end{figure}

\begin{lemma}[Logarithmic Lipschitz of the angular-correlation tail]
\label{lem:awgn-logf}
For every $t\in(0,1)$,
\begin{equation}\label{eq:awgn-logf-deriv}
  \abs{\frac{d}{dt}\log F(t)} \;\le\; \frac{n-1}{t\,(1-t^2)} .
\end{equation}
Consequently, for any $0<t_-<t_+<1$ and $t,t'\in[t_-,t_+]$,
\begin{equation}\label{eq:awgn-logf-lip}
  \abs{\log F(t)-\log F(t')} \;\le\; C_F\,n\,\abs{t-t'},
  \qquad C_F \;\triangleq\; \frac{1}{t_-\,(1-t_+^2)} .
\end{equation}
\end{lemma}

\begin{IEEEproof}
Differentiate $\log F(t)=\tfrac{n-1}{2}\log(1-t^2)+\log I_n(t)$. The first term
contributes $-(n-1)t/(1-t^2)$. For the second, differentiating under the integral
and using the pointwise bound $\tan^2\phi/(1+t^2\tan^2\phi)\le 1/t^2$ gives
$\abs{I_n'(t)/I_n(t)}\le(n-1)/t$. Adding the two bounds
yields~\eqref{eq:awgn-logf-deriv}, and~\eqref{eq:awgn-logf-lip} follows from the
mean value theorem on the compact interval $[t_-,t_+]$.
\end{IEEEproof}

The slope $\Theta(n)$ of $\log F$ in the cosine is exactly the AWGN
error-exponent slope, expressed in the angular variable. The $t$-dependence of
the constant is read off the closed form. The main term contributes the
angular-exponent slope
$-(n-1)t/(1-t^2)$ --- $\Theta(n)$ on any interior band,
genuinely blowing up only at $t\to1$, where the cap degenerates; the
correction estimate $\abs{I_n'/I_n}\le(n-1)/t$ adds the $1/t$, and the two
collapse via
$t/(1-t^2)+1/t=1/(t(1-t^2))$. The $1/t$ factor is an artifact of that uniform
estimate, not a genuine $t\to0$ blowup ($F(0)=\tfrac12$, with bounded true
slope there); its looseness is immaterial, since (H3) confines every use to a
compact interior band $[t_-,t_+]$.

\begin{lemma}[Lipschitz of the cosine in the parameter]
\label{lem:awgn-cosine}
Suppose there is a typical event $\cE_n^Y\subset\mR^n$ on which:
\begin{description}
  \item[(H1)] \emph{Lipschitz direction map.} There is $K\ge 0$ with
  $\norm{v_\theta(y)-v_{\theta'}(y)}\le K\sqrt n\,\norm{\theta-\theta'}$
  for every $y\in\cE_n^Y$ and $\theta,\theta'\in\Theta$.
  \item[(H2)] \emph{Non-degenerate direction.} There is $c_0>0$ with
  $\norm{v_\theta(y)}\ge c_0\sqrt n$ for all
  $\theta\in\Theta$, $y\in\cE_n^Y$.
\end{description}
Then for every $x\in\mR^n\setminus\{0\}$ and $y\in\cE_n^Y$,
\begin{equation}\label{eq:awgn-t-lip}
  \abs{t_\theta(x,y)-t_{\theta'}(x,y)} \;\le\; L_t\,\norm{\theta-\theta'},
  \qquad L_t \;\triangleq\; \frac{2K}{c_0}.
\end{equation}
\end{lemma}

\begin{IEEEproof}[Proof sketch]
By Cauchy--Schwarz, $\abs{t_\theta(x,y)-t_{\theta'}(x,y)}\le\norm{\hat
v_\theta-\hat v_{\theta'}}$ where $\hat v=v/\norm{v}$. An asymmetric two-term
split of the unit-vector difference bounds it by
$2\norm{v_\theta-v_{\theta'}}/\norm{v_\theta}$, combining the numerator bound
(H1) with the denominator lower bound (H2); the two $\sqrt n$ factors cancel,
leaving the order-one constant $L_t=2K/c_0$. The full proof is in
Appendix~\ref{app:cosine-lip}.
\end{IEEEproof}

For families whose direction map is Lipschitz for \emph{every} $y$ --- such
as the interference family below, where
$v_\theta(y)-v_{\theta'}(y)=S(\theta')-S(\theta)$ is $y$-free --- (H1) holds
with $\cE_n^Y=\mR^n$. The typical-set form is what the equalize-and-decode
family of \cref{sec:awgn-isi} genuinely requires, since there
$\norm{G_hy-G_{h'}y}$ scales with $\norm{y}$, which is $O(\sqrt n)$ only on a
typical event.

\begin{remark}[Only the lower bound on $\norm{v_\theta}$ is used]
\label{rem:awgn-onlylower}
The argument uses (H2) solely as a lower bound in the denominator: a larger
$\norm{v_\theta}$ would only shrink $L_t$. In the applications
$\norm{v_\theta(y)}=\Theta(\sqrt n)$ on the typical set --- the matched direction
is built from quantities of size $\Theta(\sqrt n)$ (codeword, output,
interference) --- but the upper side of this scaling is never invoked.
\end{remark}

\begin{remark}[(H2)--(H3) are for-all-$\theta$ statements]
\label{rem:awgn-typical}
Hypotheses (H2) and (H3) below each demand a \emph{single} typical event on
which the stated bound holds \emph{simultaneously for every}
$\theta\in\Theta$: the GLRT sweeps $\theta$, and the proofs consume the
all-$\theta$ control at each typical pair. Per-parameter concentration ---
$\norm{v_\theta(Y^n)}=\Theta(\sqrt n)$ with high probability for each
\emph{fixed} $\theta$ --- does not deliver this over the uncountable
$\Theta$. Two standard devices close the gap in the applications: \emph{(H2)
by projection} --- when the $\theta$-dependence of $v_\theta$ is confined to a
fixed finite-dimensional subspace, a $\theta$-free projection bound covers all
$\theta$ by one $\chi^2$ event (interference; the ISI analogue substitutes a
uniform spectral lower bound on the equalizer); and \emph{(H3) by net} --- an
$\eta$-net over the compact $\Theta$, per-point concentration, a union bound,
and \cref{lem:awgn-cosine} to extend from the net to all of $\Theta$; since
$L_t=O(1)$ and the band margins are constants, a \emph{constant} mesh
suffices, so the net has $O(1)$ points. The order matters: $L_t=2K/c_0$
presupposes (H2), so (H2) is established first, by projection, and only then
is the net extended --- no circularity. The actual verifications are
\cref{lem:awgn-interf} (interference) and \cref{lem:awgn-isi} (E\&NN/ISI,
under explicit family restrictions).
\end{remark}

Combining the two Lipschitz bounds through a Euclidean net gives the
discretization theorem.

\begin{theorem}[Discretization of parametric Gaussian metrics]
\label{thm:awgn-disc}
Let $\Theta\subset\mR^d$ be compact with Euclidean diameter $D$, and let
$\{m_\theta\}$ be of the form~\eqref{eq:awgn-linear-metric}. Assume (H1)--(H2) of
\cref{lem:awgn-cosine}, with $L_t=2K/c_0$, and:
\begin{description}
  \item[(H3)] \emph{Bounded codeword angle.} There are $0<t_-<t_+<1$, a constant
  $\delta>0$ uniform in $n$ and $\theta$, and a typical event
  $\cE_n\subset\cX^n\times\cE_n^Y$ with $\PR{(X^n,Y^n)\notin\cE_n}\le e^{-n\delta}$
  under the channel-joint law $\bbQ_X\cdot\bbW$ for \emph{every} channel in the
  family --- every admissible true parameter, including nuisance parameters not
  indexing the metric --- and $t_\theta(x,y)\in[t_-,t_+]$ for all
  $\theta\in\Theta$, $(x,y)\in\cE_n$.
\end{description}
Then for every $\eta>0$ small enough that
$L_t\eta\le\tfrac12\min(t_-,1-t_+)$ there is a finite
$\Theta_d^n\subset\Theta$ with $\abs{\Theta_d^n}\le(D\sqrt d/\eta)^d$ such that
every $\theta\in\Theta$ has a $\theta_d\in\Theta_d^n$ with
$\norm{\theta-\theta_d}\le\eta$ and, on $\cE_n$,
\begin{equation}\label{eq:awgn-pep-equiv}
  e^{-Cn\eta}\,\PEP{x}{y}_{m_\theta}
  \;\le\; \PEP{x}{y}_{m_{\theta_d}}
  \;\le\; e^{Cn\eta}\,\PEP{x}{y}_{m_\theta},
\end{equation}
where $C\triangleq C_F\,L_t$ combines the two Lipschitz constants, with $C_F$
the log-$F$ constant of \cref{lem:awgn-logf} evaluated on the \emph{enlarged}
band $[t_-/2,(1+t_+)/2]$ (the expanded value is displayed in the proof). In
particular, $\eta=1/n$ gives $\abs{\Theta_d^n}=O(n^d)$ with a constant-factor
PEP shift on the typical set.
\end{theorem}

\begin{IEEEproof}
Let $\Theta_d^n$ be a minimal Euclidean $\eta$-net; the volume bound gives
$\abs{\Theta_d^n}\le(D\sqrt d/\eta)^d$. For $\theta\in\Theta$ pick
$\theta_d$ with $\norm{\theta-\theta_d}\le\eta$. By \cref{lem:awgn-cosine}, on
$\cE_n$, $\abs{t_\theta-t_{\theta_d}}\le L_t\eta$. By (H3),
$t_\theta(x,y)\in[t_-,t_+]$ on $\cE_n$; the grid cosine therefore lies in the
enlarged band $t_{\theta_d}(x,y)\in[t_--L_t\eta,\,t_++L_t\eta]$, which the
mesh constraint $L_t\eta\le\tfrac12\min(t_-,1-t_+)$ keeps inside
$[t_-/2,\,(1+t_+)/2]\subset(0,1)$. \Cref{lem:awgn-logf}, applied on this
enlarged band, then gives
$\abs{\log F(t_\theta)-\log F(t_{\theta_d})}\le C_F n\cdot L_t\eta=Cn\eta$,
explicitly with
\[
  C\;=\;C_F\,L_t\;=\;\frac{2K/c_0}{(t_-/2)\bigl(1-((1+t_+)/2)^2\bigr)}.
\]
Since $\log\PEP{x}{y}_{m_\theta}=\log F(t_\theta)$ by~\eqref{eq:awgn-pep-via-f},
exponentiating yields~\eqref{eq:awgn-pep-equiv}.
\end{IEEEproof}

\Cref{thm:awgn-disc} produces a polynomial-cardinality representative family on
which the merge construction of \cref{sec:univ-metric} applies with penalty
$\tfrac1n\log\abs{\Theta_d^n}=O(\log n/n)\to 0$. The resulting decoder is the
generalized-likelihood-ratio test (GLRT)~\cite{lapidoth1998reliable} over the
grid.

\subsection{Universality of the continuous GLRT}
\label{sec:awgn-glrt-sub}

The natural decoder for an unknown parameter minimizes the member PEP jointly
over the message and the entire continuum:
\begin{equation}\label{eq:awgn-glrt-rule}
  \hat\imath(y)\;=\;\argmin_{i,\,\theta\in\Theta}\,\PEP{x_i}{y}_{m_\theta},
\end{equation}
with no reference to a grid. By~\eqref{eq:awgn-pep-via-f} and the monotonicity
of $F$, this is the \emph{normalized (cosine)} GLRT
$\argmax_i\sup_{\theta}t_\theta(x_i,y)$ --- \emph{not} the unnormalized
$\argmax_i\sup_\theta m_\theta(x_i,y)$, which fails rank-equivalence across
$\theta$ because the normalizer $\norm{v_\theta(y)}$ depends on $\theta$ (see
the taxonomy in \cref{sec:awgn-interference}). In merge form, clip the inverse
PEP at level $e^{cn}$ and define
\begin{equation}\label{eq:awgn-barm}
  \bar m^n(x,y)\;\triangleq\;\min\!\BRAs{-\log\inf_{\theta\in\Theta}\PEP{x}{y}_{m_\theta},\;cn};
\end{equation}
the rule~\eqref{eq:awgn-glrt-rule} is then
$\hat\imath(y)=\argmax_i\bar m^n(x_i,y)$.
The infimum over the continuum raises no measurability issue: for
$y\in\cE_n^Y$, (H1)--(H2) make $\theta\mapsto t_\theta(x,y)$ Lipschitz
(\cref{lem:awgn-cosine}) and hence
$\theta\mapsto\PEP{x}{y}_{m_\theta}=F(t_\theta(x,y))$ continuous on the
compact $\Theta$, so a single countable dense subset of $\Theta$ attains the
infimum at every such pair --- the typical-event form of the
countable-subfamily hypothesis of \cref{prop:uc-continuum-merge} --- and off
the typical event the restricted decoder defined next is constant, so it is
measurable throughout.
Hypothesis (H3) controls the channel-joint typical set, but at a fixed $y$ the
product-law slice $Q_X^n\{\tilde X^n:(\tilde X^n,y)\notin\cE_n\}$ is not bounded
by (H3) and is in fact close to $1$. To bypass it we analyze the
typical-set-restricted decoder
\begin{equation}\label{eq:awgn-barm-cE}
  \bar m^n_\cE(x,y)\;\triangleq\;
  \bar m^n(x,y)\,\Ind{(x,y)\in\cE_n}\;-\;cn\,\Ind{(x,y)\notin\cE_n},
\end{equation}
which agrees with $\bar m^n$ on every channel-typical pair and is set to the
lowest score $-cn$ off $\cE_n$. Score agreement on $\cE_n$ does \emph{not} by
itself make the two decisions coincide: a decision at $y$ compares the scores
of \emph{all} codewords, including competitor pairs $(x_j,y)\notin\cE_n$ ---
typically most of them, by the slice observation above --- whose scores the
restricted decoder floors at $-cn$ while the bare rule does not.

\begin{theorem}[Universality of the continuous GLRT]
\label{thm:awgn-glrt}
Assume (H1)--(H3) of \cref{thm:awgn-disc}, and let $\delta>0$ be the
typical-set rate, so that $\PR{(X^n,Y^n)\notin\cE_n}\le e^{-n\delta}$ under
$\bbQ_X\cdot\bbW_{\theta_*}$ for every admissible true parameter. Choose any
clipping constant $c>0$ and let $\bar m^n_\cE$ be as
in~\eqref{eq:awgn-barm-cE}. Then for every $\theta_*\in\Theta$ and every rate
$R<c$,
\begin{align}
  \Erc\bigl(R;\,\bbW_{\theta_*},\,\bbQ_X,\,\bar m^n_\cE\bigr)
  &\;\ge\;
  \min\Bigl\{\,\sup_{\theta'\in\Theta}
    \Erc\bigl(R;\,\bbW_{\theta_*},\,\bbQ_X,\,m_{\theta'}\bigr),\;c-R,\;
    \delta\Bigr\}\label{eq:awgn-glrt-exp}\\
  &\;\ge\;
  \min\bigl\{\Erc\bigl(R;\,\bbW_{\theta_*},\,\bbQ_X,\,m_{\theta_*}\bigr),\;
    c-R,\;\delta\bigr\}.\notag
\end{align}
\end{theorem}

The first form is the \emph{best-in-class} guarantee: the supremum ranges over
the entire metric family, not only the member indexed by the true parameter;
the second (matched-member) form is the special case $\theta'=\theta_*$. The
best-in-class exponent is attained in full whenever it is at most
$\min\{c-R,\delta\}$; in general the guarantee is the displayed minimum. The
$c-R$ term is an artifact of the clipping level and removable; the $\delta$
term is intrinsic to the typical-set route (\cref{rem:awgn-cliplevel}).

\begin{IEEEproof}
The proof retraces the merge argument of \cref{thm:univ-rc}, with the grid
family replaced by the continuum $\Theta$ and the clipped-inverse-PEP bound
replaced by a combination of the discretization theorem and the grid version
of the same bound; the typical-set-restricted form $\bar m^n_\cE$ is the
object analyzed --- this is the typical-set analogue of
\cref{prop:uc-continuum-merge}. Steps 1--3 (a uniform expectation bound on
$e^{\bar m^n_\cE}$, a pointwise bound on legitimate pairs, and the resulting
Markov PEP bound), deferred to Appendix~\ref{app:glrt-steps}, establish: for
every $(x,y)\in\cE_n$ and every \emph{fixed} $\theta'\in\Theta$ --- Step 2
bounds the infimum defining $\bar m^n$ by its value at an arbitrary fixed
$\theta'$, of which $\theta'=\theta_*$ is the special case ---
\begin{equation}\label{eq:awgn-step3-pep}
  \PEP{x}{y}_{\bar m^n_\cE}
  \;\le\;B_n\cdot\max\bigl\{\PEP{x}{y}_{m_{\theta'}},\;e^{-cn}\bigr\},
  \qquad B_n\;\triangleq\;e^C\cdot O\bigl(n^{d+1}\bigr),
\end{equation}
with $B_n$ subexponential and uniform in $y$; no channel-joint statement is
invoked there.

\emph{Step 4: random-coding exponent.} Fix $\theta'\in\Theta$ and split the
average error by legitimate-pair typicality:
\begin{align*}
  &\E{\min\bigl(1,e^{nR}\PEP{X^n}{Y^n}_{\bar m^n_\cE}\bigr)}\\
  &\quad\le B_n\,\E{\min\bigl(1,e^{nR}\max\{\PEP{X^n}{Y^n}_{m_{\theta'}},
    e^{-cn}\}\bigr)}+\PR{(X^n,Y^n)\notin\cE_n}\\
  &\quad\le B_n\,\E{\min\bigl(1,e^{nR}\PEP{X^n}{Y^n}_{m_{\theta'}}\bigr)}
    +B_n e^{n(R-c)}+e^{-n\delta},
\end{align*}
using \eqref{eq:awgn-step3-pep} on the $\cE_n$ event, $\min(1,\cdot)\le1$ off
it, and $\max(a,b)\le a+b$. Taking exponents (the exponent of a sum being the
minimum of the exponents) yields, for the fixed $\theta'$,
$\Erc(R;\bar m^n_\cE)\ge\min\{\Erc(R;m_{\theta'}),c-R,\delta\}$. The left side
does not depend on $\theta'$, so taking the supremum over $\theta'\in\Theta$
gives the first form of \eqref{eq:awgn-glrt-exp}, and $\theta'=\theta_*$ the
second.
\end{IEEEproof}

\begin{remark}[The clipping level: $c-R$ is removable, $\delta$ is not]
\label{rem:awgn-cliplevel}
Rerunning Steps 1--4 with the grid merge's clip $\alpha_n=e^{n^2}$ (and floor
$-n^2$ off $\cE_n$) goes through verbatim: Step 1 then gives
$\E{e^{\bar m^n_\cE(\tilde X^n,y)}}\le e^C|\Theta_d^n|(1+n^2)+e^{-n^2}
=e^C\,O(n^{d+2})$, still subexponential and uniform in $y$, and the middle
term of Step 4 becomes $B_ne^{nR-n^2}$, superexponentially small at every
rate. The guarantee is then $\min\{\sup_{\theta'}\Erc,\delta\}$ for
\emph{every} $R>0$: the $c-R$ term and the restriction $R<c$ disappear.
Equivalently, within the theorem as stated, $c$ is a free design constant and
$c\ge R+\delta$ makes the middle term non-binding. We nonetheless state the
theorem at a general clip $e^{cn}$ because the clip belongs to the
\emph{decoder}, not only to the proof: for the bare GLRT $\bar m^n$ --- the
implementable rule whose universality is the open point of
\cref{rem:awgn-bare} --- a low clipping level caps rogue competitor scores on
the product-law event $\cE_n^c$, which any resolution of that open question
must control.

The $\delta$ term, by contrast, is not removable within this route, and
$\delta$ cannot be made arbitrarily large: any (H3)-compliant typical set
excludes the cancellation event $\{\BRAi{X^n,v_{\theta_*}(Y^n)}\le0\}$
(there $t_{\theta_*}\le0<t_-$). For the interference family this event has
channel probability $e^{-n\,\SNR/2\,(1+o(1))}$: $v_{\theta_*}(Y^n)=hX^n+Z^n$,
so conditionally on $X^n$ the score
$\BRAi{X^n,v_{\theta_*}(Y^n)}=hnP+\BRAi{X^n,Z^n}$ is Gaussian with mean
$hnP$ and variance $nP\sigma^2$: the event is exactly a standard-normal
tail at $\sqrt{n\,\SNR}$ deviations, of exponent $\SNR/2$ (the ISI
analogue carries its own constants). Hence, uniformly over the SNR-constrained
family, $\delta\le h_{\min}^2P/(2\sigma_{\max}^2)$. Widening the band $[t_-,t_+]$
raises $\delta$ toward this cap at the price of a larger $C_F$ in
\cref{lem:awgn-logf}.
\end{remark}

\begin{remark}[The bare GLRT]
\label{rem:awgn-bare}
The decoder one would implement is the bare $\bar m^n$
of~\eqref{eq:awgn-barm} --- joint minimization over $(i,\theta)$. It assigns
the same score as $\bar m^n_\cE$ to every channel-typical pair, but --- as
noted before \cref{thm:awgn-glrt} --- score agreement does not make the two
decisions coincide. Whether $\bar m^n$ attains
the same exponent without the restriction (which would require controlling
competitor scores on the product-law slice) is left open.
\end{remark}

\subsection{Application: AWGN with deterministic interference}
\label{sec:awgn-interference}

For $Y=hX+S(\theta)+Z$ with $Z\sim\cN(0,\sigma^2 I_n)$, unknown
$\theta\in\Theta\subset\mR^d$ compact, $h\in[h_{\min},h_{\max}]$,
$\sigma\in[\sigma_{\min},\sigma_{\max}]$ subject to the SNR bound
\begin{equation}\label{eq:awgn-snr-bound}
  \SNR(h,\sigma)\;=\;h^2P/\sigma^2\;\le\;\Gamma_{\max},
\end{equation}
and structured interference $S(\theta)=\sum_{j=1}^d\theta_j\psi_j$ with bounded
basis $\norm{\psi_j}^2=O(n)$, the matched-ML metric reduces on the shell to
\begin{equation}\label{eq:awgn-interf-metric}
  m_\theta(x,y)\;=\;\BRAi{x,\,v_\theta(y)},
  \qquad v_\theta(y)\;\triangleq\;y-S(\theta),
\end{equation}
parametrized by $\theta$ alone. Hypothesis (H1) is immediate and needs no
typicality: $v_\theta(y)-v_{\theta'}(y)=S(\theta')-S(\theta)$ for \emph{every}
$y$, so the basis expansion gives
$\norm{S(\theta)-S(\theta')}\le(\sum_j\norm{\psi_j}^2)^{1/2}
\norm{\theta-\theta'}=O(\sqrt n)\norm{\theta-\theta'}$, i.e.\ $K=O(1)$.
Hypotheses (H2)--(H3) are for-all-$\theta$ statements on a single typical
event (\cref{rem:awgn-typical}); their verification, uniform across the
SNR-constrained family, is the following lemma.

\begin{lemma}[(H2)--(H3) for the interference family, uniformly over $\Theta$]
\label{lem:awgn-interf}
Let $V\triangleq\mathrm{span}\{\psi_1,\dots,\psi_d\}$, let $P_{V^\perp}$ be
the orthogonal projection onto $V^\perp$, and set
$S_{\max}\triangleq\sup_{\theta\in\Theta}\norm{S(\theta)}/\sqrt n$ (finite,
uniformly in $n$, by compactness of $\Theta$ and $\norm{\psi_j}^2=O(n)$).
Then:
\begin{enumerate}
\item[(i)] For all $n\ge2d$, (H2) holds with $c_0=\sigma_{\min}/2$ and the
  $\theta$-free event
  $\cE_n^Y\triangleq\{y\in\mR^n:\norm{P_{V^\perp}y}\ge(\sigma_{\min}/2)
  \sqrt n\}$.
\item[(ii)] There exist constants $0<t_-<t_+<1$, $\delta>0$ and $n_0$,
  depending only on
  $(P,\allowbreak d,\allowbreak D,\allowbreak K,\allowbreak
  h_{\min},\allowbreak h_{\max},\allowbreak\sigma_{\min},\allowbreak
  \sigma_{\max},\allowbreak S_{\max})$, such
  that (H3) holds for all $n\ge n_0$, uniformly over all true parameters
  $(h,\sigma,\theta_*)$ in the SNR-constrained family.
\end{enumerate}
\end{lemma}

\begin{IEEEproof}[Proof sketch]
(i) \emph{Projection.} $S(\theta)\in V$ for every $\theta$, so
$P_{V^\perp}v_\theta(y)=P_{V^\perp}y$ and
$\norm{v_\theta(y)}\ge\norm{P_{V^\perp}y}$ \emph{simultaneously for all}
$\theta$: one event serves the whole family. Under the channel,
$P_{V^\perp}Y^n$ is, conditionally on $X^n$, noncentral Gaussian on
$V^\perp$; its squared norm stochastically dominates $\sigma^2\chi^2_{n-d}$
and the $\chi^2$ lower tail gives $\PR{Y^n\notin\cE_n^Y}\le e^{-(n-d)/12}$.
(ii) The upper band mirrors (i): projecting onto
$W_x\triangleq(V+\mathrm{span}\{x\})^\perp$ annihilates codeword and
interference simultaneously for all $\theta$, leaving pure noise, so
$1-t_\theta^2\ge\norm{P_{W_x}Z^n}^2/\norm{v_\theta(Y^n)}^2$ is bounded below
by a constant on a single $\chi^2$ event. The lower band uses the net: on a
constant-mesh net ($O(1)$ points) the score $\BRAi{X^n,v_{\theta_i}(Y^n)}$
concentrates around $nhP$ per point; a union bound and
\cref{lem:awgn-cosine} --- available since (H2) is already established ---
extend the band to all of $\Theta$. The full proof is in
Appendix~\ref{app:interf-typical}.
\end{IEEEproof}

By \cref{thm:awgn-glrt} the continuous GLRT, in normalized (cosine) form,
\begin{equation}\label{eq:awgn-glrt-interference}
  \hat\imath(y)\;=\;\argmax_i\,\sup_{\theta\in\Theta}\,
  \frac{\BRAi{x_i,\,y-S(\theta)}}{\norm{y-S(\theta)}}
\end{equation}
attains the matched-ML random-coding exponent, up to the $\delta$-cap, for
every $(h,\sigma,\theta)$ satisfying~\eqref{eq:awgn-snr-bound}. Two caveats
locate the claim precisely. First, the guarantee is the min-capped one of
\cref{thm:awgn-glrt}: the full matched-ML exponent is attained whenever the
cap $\min\{c-R,\delta\}$ is not binding. Second, the certified decoder is the
typical-set-restricted form $\bar m^n_\cE$ of~\eqref{eq:awgn-barm-cE}, which
assigns the same score as the bare rule~\eqref{eq:awgn-glrt-interference} to
every channel-typical pair; whether the bare rule itself attains the same
exponent is the open point of \cref{rem:awgn-bare}.

\paragraph{Which rule is this?}
For each fixed $\theta$ the normalizer $\norm{y-S(\theta)}$ is common to all
codewords, so the fixed-$\theta$ rule is the usual correlation decoder. The
joint supremum over $\theta$ requires the normalized (cosine) form: the
unnormalized $\argmax_i\sup_\theta\BRAi{x_i,y-S(\theta)}$ is \emph{not}
rank-equivalent, because $\norm{y-S(\theta)}$ depends on $\theta$. Nor is the
joint rule the profiled-likelihood GLRT: maximizing the likelihood over
$(h,\sigma)$ leaves a score whose $\theta$-dependent factor survives the joint
supremum. The cosine form is the \emph{continuum PEP-merge} of
\cref{prop:uc-continuum-merge}: it coincides with the GLR ranking at each
fixed $\theta$ on the relevant range $t\ge0$, but across $\theta$ it commits
to the PEP scale --- the canonical one --- rather than the raw likelihood
scale. Raw-likelihood merging can fail outright:
\cite[Sec.~3]{lapidoth1998universality} exhibit a two-member family --- a pair
of artificial, blocklength-indexed noise laws, so not itself finite-state ---
for which the GLRT is not universal, while the two-element merge of the
matched decoders is.

\begin{remark}[Comparison with Merhav]
\label{rem:awgn-merhav}
Merhav~\cite{merhav1993universal} treats interference admitting an expansion
in a fixed system of bounded orthonormal functions with absolutely summable
coefficients, and decodes by a continuous analogue of MMI: a GLRT over
auxiliary backward channels of slowly growing order $k_n=O(n^{1/3})$, analyzed
by partitioning competitors into subexponentially many conditional
$\epsilon$-types over a quantized grid. In effect this is a subexponential
covering
of his interference class, at richer generality than the fixed-dimensional
family here (growing order versus fixed $d$). We expect the merge route to
reach his class as well --- absolute summability lets one truncate the
expansion and grid the surviving coefficients --- but we have not carried out
the (H2)--(H3) verification for the truncated family; we record the extension
as an expectation, not a claim. His argument is likewise restricted to a
high-probability set, with the exceptional probability calibrated per rate to
out-decay the target exponent; the residual exponent loss vanishes with the
shell thickness of his input ensemble rather than through a fixed typicality
rate, which is why his statement carries no $\min\{\cdot,\delta\}$ cap.
Neither route covers arbitrary bounded-energy interference.
\end{remark}

\subsection{Application: ISI via the equalize-and-decode family}
\label{sec:awgn-isi}

For the Gaussian ISI channel
\begin{equation}\label{eq:awgn-isi-channel}
  Y_i\;=\;\sum_{\ell=0}^{L-1}h_\ell\,X_{i-\ell}+Z_i,
  \qquad Z_i\sim\cN(0,\sigma^2),
\end{equation}
with unknown tap vector $h\in\cH\subset\mR^L$ compact and Toeplitz operator
$H_h$ ($Y=H_hX+Z$), matched-ML carries the quadratic $\norm{H_h x}^2$ term and
sits outside the linear-in-$x$ class. We therefore work inside the
equalize-and-decode (E\&NN) family: fix an equalizer map $h\mapsto G_h$ and use
\begin{equation}\label{eq:awgn-isi-metric}
  m_h(x,y)\;=\;\BRAi{x,\,G_h y}.
\end{equation}
Standard choices are matched filter $G_h=H_h^\top$, zero-forcing
$G_h=H_h^\dagger$, and MMSE $G_h=(H_h^\top H_h+(\sigma^2/P)I)^{-1}H_h^\top$.
For the matched filter $G_h$ is linear in $h$ and
$\norm{H_h^\top y-H_{h'}^\top y}\le\sqrt L\,\norm{y}\,\norm{h-h'}$, which is
$O(\sqrt n)\norm{h-h'}$ \emph{on the typical set} $\norm{y}=O(\sqrt n)$ ---
this is exactly the typical-$y$ form of (H1) in \cref{lem:awgn-cosine}; for
zero-forcing and MMSE the same holds under a uniform spectral lower bound on
$\cH$ that excludes taps near a channel null. Unlike the interference family,
the projection trick is unavailable here --- $h$ multiplies $y$ instead of
shifting it by a vector in a fixed subspace --- and the (H2)--(H3)
verification requires genuine restrictions on the family, stated explicitly.

\begin{lemma}[(H2)--(H3) for the E\&NN family, uniformly over $\cH$]
\label{lem:awgn-isi}
Let $\cH\subset\mR^L$ be compact with diameter $D_\cH$, let $s_{\min}(\cdot)$
denote the smallest singular value, and assume:
\begin{enumerate}
\item[(a)] \emph{Spectral band.} There are constants $0<\gamma_G\le\Gamma_G$
  and $\Gamma_H<\infty$ with $s_{\min}(G_h)\ge\gamma_G$,
  $\norm{G_h}_{\mathrm{op}}\le\Gamma_G$, and
  $\norm{H_h}_{\mathrm{op}}\le\Gamma_H$ for all $h\in\cH$.
\item[(b)] \emph{(H1) on the typical set.}
  $\norm{G_hy-G_{h'}y}\le K\sqrt n\,\norm{h-h'}$ for all $h,h'\in\cH$ and all
  $y\in\cE_n^Y$ of part (i) below; on $\cE_n^Y$, $\norm{y}\le C_Y\sqrt n$, so
  the operator-Lipschitz bounds above deliver such a $K=O(1)$.
\item[(c)] \emph{Uniform equalizer--channel alignment.}
  \[
    \kappa\;\triangleq\;\inf_{h,h'\in\cH}\;
    \lambda_{\min}\!\BRA{\tfrac12\bigl(G_hH_{h'}+(G_hH_{h'})^\top\bigr)}
    \;>\;0.
  \]
\end{enumerate}
Then there are constants $0<t_-<t_+<1$, $\delta>0$ and $n_0$, depending only
on $(P,\allowbreak L,\allowbreak\gamma_G,\allowbreak\Gamma_G,\allowbreak
\Gamma_H,\allowbreak\sigma_{\min},\allowbreak\sigma_{\max},\allowbreak
\kappa,\allowbreak D_\cH,\allowbreak K)$, such that for all $n\ge n_0$: (i) (H2) holds with
$c_0=\gamma_G\sigma_{\min}/\sqrt2$ and the $h$-free event
$\cE_n^Y\triangleq\{y:(\sigma_{\min}/\sqrt2)\sqrt n\le\norm{y}\le C_Y\sqrt n\}$,
$C_Y\triangleq\Gamma_H\sqrt P+\sqrt2\,\sigma_{\max}$; (ii) (H3) holds with the
band $[t_-,t_+]$ and rate $\delta$, uniformly over the true $h_*\in\cH$ and
$\sigma\in[\sigma_{\min},\sigma_{\max}]$.
\end{lemma}

\begin{IEEEproof}[Proof sketch]
(i) By (a), $\norm{G_hy}\ge s_{\min}(G_h)\norm{y}\ge\gamma_G\norm{y}$ for
every $h$ at once, so the $h$-free norm event carries (H2); the norm event's
probability follows from noncentral-Gaussian domination of $\sigma^2\chi^2_n$
(\cref{lem:awgn-noncentral-dom}) and a $\chi^2$ upper tail. (ii) Lower band:
by (c) the signal part of the score is bounded \emph{deterministically} on
the shell, $\BRAi{x,G_hH_{h_*}x}\ge\kappa nP$ for every pair $(h,h_*)$, so
only the noise cross-term needs a constant-mesh net and a union bound; upper
band: per net point, the component of $G_hY^n$ orthogonal to $X^n$ is
conditionally Gaussian with covariance $\succeq\gamma_G^2\sigma^2I$ on an
$(n-1)$-dimensional space. In both bands \cref{lem:awgn-cosine} --- valid
since (H2) is established first --- extends from the net to all of $\cH$.
The full verification parallels the interference proof of
Appendix~\ref{app:interf-typical} step for step --- concentration events per
net point, deterministic signal bound, projection for the upper band, net
extension --- and is given in Appendix~\ref{app:isi-typical}.
\end{IEEEproof}

\begin{remark}[What (a)--(c) do and do not allow]
\label{rem:awgn-isi-scope}
Condition (c) holds automatically on the diagonal $h'=h$ under (a): for the
matched filter $\tfrac12(G_hH_h+(G_hH_h)^\top)=H_h^\top H_h
\succeq\gamma_G^2I$; for zero-forcing $G_hH_h=I$; for MMSE the eigenvalues are
$s^2/(s^2+\sigma^2/P)$ over the singular values $s$ of $H_h$, bounded below when
$H_h$ has no spectral null. By continuity, (c) therefore holds for families
$\cH$ of sufficiently small diameter around a nominal filter. It is, however,
a genuine restriction for large-uncertainty families: for strongly mismatched
pairs the cosine can concentrate at a non-positive value --- for the matched
filter, $\BRAi{X^n,H_h^\top H_{h_*}X^n}$ has mean approximately
$nP\BRAi{h,h_*}$, non-positive for anti-correlated tap vectors --- and then no
lower band $t_->0$ exists and (H3) fails as stated. Likewise, (a) excludes
spectral nulls for all the equalizers above. \emph{Outside (a)--(c) we make no
universality claim for the E\&NN family.}
\end{remark}

By \cref{thm:awgn-glrt}, for each fixed equalizer map satisfying (a)--(c) the
continuous GLRT, in normalized form,
\begin{equation}\label{eq:awgn-glrt-isi}
  \hat\imath(y)\;=\;\argmax_i\,\sup_{h\in\cH}\,
  \frac{\BRAi{x_i,\,G_h y}}{\norm{G_hy}}
\end{equation}
attains, for every true $h_*\in\cH$, the best-in-class (E\&NN) random-coding
exponent
\begin{equation*}
  \sup_{h'\in\cH}\Erc(R;\bbW_{h_*},\bbQ_X,m_{h'})
\end{equation*}
subject to the $\delta$-cap of that theorem, with merge penalty
$O(\log n/n)$. The best-in-class exponent dominates the exponent of every
fixed equalizer index, not only the member tuned to the true taps; it is in
general strictly below the matched-ML exponent of $-\norm{y-H_h x}^2$.

\begin{remark}[Matched-ML for ISI: out of scope]
\label{rem:awgn-isi-open}
The decoder~\eqref{eq:awgn-glrt-isi} is universal within a fixed E\&NN family
only. Matched-ML universality for ISI requires the quadratic term
$\norm{H_h x}^2$, which leaves the linear-in-$x$
class~\eqref{eq:awgn-linear-metric}: on the power shell this term is not constant,
so the PEP decision region is an ellipsoid rather than a spherical cap and the
reduction of the PEP to the single cosine~\eqref{eq:awgn-cosine} no longer holds.
Matched-ML universality for the ISI channel --- a strictly stronger statement
requiring a different technical route than the cosine chain of this section ---
is \emph{not addressed by our approach}; we make no matched-ML claim for ISI
here. For the matched-ML \emph{exponent} on ISI, see Huleihel and
Merhav~\cite{huleihel2015universal}, who construct a data-driven equalizer
(tuned to the channel output, not to a fixed parameter) achieving the
matched-ML random-coding exponent; an earlier types-route construction is
Farkas~\cite{farkas2008blind}, who decodes by maximizing the mutual information
of a grid-quantized least-squares channel fit. For the matched-ML
\emph{capacity} (rate level only) on the cyclic-convolution model with unknown
filter, see Lapidoth and Telatar~\cite{lapidoth2000gaussian}, who show that
the bare GLRT $\argmin_{i,h}\norm{y-H_h x_i}^2$ attains any rate below the
matched-ML capacity. All three use techniques distinct from ours; we do not
redevelop them.

The merge itself is metric-agnostic (\cref{prop:uc-continuum-merge} consumes
only a moment bound), so the continuum merge over the matched-ML tap family is
already a well-defined decoder; what is missing is only the discretization
step, which for the quadratic metric would require a shell-prior
\emph{level-smoothness} statement in place of the closed form $F(t)$ available
here for the linear-in-$x$ class --- in substance a counterpart of the
conditional-type calculus of~\cite{huleihel2015universal}, which we record as
an open lemma and do not attempt. Note finally that the change-of-measure
route of~\cite{feder1998universal}, which converts likelihood closeness of
nearby channels directly into error-probability closeness for \emph{any}
decoder, side-steps PEP tails entirely; it is available only because the
metric family there \emph{is} the likelihood family, and has no analogue for
general metric classes, where the PEP route applies.
\end{remark}

%% file: sections/07-conclusion.tex
\section{Conclusion}\label{sec:conc}

Universal decoding, in this paper, was not a new theory but a corollary of the
pairwise-error-probability framework of~\cite{elkayampep1}. Three points
summarize the development.

\emph{The PEP is metric-agnostic, so a family of metrics is a family of spectra.}
Because the PEP and its error spectrum are defined for an arbitrary decoding
metric, decoding well against a whole family of metrics or channels is the
problem of comparing --- and merging --- the corresponding spectra, after the
PEP has canonicalized each metric onto the common probability scale
(\cref{sec:framework}).

\emph{A clipped inverse-PEP merge metric is random-coding universal.} The merge
metric of \cref{sec:univ-metric} pools the per-metric pairwise comparisons into a
single decision rule --- in the Kraft reading of \cref{rem:uc-kraft}, the
normalized-maximum-likelihood envelope of the conditional probability
assignments that the metrics induce. Its random-coding exponent matches
that of the best metric
in the family, simultaneously for every channel in the family
(\cref{thm:univ-rc}), and it is an asymptotic minimax/equalizer rule
(\cref{thm:univ-minimax}) whose regret is the envelope's log-mass $\gamma_n$. The guarantee derandomizes by expurgation to a
single deterministic code over any subexponential channel family
(\cref{thm:uc-rc-to-det}), including for the erasure decoder
(\cref{thm:uc-det-erasure}).

\emph{Classical universal decoders are special cases.} Reduced to the method of
types, the merge construction recovers the maximum-mutual-information decoder,
its tilted variant, and finite-state universal decoders (\cref{sec:dmc}); it
extends to a decoder with an erasure option (\cref{sec:erasure}); and, through a
discretization of separable metric families, it reaches continuous-alphabet
channels --- matched-ML universality up to a typical-set cap for AWGN with
deterministic interference, and, for ISI, the best exponent in a family of
equalize-and-decode rules under explicit assumptions on the equalizers and the
channel spectrum, under the same cap (\cref{sec:awgn}).

Three directions are left open:
\begin{itemize}
\item \emph{Beyond finite/separable families.} The merge construction is
quantified by the size (or covering number) of the family; characterizing the
exact universality penalty for richer families, and the sharpest discretization,
remains open.
\item \emph{Universality and prior optimization together.} The companion
paper~\cite{elkayampep1} made the prior-optimized converse a linear program;
combining prior optimization with the merge metric --- a universal \emph{and}
prior-optimal decoder at finite blocklength --- is not yet understood.
\item \emph{The operational price of universality.} The minimax theorem of
\cref{sec:univ-metric} pins the value of the fixed-output regret game, not
the channel-averaged cost of universality (\cref{rem:uc-minimax-scale}). The
operational analogue of the redundancy--capacity theorem of universal
coding and prediction~\cite{merhav1995strong,merhav1998universal} --- the
exact $\inf_u\sup_{\bbW\in\Cset}$ gap between achieved
and benchmark performance, expected at scale $(\log n)/n$ and bounded above
by the game value $\gamma_n$ --- is open even for a two-channel family.
\end{itemize}

%% file: sections/08-appendix.tex

\appendices

\input{sections/appendix-parts/01-recap}
\input{sections/appendix-parts/02-framework}
\input{sections/appendix-parts/03-universal-metric}
\input{sections/appendix-parts/04-dmc}
\input{sections/appendix-parts/05-erasure}
\input{sections/appendix-parts/06-awgn}

%% file: sections/appendix-parts/01-recap.tex

\section{Proof of the Exponent Reading (\Cref{prop:exponent})}
\label{app:exponent-reading}

\begin{IEEEproof}
Write $p=\PEP{X^n}{Y^n}$. For the lower bound, on the event
$\{-\tfrac1n\log p\le R+\rho\}$, of probability $\Fspec{Q_X^n}{R+\rho}$, the
integrand is at least $\min\{1,e^{-n\rho}\}=e^{-n\rho}$. For the upper bound,
slice: on $\{-\tfrac1n\log p\le R\}$ the integrand is at most $1$ and the
probability is $\Fspec{Q_X^n}{R}$ (the $\rho'=0$ term); on each of the $n^2$
slices $\{-\tfrac1n\log p\in(R+\tfrac kn,\,R+\tfrac{k+1}{n}]\}$,
$k=0,\dots,n^2-1$, the integrand is at most
$e^{-k}=e\cdot e^{-n\frac{k+1}{n}}$ and the probability is at most
$\Fspec{Q_X^n}{R+\tfrac{k+1}{n}}$, so each slice contributes at most
$e\sup_{\rho'\ge0}e^{-n\rho'}\Fspec{Q_X^n}{R+\rho'}$; on the remainder,
$p\le e^{-n(R+n)}$ and the integrand is at most $e^{-n^2}$. The exponent
display follows by optimizing the lower bound over $\rho$ and taking
$-\tfrac1n\log$ of \eqref{eq:recap-sandwich}, noting
$\sup_{\rho'\ge0}e^{-n\rho'}\Fspec{Q_X^n}{R+\rho'}
=e^{-n\inf_{\rho}\{E_1(R+\rho)+\rho\}}$: the condition
$\inf_{\rho}\{E_1(R+\rho)+\rho\}\le n$ makes the $e^{-n^2}$ term
subdominant, and the prefactor contributes the $O(\tfrac{\log n}{n})$.
\end{IEEEproof}

%% file: sections/appendix-parts/02-framework.tex

\section{The Gaussian Shift Bound for the Interference Family}
\label{app:rem-interf-shift}

This appendix verifies the law-closeness hypothesis of \cref{lem:uc-transfer}
for the interference family of \cref{sec:awgn-interference}, as used in the
continuum-upgrade argument of \cref{rem:uc-cd-to-c}. The setting: for
$\theta$ in a compact $\Theta\subset\mR^d$ ($d$ fixed), the channel
$W_\theta^n(\cdot\mid x^n)$ is the law of $hx^n+S(\theta)+Z^n$ with
$Z^n\sim\cN(0,\sigma^2 I_n)$, structured interference
$S(\theta)=\sum_{j=1}^d\theta_j\psi_j$ built from fixed signals with
$\norm{\psi_j}^2=O(n)$, and gain--noise pair $h\in[h_{\min},h_{\max}]$,
$\sigma\in[\sigma_{\min},\sigma_{\max}]$, $\sigma_{\min}>0$. In the transfer
application only $\theta$ is discretized; the representative channel shares
$(h,\sigma)$ with the true one. By Cauchy--Schwarz on the basis expansion,
\begin{equation}\label{eq:app-shift-lip}
  \norm{S(\theta)-S(\theta')}\;\le\;K\sqrt n\,\norm{\theta-\theta'}
  \quad\text{for all }\theta,\theta'\in\Theta,
  \qquad
  K\;\triangleq\;\sup_{n}\Bigl(\tfrac1n\textstyle\sum_{j=1}^d
  \norm{\psi_j}^2\Bigr)^{1/2}\!<\infty .
\end{equation}

\begin{lemma}[Gaussian shift bound]\label{app:lem-interf-shift}
  Fix $n\ge1$, $x^n\in\mR^n$, a shell constant $r_0>\sigma_{\max}$, and
  $\theta,\theta_d\in\Theta$ with $\norm{\theta-\theta_d}\le1/n$. Let
  \[
    B_{x^n}\;\triangleq\;\BRAs{y^n\in\mR^n:\;
    \norm{y^n-hx^n-S(\theta)}>r_0\sqrt n}
  \]
  be the shell-escape set of the $\theta$-law. Then:
  \begin{enumerate}
  \item[(i)] for every $y^n\notin B_{x^n}$,
    \[
      \abs{\log W_\theta^n(y^n\mid x^n)-\log W_{\theta_d}^n(y^n\mid x^n)}
      \;\le\;\frac{2r_0K+K^2}{2\sigma_{\min}^2}\;\triangleq\;a\;=\;O(1);
    \]
  \item[(ii)] $\displaystyle W_\theta^n(B_{x^n}\mid x^n)\;\le\;e^{-nE'(r_0)}$,
    where
    $E'(r_0)\triangleq\tfrac12\bigl(r_0^2/\sigma_{\max}^2-1
    -2\log(r_0/\sigma_{\max})\bigr)>0$, and $E'(r_0)\to\infty$ as $r_0\to\infty$,
    uniformly over the family.
  \end{enumerate}
  Both bounds also hold with the roles of $\theta$ and $\theta_d$ exchanged,
  with $B_{x^n}$ replaced by the shell-escape set of the $\theta_d$-law.
\end{lemma}

\begin{IEEEproof}
  The Gaussian normalizations cancel, leaving the log-likelihood gap
  $\bigl(2\BRAi{r,\Delta}+\norm{\Delta}^2\bigr)/(2\sigma^2)$ with
  $r\triangleq y^n-hx^n-S(\theta)$ and
  $\Delta\triangleq S(\theta)-S(\theta_d)$. The mesh and
  \eqref{eq:app-shift-lip} give $\norm{\Delta}\le K/\sqrt n$, while
  $\norm{r}\le r_0\sqrt n$ off $B_{x^n}$; Cauchy--Schwarz then bounds the gap
  by $(2r_0K+K^2/n)/(2\sigma^2)\le a$, proving (i): the $1/n$-mesh shrinks
  the shift to $O(1/\sqrt n)$ against an $O(\sqrt n)$ on-shell residual. For
  (ii), under $W_\theta^n(\cdot\mid x^n)$ the residual is exactly the noise,
  so $\norm{r}^2/\sigma^2\sim\chi^2_n$, and the Chernoff bound
  $\PR{\chi^2_n\ge nu}\le e^{-\frac n2(u-1-\log u)}$ at
  $u=r_0^2/\sigma^2\ge r_0^2/\sigma_{\max}^2>1$ gives the exponent
  $E'(r_0)$, since $u\mapsto u-1-\log u$ is increasing on $[1,\infty)$. Exchanging $\theta\leftrightarrow\theta_d$ flips the sign of
  $\Delta$ and leaves both bounds unchanged; the exchanged residual is again
  exactly $Z^n$ under the $\theta_d$-law.
\end{IEEEproof}

Consequently the hypothesis of \cref{lem:uc-transfer} holds per codeword
with $W^n=W_\theta^n$, $W'^n=W_{\theta_d}^n$: for fixed $\delta>0$ and
$n\ge a/\delta$, $W_\theta^n(y^n\mid x^n)\le e^{n\delta}\,W_{\theta_d}^n
(y^n\mid x^n)$ off $B_{x^n}$, with $E'(r_0)$ made to exceed any target
exponent by the choice of $r_0$; the exchanged orientation --- needed for
the second application in \cref{rem:uc-cd-to-c} --- holds with its own rare
set by the symmetry of the lemma.

%% file: sections/appendix-parts/03-universal-metric.tex

\section{Proof of the Averaged-Regret Lemma}
\label{app:avg-regret}

We prove \cref{lem:uc-avg-regret}. Setting: $y\in\cY^n$ is fixed;
$\Md^n=\{m_1^n,\dots,m_{K_n}^n\}$ is a finite metric family;
$\hat\theta(x,y)=\argmin_{1\le k\le K_n}\PEPm{x}{y}{m_k^n}$ is the
in-hindsight best index (any minimizer; ties broken arbitrarily);
$\cF_n=\{x\in\cX^n:\PEPm{x}{y}{m_{\hat\theta(x,y)}}\ge e^{-n^2}\}$ is the
finite-exponent set of \cref{prop:uc-equalizer}; $\beta\ge0$ is a constant;
and the metric $u$ satisfies
\begin{equation}\label{eq:app-avg-hyp}
  \PEPm{x}{y}{u}\;\le\;e^{-n\beta}\,\PEPm{x}{y}{m_{\hat\theta(x,y)}}
  \qquad\text{for all }x\in\cF_n .
\end{equation}
All expectations are in the fixed-$y^n$/product law: $\tilde X^n\sim Q_X^n$,
drawn independently of $y$. The strategy is to sandwich the
$\cF_n$-restricted averaged tie-as-error PEP of $u$ between a uniformity
lower bound and an upper bound from \eqref{eq:app-avg-hyp}, then solve for
$\beta$.

\begin{IEEEproof}
\emph{Step 1: lower bound by PEP uniformity.} At the fixed $y$,
\cref{prop:pep-uniform} applied to the metric $u$ makes the randomized
(dither-broken) PEP of $u$ uniform on $[0,1]$ under $\tilde X^n\sim Q_X^n$,
so its mean is exactly $\tfrac12$. This holds for \emph{every} metric $u$
--- membership in $\Md^n$ is not required --- and is the uniformity input
the converse needs. The tie-as-error PEP \eqref{eq:uc-pep-tie} dominates the
randomized one pointwise. Restricting to $\cF_n$ via the indicator and using
that the randomized PEP takes values in $[0,1]$ to bound the contribution of
the complement,
\begin{equation}\label{eq:app-avg-lb}
  \Es{\tilde X^n}{\PEPm{\tilde X^n}{y}{u}\,\Ind{\tilde X^n\in\cF_n}}
  \;\ge\;\tfrac12-Q_X^n(\cX^n\setminus\cF_n)
  \;\ge\;\tfrac12-K_ne^{-n^2},
\end{equation}
the last step by \cref{prop:uc-equalizer}
($Q_X^n(\cX^n\setminus\cF_n)\le K_ne^{-n^2}$: uniformity plus a union bound
over the $K_n$ family members).

\emph{Step 2: upper bound from the hypothesis.} On $\cF_n$ we have
$\PEPm{x}{y}{m_{\hat\theta(x,y)}}\le1$, so \eqref{eq:app-avg-hyp} gives
$\PEPm{x}{y}{u}\le e^{-n\beta}$ for every $x\in\cF_n$. Averaging the
indicator-restricted PEP under $\tilde X^n\sim Q_X^n$,
\begin{equation}\label{eq:app-avg-ub}
  \Es{\tilde X^n}{\PEPm{\tilde X^n}{y}{u}\,\Ind{\tilde X^n\in\cF_n}}
  \;\le\;e^{-n\beta}\,Q_X^n(\cF_n)\;\le\;e^{-n\beta}.
\end{equation}

\emph{Step 3: pinch.} Chaining \eqref{eq:app-avg-lb} and
\eqref{eq:app-avg-ub} gives $\tfrac12-K_ne^{-n^2}\le e^{-n\beta}$. Whenever
$\tfrac12-K_ne^{-n^2}>0$ (in particular for all $n$ large under (H-sub),
since then $K_ne^{-n^2}\to0$ super-exponentially), taking $-\tfrac1n\log$
yields
\[
  \beta\;\le\;\frac1n\log\frac{1}{\tfrac12-K_ne^{-n^2}}
  \;=\;\frac{\log2}{n}+\frac1n\log\frac{1}{1-2K_ne^{-n^2}}
  \;=\;\frac{\log2}{n}+o(1/n)\;=\;O(1/n),
\]
which is \eqref{eq:uc-avg-regret}.
\end{IEEEproof}

\section{Proof of the Minimax Theorem}
\label{app:univ-minimax}

We prove \cref{thm:univ-minimax}: the normalizer bracket, achievability, and
the converse (by a two-case dichotomy on $\gamma_n^{\cF_n}$). Setting: $y$
is fixed and $\tilde X^n\sim Q_X^n$ (the fixed-$y^n$/product law);
$\Md^n$ is a finite family with $K_n$ members satisfying (H-sub); $U^n$ is
the clipped merge metric \eqref{eq:uc-U-def}; $\cF_n$ is the finite-exponent
set of \cref{prop:uc-equalizer};
\begin{equation*}
  \gamma_n=\tfrac1n\log\Es{\tilde X^n}{e^{U^n(\tilde X^n,y)}}
  \quad\text{and}\quad
  \gamma_n^{\cF_n}=\tfrac1n\log\Es{\tilde X^n}{e^{U^n(\tilde X^n,y)}\Ind{\tilde X^n\in\cF_n}}
\end{equation*}
are the two normalizers; and $r_{y,n}$ is the regret of \cref{def:uc-regret}.

\begin{IEEEproof}
\emph{Achievability and the normalizer bracket.} Achievability is
\cref{prop:uc-equalizer}: on $\cF_n$ the clip is inactive, so
$e^{-U^n(x,y)}=\PEPm{x}{y}{m_{\hat\theta(x,y)}}$ exactly, and Markov's
inequality on $e^{U^n}$ gives $r_{y,n}(U^n,x)\le\gamma_n$ for every
$x\in\cF_n$.

For the bracket: $\gamma_n^{\cF_n}\le\gamma_n$ since
$\Ind{\tilde X^n\in\cF_n}\le1$, and $\gamma_n\ge0$ since $e^{U^n}\ge1$. For
the upper endpoint, the expectation display in the proof of
\cref{lem:uc-pep-bound-u} gives
$\Es{\tilde X^n}{e^{U^n(\tilde X^n,y)}}\le K_n(1+n^2)$, hence
$\gamma_n\le\tfrac1n\log K_n+O(\log n/n)$.

\emph{Converse setup.} Fix a constant $\epsilon>0$ and suppose, towards a
contradiction, that some metric $u_n^\star$ satisfies
$r_{y,n}(u_n^\star,x)\le\gamma_n^{\cF_n}-\epsilon$ for all $x\in\cF_n$.
Unwinding the regret (\cref{def:uc-regret}),
\begin{equation}\label{eq:app-mm-hyp}
  \PEPm{x}{y}{u_n^\star}\;\le\;
  e^{n(\gamma_n^{\cF_n}-\epsilon)}\,\PEPm{x}{y}{m_{\hat\theta(x,y)}}
  \qquad\text{for all }x\in\cF_n .
\end{equation}

\emph{Case A: $\gamma_n^{\cF_n}<\epsilon$.} Set
$\beta\triangleq\epsilon-\gamma_n^{\cF_n}>0$. Then \eqref{eq:app-mm-hyp}
reads
$\PEPm{x}{y}{u_n^\star}\le e^{-n\beta}\PEPm{x}{y}{m_{\hat\theta(x,y)}}$ for
all $x\in\cF_n$ --- exactly the hypothesis of \cref{lem:uc-avg-regret} ---
so
\[
  \epsilon-\gamma_n^{\cF_n}\;=\;\beta\;\le\;
  \frac1n\log\frac{1}{\tfrac12-K_ne^{-n^2}}\;=\;O(1/n),
\]
i.e.\ $\gamma_n^{\cF_n}\ge\epsilon-O(1/n)$. But the normalizer bracket above
gives $\gamma_n^{\cF_n}\le\gamma_n\le\tfrac1n\log K_n+O(\log n/n)=o(1)$
under (H-sub), which is incompatible with
$\gamma_n^{\cF_n}\ge\epsilon-O(1/n)$ for the fixed constant $\epsilon>0$
once $n$ is large. Hence, for all $n$ large, Case A yields a contradiction.

\emph{Case B: $\gamma_n^{\cF_n}\ge\epsilon$.} Now
$e^{-n(\gamma_n^{\cF_n}-\epsilon)}\le1$. On $\cF_n$ the clip is inactive, so
$e^{-U^n(x,y)}=\PEPm{x}{y}{m_{\hat\theta(x,y)}}$ and \eqref{eq:app-mm-hyp}
becomes
$\PEPm{x}{y}{u_n^\star}\le e^{n(\gamma_n^{\cF_n}-\epsilon)}\,e^{-U^n(x,y)}$.
Reciprocating and pairing with the clip level $e^{n^2}$,
\begin{equation}\label{eq:app-mm-recip}
  \min\bigl\{1/\PEPm{x}{y}{u_n^\star},\,e^{n^2}\bigr\}
  \;\ge\;e^{-n(\gamma_n^{\cF_n}-\epsilon)}\,e^{U^n(x,y)}
  \qquad\text{for }x\in\cF_n ;
\end{equation}
the clip is preserved on the right because
$e^{-n(\gamma_n^{\cF_n}-\epsilon)}\le1$ and $U^n\le n^2$, so the right-hand
side never exceeds $e^{n^2}$. Averaging \eqref{eq:app-mm-recip} under
$\tilde X^n\sim Q_X^n$, restricted to $\cF_n$ (the unrestricted left-hand
average is at least the restricted one, since the integrand is
non-negative), and using the definition of $\gamma_n^{\cF_n}$,
\begin{align*}
  \Es{\tilde X^n}{\min\bigl\{1/\PEPm{\tilde X^n}{y}{u_n^\star},\,e^{n^2}\bigr\}}
  &\ge e^{-n(\gamma_n^{\cF_n}-\epsilon)}\,
    \Es{\tilde X^n}{e^{U^n(\tilde X^n,y)}\,\Ind{\tilde X^n\in\cF_n}}\\
  &=e^{-n(\gamma_n^{\cF_n}-\epsilon)}\cdot e^{n\gamma_n^{\cF_n}}
   \;=\;e^{n\epsilon}.
\end{align*}
But \cref{lem:uc-clipped-inv-pep} --- whose statement holds for any metric
at fixed $y$, family member or not --- bounds the left side by $1+n^2$ at
clip level $e^{n^2}$. Hence $e^{n\epsilon}\le1+n^2$, i.e.\
$\epsilon\le O(\log n/n)$ --- a contradiction for the fixed constant
$\epsilon>0$ at all $n$ large.

Cases A and B exhaust the dichotomy and both exclude the converse hypothesis
for all $n$ large, completing the proof.
\end{IEEEproof}

\section{The Fine-Scale Minimax Converse}
\label{app:minimax-scale}

We derive the fine-scale bound recorded in \cref{rem:uc-minimax-scale}:
\emph{for every $\delta>0$ and all $n$ large, every metric $u$ satisfies}
\[
  \max_{x\in\cF_n}r_{y,n}(u,x)\;\ge\;\gamma_n^{\cF_n}\,-\,(2+\delta)\,
  \frac{\log n}{n}.
\]
The setting and notation are those of Appendix~\ref{app:univ-minimax}.

\begin{IEEEproof}
Write $\rho_n$ for the left-hand side. If $\rho_n\ge\gamma_n^{\cF_n}$ there
is nothing to prove. If $0\le\rho_n<\gamma_n^{\cF_n}$, set
$\epsilon_n\triangleq\gamma_n^{\cF_n}-\rho_n\in(0,\gamma_n^{\cF_n}]$ and run
Case~B of Appendix~\ref{app:univ-minimax} verbatim with the vanishing margin
$\epsilon_n$ in place of the constant $\epsilon$ --- the clip is preserved in
\eqref{eq:app-mm-recip} because $e^{-n(\gamma_n^{\cF_n}-\epsilon_n)}\le1$
and $U^n\le n^2$ --- so the averaging display gives
$e^{n\epsilon_n}\le1+n^2$, i.e.\
$\epsilon_n\le\tfrac1n\log(1+n^2)\le(2+\delta)\tfrac{\log n}{n}$ for all $n$
large. If $\rho_n<0$, the same reciprocation with margin $0$ gives
$\min\{1/\PEPm{x}{y}{u},e^{n^2}\}\ge e^{U^n(x,y)}$ on $\cF_n$; averaging
restricted to $\cF_n$ and applying \cref{lem:uc-clipped-inv-pep} to $u$
yields $e^{n\gamma_n^{\cF_n}}\le1+n^2$, while \cref{lem:uc-avg-regret} with
$\beta=-\rho_n$ gives $\rho_n\ge-\tfrac{\log2}{n}-o(1/n)$; combining,
$\rho_n\ge\gamma_n^{\cF_n}-(2+\delta)\tfrac{\log n}{n}$ for all $n$ large,
since $\delta\log n\ge\log2+o(1)$ eventually.
\end{IEEEproof}

%% file: sections/appendix-parts/04-dmc.tex

\section{Type-Mass Proofs for the DMC Reductions}
\label{app:dmc-typemass}

This appendix collects the type-counting proofs of
\cref{lem:uc-typemass-dom}, \cref{cor:uc-mmi-universal}, and
\cref{thm:uc-tilted}.

\begin{IEEEproof}[Proof of \cref{lem:uc-typemass-dom}]
  Fix $(x^n,y^n)$ and write $T=T_{x^ny^n}$,
  $T_Y=T_{y^n}$. For a joint type $\tilde T$ with marginals $(T_X^\star,T_Y)$
  --- only such types receive prior mass at the fixed $y^n$ --- let
  $P_n(\tilde T)\triangleq Q_X^n\{\tilde x^n:T_{\tilde x^ny^n}=\tilde T\}
   =|\{\tilde x^n\in\cT(T_X^\star):T_{\tilde x^ny^n}=\tilde T\}|/|\cT(T_X^\star)|$.
  Conditional-type and type-class
  counting~\cite[Lems.~2.3, 2.5]{csiszar1981information} give, whenever the
  conditional class is nonempty,
  $(n+1)^{-|\cX||\cY|}e^{nH_{\tilde T}(X\mid Y)}
   \le|\{\tilde x^n:T_{\tilde x^ny^n}=\tilde T\}|\le e^{nH_{\tilde T}(X\mid Y)}$
  and $(n+1)^{-|\cX|}e^{nH(T_X^\star)}\le|\cT(T_X^\star)|\le e^{nH(T_X^\star)}$;
  since $H(T_X^\star)-H_{\tilde T}(X\mid Y)=\Ihat(\tilde T)$ for
  $\tilde T_X=T_X^\star$, this yields the sandwich
  \begin{equation}\label{eq:uc-mmi-typemass}
    (n+1)^{-|\cX||\cY|}\,e^{-n\Ihat(\tilde T)}
    \;\le\;P_n(\tilde T)\;\le\;(n+1)^{|\cX|}\,e^{-n\Ihat(\tilde T)}.
  \end{equation}

  \emph{(a) Tie lower bound.} Every $\tilde x^n$ with $T_{\tilde x^ny^n}=T$
  satisfies $m^n(\tilde x^n,y^n)=m^n(x^n,y^n)$ and is charged in full under the
  tie-as-error rule \eqref{eq:uc-pep-tie}; the class is nonempty ($x^n$ itself
  realizes it), so
  $\PEPm{x^n}{y^n}{m^n}\ge P_n(T)\ge(n+1)^{-|\cX||\cY|}e^{-n\Ihat(T)}$.

  \emph{(b) Matching upper bound via the member $T'=T$.} Write
  $g_T(a,b)\triangleq\log\tfrac{T(a,b)}{T_X^\star(a)T_Y(b)}$, with the
  convention $\log0=-\infty$. Since
  $m_{W_T}(\tilde x^n,y^n)=n\langle T_{\tilde x^ny^n},\log W_T\rangle$ with
  $W_T(b\mid a)=T(a,b)/T_X^\star(a)$, and $m_{W_T}(x^n,y^n)$ is finite
  ($T(a,b)>0$ implies $W_T(b\mid a)>0$), a competitor of joint type $\tilde T$
  enters the PEP event $\{m_{W_T}(\tilde x^n,y^n)\ge m_{W_T}(x^n,y^n)\}$ only
  if $\langle\tilde T,\log W_T\rangle\ge\langle T,\log W_T\rangle$; in
  particular $\tilde T$ must vanish outside the support of $T$ (a cell with
  $\tilde T(a,b)>0=T(a,b)$ forces the metric to $-\infty$). Using
  $\log W_T(b\mid a)=g_T(a,b)+\log T_Y(b)$ and the shared $\cY$-marginal
  $\tilde T_Y=T_Y$, the event forces
  $\langle\tilde T,g_T\rangle\ge\langle T,g_T\rangle=\Ihat(T)$, and the
  splitting identity ($\tilde T\ll T$, shared marginals $(T_X^\star,T_Y)$)
  \[
    \Ihat(\tilde T)\;=\;D\bigl(\tilde T\,\big\|\,T_X^\star\otimes T_Y\bigr)
    \;=\;\langle\tilde T,g_T\rangle+D(\tilde T\,\|\,T)\;\ge\;\Ihat(T)
  \]
  holds on the event. Summing the right half of \eqref{eq:uc-mmi-typemass}
  over the at most $(n+1)^{|\cX||\cY|}$ joint types in the event,
  $\PEPm{x^n}{y^n}{m_{W_T}}\le(n+1)^{|\cX||\cY|+|\cX|}e^{-n\Ihat(T)}$.
  Chaining (b) with (a) gives \eqref{eq:uc-mmi-domination}.
\end{IEEEproof}

\begin{IEEEproof}[Proof of \cref{cor:uc-mmi-universal}]
  The proof is self-contained: it uses only the type-counting domination of
  \cref{lem:uc-typemass-dom}, the merge bound of
  \cref{lem:uc-pep-bound-u}, and the exponent computation of \cref{thm:univ-rc}
  --- no property of the MMI decoder, and no step of the classical
  Csisz\'ar--K\"orner analysis, is imported.

  \emph{Step 1: merge bound and exponent computation.} Substituting
  \eqref{eq:uc-mmi-domination} of \cref{lem:uc-typemass-dom} with $m^n=m_W$
  into \eqref{eq:uc-pep-u-min} of
  \cref{lem:uc-pep-bound-u} gives, at every
  $(x^n,y^n)\in\cT(T_X^\star)\times\cY^n$,
  $\PEPm{x^n}{y^n}{U^n}\le K_n(1+n^2)(n+1)^{2|\cX||\cY|+|\cX|}
   \max\{\PEPm{x^n}{y^n}{m_W},e^{-n^2}\}$, with $K_n=O(n^{|\cX||\cY|-1})$ from
  \cref{lem:uc-dmc-polytypes}. The channel-joint law $Q_X^n\cdot W^n$ is
  supported on $\cT(T_X^\star)\times\cY^n$, so repeating the proof of
  \cref{thm:univ-rc} with the enlarged --- still subexponential --- prefactor
  yields $\Erc(R;\bbW,Q_X^n,U^n)\ge\Erc(R;\bbW,Q_X^n,m_W)$ for every
  $W\in\cP(\cY\mid\cX)$ and every rate $R$. Neither $\Md^n$ nor $U^n$ depends
  on $W$, so a single decoder serves every DMC simultaneously.

  \emph{Step 2: transfer to the MMI decoder.} \Cref{thm:uc-mmi} identifies
  $U^n$ with the MMI metric up to an additive shift $\epsilon_n=o(n)$ of a
  joint-type-based metric, and the shift is exponent-neutral, as follows.
  The PEP event of $U^n$ at $(x^n,y^n)$ is contained
  in $\{n\Ihat(T_{\tilde x^ny^n})\ge n\Ihat(T_{x^ny^n})-2\epsilon_n\}$.
  Summing over the at most $(n+1)^{|\cX||\cY|}$ joint types in this event,
  each carrying class mass $e^{-n\Ihat+o(n)}$ by \eqref{eq:uc-mmi-typemass},
  gives $\PEPm{x^n}{y^n}{U^n}\le e^{-n\Ihat(T_{x^ny^n})+2\epsilon_n+o(n)}$,
  while the tie mass at $T_{x^ny^n}$ gives the matching lower bound
  $e^{-n\Ihat(T_{x^ny^n})-o(n)}$ for both metrics ($U^n$ is itself
  joint-type-based, being a function of the members' type-based PEPs). The two decoders' PEPs
  therefore agree within $e^{\pm o(n)}$ at every $(x^n,y^n)$, they share every
  random-coding exponent, and the MMI decoder inherits the guarantee of
  Step~1.

  For calibration only: under uniform-on-type-class input the matched-ML
  random-coding exponent equals the empirical-mutual-information exponent of
  Csisz\'ar--K\"orner~\cite[Thm.~10.2]{csiszar1981information}, so the common
  value achieved above is the classical MMI exponent; the proof neither uses
  nor re-derives that identity --- the domination
  \eqref{eq:uc-mmi-domination} replaces it.
\end{IEEEproof}

\begin{IEEEproof}[Proof of \cref{thm:uc-tilted}]
  Fix $(x^n,y^n)$ and write $T=T_{x^ny^n}$, $T_Y=T_{y^n}$. For a joint type
  $\tilde T\in\cT_n(\cX\times\cY)$ with $\cY$-marginal $T_Y$, let
  $P_n(\tilde T)\triangleq Q_X^{\otimes n}\{\tilde x^n:T_{\tilde x^ny^n}
  =\tilde T\}$ be the prior mass of its conditional type class. The exact
  per-sequence mass
  $Q_X^{\otimes n}(\tilde x^n)=e^{-n[H(\tilde T_X)+D(\tilde T_X\|Q_X)]}$,
  conditional-type counting~\cite[Lem.~2.5]{csiszar1981information}, and
  identity \eqref{eq:uc-tilt-identity} give the sandwich
  \begin{equation}\label{eq:uc-typemass}
    (n+1)^{-|\cX||\cY|}\,e^{-nD(\tilde T\|Q_X\otimes T_Y)}
    \;\le\;P_n(\tilde T)\;\le\;e^{-nD(\tilde T\|Q_X\otimes T_Y)},
  \end{equation}
  the memoryless counterpart of \eqref{eq:uc-mmi-typemass}. With this
  sandwich in place, the two halves (a)--(b) of the proof of
  \cref{lem:uc-typemass-dom} apply verbatim with
  $g_T(a,b)\triangleq\log\tfrac{T(a,b)}{Q_X(a)T_Y(b)}$, so that
  $\tilde m_T(\tilde x^n,y^n)=n\langle T_{\tilde x^ny^n},g_T\rangle$ and
  $\langle T,g_T\rangle=D(T\|Q_X\otimes T_Y)$: the tie lower bound gives
  $\PEPm{x^n}{y^n}{m^n}\ge P_n(T)\ge(n+1)^{-|\cX||\cY|}
   e^{-nD(T\|Q_X\otimes T_Y)}$ for every joint-type-based $m^n$, and the
  member-$T'=T$ upper bound --- support containment, then the splitting
  identity
  $D(\tilde T\|Q_X\otimes T_Y)=\langle\tilde T,g_T\rangle+D(\tilde T\|T)
   \ge\langle T,g_T\rangle$ on the PEP event --- gives
  $\PEPm{x^n}{y^n}{\tilde m_T}\le(n+1)^{|\cX||\cY|}
   e^{-nD(T\|Q_X\otimes T_Y)}$. Since
  $D(T\|Q_X\otimes T_Y)=\Ihat(T)+D(T_X\|Q_X)\le\log|\cX|
   +\log(1/\min_aQ_X(a))$ is bounded uniformly in $(x^n,y^n)$, the $n^2$ clip
  is inactive for all $n$ large, and (i) follows by
  \eqref{eq:uc-tilt-identity}; chaining the two displays gives (ii).
\end{IEEEproof}

\section{The FSM Ranking Count: Proof of \texorpdfstring{\cref{lem:uc-fsm-polycount}}{the FSM Counting Lemma}}
\label{app:fsm-polycount}

We prove \cref{lem:uc-fsm-polycount}. Recall from \cref{subsec:uc-dmc-fsc}
that an FSM metric evaluates as
$m_{(g,q,s_0)}(x^n,y^n)=\langle q,\,n\,T^g_{x^ny^n}\rangle$, where
$T^g_{x^ny^n}\in\mR^{D}$ is the empirical distribution of the triples
$(s_{i-1},x_i,y_i)$ along the run of $g$ from $s_0$ and
$D=|\cS_n||\cX||\cY|$; $\PEPm{x^n}{y^n}{\cdot}$ is the tie-as-error PEP
\eqref{eq:uc-pep-tie}.

\begin{IEEEproof}
Fix the next-state map $g$ and initial state $s_0$; there are at most
$|\cS_n|^{D}$ maps $g$ and $\le|\cS_n|$ initial states. Given $(g,s_0)$, the
metric is a \emph{linear} functional of the triple type $T^g_{x^ny^n}$,
parametrized by $q\in\mR^{D}$, so the tie-as-error PEP depends on $q$ only
through the \emph{ranking} (total preorder, ties included) that $q$ induces
on the $N\le(n+1)^{D}$ achievable triple types --- i.e.\ through the sign
pattern of $\langle q,T-T'\rangle$ over the
$H\le\binom{N}{2}\le(n+1)^{2D}$ difference hyperplanes in $\mR^{D}$. A
ranking is thus a \emph{face} of the hyperplane arrangement, not merely one
of its full-dimensional regions, so the elementary region bound
$\sum_{i=0}^{D}\binom{H}{i}$ of Buck~\cite{buck1943partition} (see
also~\cite{zaslavsky1975facing})
cannot be quoted directly; we convert it into a face count flat by flat.
Every face lies on a unique flat of the arrangement (the intersection of the
hyperplanes on which its sign vector vanishes; a hyperplane containing that
flat necessarily belongs to the vanishing set) and is contained in a single
region of the arrangement that the remaining hyperplanes induce on the flat,
with distinct faces on the same flat landing in distinct regions. Every flat
is an intersection of at most $D$ linearly independent members, so there are
at most $\sum_{k=0}^{D}\binom{H}{k}$ flats, and each flat --- of dimension at
most $D$, cut by at most $H$ induced hyperplanes --- carries at most
$\sum_{i=0}^{D}\binom{H}{i}$ regions by the region bound applied within it.
Hence, for $H\ge1$ (the case $H=0$ having a single ranking), the number of
faces, and so of distinct rankings, is at most
\[
  \Bigl[\sum_{i=0}^{D}\binom{H}{i}\Bigr]^{2}
  \;\le\;\bigl[(D+1)H^{D}\bigr]^{2}
  \;\le\;(n+1)^{4D^{2}+2D},
\]
using $H\le(n+1)^{2D}$ and $(D+1)^{2}\le(n+1)^{2D}$; the tie factor thus
costs at most a subexponential slack ($e^{O(n^{2\alpha}\log n)}$ for
$|\cS_n|=n^\alpha$) over the naive region count and does not move the
$\alpha<\tfrac12$ threshold below. Choosing one representative $q$ per face
gives a metric whose PEP \emph{equals} that of every metric in its face at
every $(x^n,y^n)$ --- the pointwise-domination clause --- and multiplying
over the $\le|\cS_n|^{D}$ next-state maps and $\le|\cS_n|$ initial states,
$|\Md^n|\le|\cS_n|^{D+1}(n+1)^{4D^2+2D}$. Summing rates, for
$|\cS_n|=n^\alpha$,
\[
  \tfrac1n\log|\Md^n|
  \;\le\;\tfrac{(D+1)\log|\cS_n|}{n}+\tfrac{(4D^2+2D)\log(n+1)}{n}
  \;=\;O\bigl(n^{2\alpha-1}\log n\bigr)\;\to\;0
  \qquad\text{for }\alpha<\tfrac12.
\]
Counting instead one representative per \emph{joint type} (the cruder
$(n+1)^{D}$ bound) would not in general yield the pointwise-domination
clause, since metrics sharing a per-coordinate type can rank whole sequences
differently; the sharper $\alpha<1$
of~\cite[Sec.~V.C]{elkayam2014universal} uses a dedicated construction.
\end{IEEEproof}

%% file: sections/appendix-parts/05-erasure.tex

\section{Derandomization of the Universal Erasure Decoder}
\label{app:det-erasure}

This appendix proves \cref{thm:uc-det-erasure}. The hypotheses are: the
channel family $\Cset$ satisfies (H-sub) (\cref{def:uc-hsub}), and the metric
sequence $\bbU$ satisfies (H-univ($R'$)) for every
$R'\in[R,R+\gamma_{\max}]$, for a fixed $\gamma_{\max}\ge0$. The decoder
$\cD_n^\gamma$ is the PEP-based erasure rule of \cref{def:uc-pep-erasure}
with metric $U^n$ and tie-as-error scores
$-\log\PEPm{x_i^n}{y^n}{U^n}$ (the argmin tie, if any, broken by lowest
index --- a choice that does not depend on $\gamma$). For a channel
$W^n\in\Cset^n$ and a rate parameter $\rho>0$, write
\begin{equation}\label{eq:app-det-B}
  B_n^{W}(\rho)\;\triangleq\;
  \E{\min\BRAs{1,\,e^{n\rho}\,\PEPm{X^n}{Y^n}{U^n}}},
  \qquad (X^n,Y^n)\sim Q_X^n\cdot W^n,
\end{equation}
the per-blocklength functional whose exponent is
$\Erc(\rho;\bbW,\bbQ_X,\bbU)$ by \cref{prop:uc-erc-pep}.

\emph{Standing per-$n$ bounds.} For every $W^n$, every $\gamma\ge0$, and the
i.i.d.\ rate-$R$ ensemble,
\begin{equation}\label{eq:app-det-rc-bounds}
  \E{P_{\mathrm{era}}(C^n,W^n,\cD_n^\gamma)}\le B_n^{W}(R+\gamma),
  \qquad
  \E{P_{\mathrm{und}}(C^n,W^n,\cD_n^\gamma)}\le e^{-n\gamma}B_n^{W}(R).
\end{equation}
These are the bounds of \cref{lem:uc-era-oneshot} transferred to the
tie-as-error scores via \cref{lem:uc-tie-dither}, exactly as in the proof of
\cref{thm:uc-erasure-via-rc}: the competitor tail survives as ``$\le$'' and
the transmitted-word spectrum is by definition the tie-as-error spectrum
entering \eqref{eq:app-det-B}.

\begin{IEEEproof}[Proof of \cref{thm:uc-det-erasure}]
\emph{Grid and expurgation.} Fix $n$ and let
$G_n\triangleq\bigl(\{0,\tfrac1n,\tfrac2n,\dots\}\cap[0,\gamma_{\max}]\bigr)
\cup\{\gamma_{\max}\}$, so that every $\gamma\in[0,\gamma_{\max}]$ has grid
neighbors $\gamma^-\le\gamma\le\gamma^+$ with $\gamma^+-\gamma^-\le1/n$ and
$|G_n|\le\gamma_{\max}n+2$. Index the $N_n\triangleq2\,|\Cset^n|\,|G_n|$
constraints by a channel, a grid margin, and an event
$\bullet\in\{\mathrm{era},\mathrm{und}\}$; the decoder is built from $\bbU$
alone, so $\Mset$ indexes no constraint. Markov's inequality per constraint
(at level $2N_n$ times the mean \eqref{eq:app-det-rc-bounds}) and a union
bound leave at least half the codebooks satisfying
\begin{equation}\label{eq:app-det-grid}
  P_{\mathrm{era}}(C^n,W^n,\cD_n^{\gamma_i})\le2N_n\,B_n^{W}(R+\gamma_i),
  \qquad
  P_{\mathrm{und}}(C^n,W^n,\cD_n^{\gamma_i})\le2N_n\,e^{-n\gamma_i}B_n^{W}(R)
\end{equation}
for every $W^n\in\Cset^n$ and $\gamma_i\in G_n$ simultaneously; pick one.
The penalty $2N_n$ is subexponential by (H-sub).

\emph{Monotonicity and the continuum.} Pathwise, raising the margin only
strengthens both erase clauses of \eqref{eq:uc-erase-cond} (the confidence
bar $n(R+\gamma)$ and the separation bar $n\gamma$ both rise), and the best
index does not depend on the margin; hence, for the fixed code,
$P_{\mathrm{era}}(\gamma)$ is non-decreasing and $P_{\mathrm{und}}(\gamma)$
non-increasing in $\gamma$. Sandwiching $\gamma$ between its grid neighbors
($P_{\mathrm{era}}(\gamma)\le P_{\mathrm{era}}(\gamma^+)$,
$P_{\mathrm{und}}(\gamma)\le P_{\mathrm{und}}(\gamma^-)$) and using
$\min\{1,e^{n\rho'}p\}\le e^{n(\rho'-\rho)}\min\{1,e^{n\rho}p\}$ for
$\rho'\ge\rho$ to trade $\gamma^\pm$ for $\gamma$,
\begin{equation*}
  P_{\mathrm{era}}(C^n,W^n,\cD_n^{\gamma})\le2eN_n\,B_n^{W}(R+\gamma),
  \qquad
  P_{\mathrm{und}}(C^n,W^n,\cD_n^{\gamma})\le2eN_n\,e^{-n\gamma}B_n^{W}(R),
\end{equation*}
one factor $e$ per block, exponent-neutral.

\emph{Exponent passage.} Assemble the codebooks over $n$ into $\bbC$; taking
$-\tfrac1n\log$ and the $\liminf$, the subexponential factors vanish, and
\cref{prop:uc-erc-pep} with $\Ehat(e^{-n\gamma}a_n)=\gamma+\Ehat(a_n)$ gives
$\Eera\ge\Erc(R+\gamma;\bbU)$ and $\Eund\ge\gamma+\Erc(R;\bbU)$ for every
$\gamma\in[0,\gamma_{\max}]$ and $\bbW\in\Cset$ simultaneously. Only then
invoke (H-univ($R'$)) at $R'=R+\gamma$ and $R'=R$ to pass from $\bbU$ to
every $\bbm\in\Mset$, yielding
\eqref{eq:uc-det-era}--\eqref{eq:uc-det-und}.
\end{IEEEproof}

%% file: sections/appendix-parts/06-awgn.tex

\section{The Angular-Correlation Tail on the Shell}
\label{app:shell-f}

This appendix proves the closed form~\eqref{eq:awgn-f-form} used throughout
\cref{sec:awgn}, in its general spherically-symmetric form.

\begin{lemma}[The pairwise error integral]
\label{lem:awgn-shell-f}
Let $y\in\mR^n$ be fixed and nonzero, and let $X^n\in\mR^n$ be a random
vector whose law is spherically symmetric with $\PR{X^n=0}=0$. Let
$\hat\rho\triangleq\BRAi{X^n,y}/(\norm{X^n}\norm{y})$ be the (random)
empirical correlation. Then for $t\in(0,1)$,
\begin{align}
  \PR{\hat\rho\ge t}
  &\;=\;\frac{1}{\pi}\,(1-t^2)^{(n-1)/2}
    \int_0^{\pi/2}\!\bigl(1+t^2\tan^2\phi\bigr)^{(1-n)/2}\,d\phi
    \label{eq:app-shell-f-form}\\
  &\;=\;\frac{1}{\pi}\,(1-t^2)^{(n-1)/2}
    \int_0^{\infty}\!\frac{1}{1+x^2}\,\bigl(1+t^2x^2\bigr)^{(1-n)/2}\,dx.
    \notag
\end{align}
\end{lemma}

The uniform-shell input $Q_X^n=\Unif{\sqrt{nP}\,\mS^{n-1}}$ satisfies the
hypotheses, and the law of $\hat\rho$ is the same for every nonzero direction
$y$; \eqref{eq:app-shell-f-form} is therefore the function
$F(t)=\PR{\hat\rho\ge t}$ of~\eqref{eq:awgn-f-form}. The identity
\eqref{eq:awgn-pep-via-f} follows: for a metric $m(x,y)=\BRAi{x,v(y)}$ with
$v(y)\neq0$ and $\tilde X^n\sim Q_X^n$, the event
$\BRAi{\tilde X^n,v(y)}\ge\BRAi{x,v(y)}$ coincides, after dividing by the
shell-constant $\norm{\tilde X^n}\,\norm{v(y)}=\norm{x}\,\norm{v(y)}$, with
$\{\hat\rho\ge t(x,y)\}$ for the lemma applied with the fixed nonzero vector
$v(y)$, so $\PEP{x}{y}_m=F(t(x,y))$.

\begin{IEEEproof}
Following \cite[Lemma~4]{lomnitz2011communication}: spherical symmetry lets
us take $X^n\sim\cN(0,I_n)$ and rotate $y$ to $(1,0,\dots,0)$, so with
$X_2^n\triangleq(X_2,\dots,X_n)$,
\[
  \PR{\hat\rho\ge t}
  \;=\;\tfrac12\,\PR{X_1^2\ge\tfrac{t^2}{1-t^2}\norm{X_2^n}^2}
  \;=\;\Es{X_2^n}{Q\!\BRA{\sqrt{\tfrac{t^2}{1-t^2}\norm{X_2^n}^2}}}.
\]
Craig's formula~\cite{craig1991new}
$Q(x)=\frac1\pi\int_0^{\pi/2}e^{-x^2/(2\sin^2\phi)}\,d\phi$ and
the Gaussian integral
$\Es{X_2^n}{e^{-\frac12c\norm{X_2^n}^2}}=(1+c)^{-(n-1)/2}$, applied under
the $\phi$-integral at
$c=t^2/\bigl((1-t^2)\sin^2\phi\bigr)$, with the identity
$1+c=(1+t^2\cot^2\phi)/(1-t^2)$ and the substitution
$\phi\mapsto\pi/2-\phi$, give the first form
of~\eqref{eq:app-shell-f-form}; $x=\tan\phi$ gives the second.
\end{IEEEproof}

\section{Proof of \texorpdfstring{\cref{lem:awgn-cosine}}{the Cosine-Lipschitz Lemma}}
\label{app:cosine-lip}

We prove \cref{lem:awgn-cosine}. Throughout, $y\in\cE_n^Y$ is fixed --- (H1)
and (H2) are assumed only there --- and we suppress it from $v_\theta$. Write
the cosine as an inner product of unit vectors:
\[
  t_\theta(x,y)\;=\;\BRAi{\hat x,\hat v_\theta},
  \qquad
  \hat x\triangleq\frac{x}{\norm{x}},\quad
  \hat v_\theta\triangleq\frac{v_\theta}{\norm{v_\theta}}.
\]
Under (H2), $\norm{v_\theta}\ge c_0\sqrt n>0$ for $y\in\cE_n^Y$, so
$\hat v_\theta$ is well-defined; $\norm{x}>0$ by hypothesis.

\begin{IEEEproof}
Cauchy--Schwarz with $\norm{\hat x}=1$ gives
$\abs{t_\theta(x,y)-t_{\theta'}(x,y)}\le
\norm{\hat v_\theta-\hat v_{\theta'}}$, and the (asymmetric) split
\begin{equation}\label{eq:app-cosine-split}
  \hat v_\theta-\hat v_{\theta'}
  \;=\;\frac{v_\theta-v_{\theta'}}{\norm{v_\theta}}
  \;+\;v_{\theta'}\cdot
  \frac{\norm{v_{\theta'}}-\norm{v_\theta}}{\norm{v_\theta}\,\norm{v_{\theta'}}},
\end{equation}
with the reverse triangle inequality in the second term (whose factor
$\norm{v_{\theta'}}$ cancels between numerator and denominator), gives
$\norm{\hat v_\theta-\hat v_{\theta'}}\le
2\norm{v_\theta-v_{\theta'}}/\norm{v_\theta}$: only $\norm{v_\theta}$
survives as a lower-bound requirement, which is why (H2) is needed only as
a lower bound (cf.\ \cref{rem:awgn-onlylower}). Bounding the numerator by
(H1) and the denominator by (H2) --- both available since $y\in\cE_n^Y$ ---
the two $\sqrt n$ factors cancel exactly:
\[
  \norm{\hat v_\theta-\hat v_{\theta'}}
  \;\le\;\frac{2\cdot K\sqrt n\,\norm{\theta-\theta'}}{c_0\sqrt n}
  \;=\;\frac{2K}{c_0}\,\norm{\theta-\theta'},
\]
which is~\eqref{eq:awgn-t-lip} with $L_t=2K/c_0$.
\end{IEEEproof}

\section{Proof of \texorpdfstring{\cref{thm:awgn-glrt}}{the Continuous-GLRT Theorem}, Steps 1--3}
\label{app:glrt-steps}

We supply Steps~1--3 of the proof of \cref{thm:awgn-glrt}. They retrace the
merge proof of \cref{thm:univ-rc}, with the grid family replaced by the
continuum $\Theta$ and the clipped-inverse-PEP bound replaced by a
combination of the discretization theorem (\cref{thm:awgn-disc}) and the grid
version of the same bound; throughout, the typical-set-restricted form
$\bar m^n_\cE$ of~\eqref{eq:awgn-barm-cE} is the object analyzed. The three
steps establish the PEP bound~\eqref{eq:awgn-step3-pep} on $\cE_n$ that
Step~4 --- given in the main text, where it produces the
$\min\{\cdot,c-R,\delta\}$ cap --- consumes. We abbreviate
$p_{e,\theta}(x,y)\triangleq\PEP{x}{y}_{m_\theta}$.

\begin{IEEEproof}[Proof of Steps 1--3]
\emph{Step 1: uniform expectation bound on $e^{\bar m^n_\cE}$.} We bound
$\E{e^{\bar m^n_\cE(\tilde X^n,y)}}$ uniformly in $y$, where
$\tilde X^n\sim Q_X^n$. On $\{(\tilde X^n,y)\in\cE_n\}$ the discretization
theorem (eq.~\eqref{eq:awgn-pep-equiv}, with $\eta=1/n$, mesh-admissible
for all $n$ large) gives
\[
  \inf_{\theta\in\Theta}p_{e,\theta}(\tilde x,y)
  \;\ge\;e^{-Cn\eta}\min_{\theta_d\in\Theta_d^n}p_{e,\theta_d}(\tilde x,y),
\]
whence
\[
  e^{\bar m^n_\cE(\tilde x,y)}\;=\;e^{\bar m^n(\tilde x,y)}
  \;=\;\min\!\BRAs{1/\inf_\theta p_{e,\theta}(\tilde x,y),\;e^{cn}}
  \;\le\;e^{Cn\eta}\,e^{U^n_{\mathrm{grid}}(\tilde x,y)},
\]
where
$U^n_{\mathrm{grid}}(x,y)\triangleq
\min\{-\log\min_{\theta_d\in\Theta_d^n}p_{e,\theta_d}(x,y),\;cn\}$
is the grid-merge metric clipped at the same level $cn$. On
$\{(\tilde X^n,y)\notin\cE_n\}$ the typical-set restriction gives
$e^{\bar m^n_\cE(\tilde x,y)}=e^{-cn}$, which contributes at most $e^{-cn}$
to the expectation. Therefore
\[
  \E{e^{\bar m^n_\cE(\tilde X^n,y)}}
  \;\le\;e^{Cn\eta}\,\E{e^{U^n_{\mathrm{grid}}(\tilde X^n,y)}}+e^{-cn}.
\]
Since
\[
  e^{U^n_{\mathrm{grid}}}
  \;=\;\max_{\theta_d\in\Theta_d^n}\min\!\BRAs{p_{e,\theta_d}^{-1},\,e^{cn}}
  \;\le\;\sum_{\theta_d\in\Theta_d^n}\min\!\BRAs{p_{e,\theta_d}^{-1},\,e^{cn}},
\]
\cref{lem:uc-clipped-inv-pep} with clip level $\alpha_n=e^{cn}$ bounds each
summand's expectation by $1+cn$, whence
$\E{e^{U^n_{\mathrm{grid}}(\tilde X^n,y)}}\le\abs{\Theta_d^n}(1+cn)
=O(n^{d+1})$. With $\eta=1/n$ the prefactor $e^{Cn\eta}=e^C$ is constant and
the residual $e^{-cn}$ exponentially small, so
\[
  \E{e^{\bar m^n_\cE(\tilde X^n,y)}}\;\le\;B_n\;\triangleq\;
  e^C\cdot O\bigl(n^{d+1}\bigr)
\]
\emph{uniformly in $y$}, with $B_n$ subexponential in $n$. No channel-joint
statement is invoked.

\emph{Step 2: pointwise bound on $e^{-\bar m^n_\cE(x,y)}$ for legitimate
pairs.} Fix an arbitrary $\theta'\in\Theta$. For $(x,y)\in\cE_n$ the
typical-set restriction is vacuous and, since the infimum over $\Theta$ is
bounded above by its value at the fixed $\theta'$,
\[
  e^{-\bar m^n_\cE(x,y)}\;=\;e^{-\bar m^n(x,y)}
  \;=\;\max\!\BRAs{\inf_{\theta\in\Theta}p_{e,\theta}(x,y),\;e^{-cn}}
  \;\le\;\max\!\BRAs{p_{e,\theta'}(x,y),\;e^{-cn}}.
\]
The choice $\theta'=\theta_*$ recovers the matched-member bound; the freedom
to take an \emph{arbitrary} fixed $\theta'$ is what delivers the
best-in-class ($\sup_{\theta'}$) form of the theorem, since Steps~3--4 carry
$\theta'$ unchanged and the supremum is taken only at the end, over the
resulting family of lower bounds. For $(x,y)\notin\cE_n$ we have
$e^{-\bar m^n_\cE(x,y)}=e^{cn}$, but this branch is split out separately in
Step~4.

\emph{Step 3: PEP bound by Markov, on $\cE_n$.} For a legitimate pair
$(x,y)\in\cE_n$ and the fixed $\theta'\in\Theta$ of Step~2,
\begin{align*}
  \PEP{x}{y}_{\bar m^n_\cE}
  &\;=\;\PRs{\tilde X^n}{e^{\bar m^n_\cE(\tilde X^n,y)}
    \ge e^{\bar m^n_\cE(x,y)}}\\
  &\;\le\;\E{e^{\bar m^n_\cE(\tilde X^n,y)}}\,e^{-\bar m^n_\cE(x,y)}
  \;\le\;B_n\cdot\max\!\BRAs{p_{e,\theta'}(x,y),\;e^{-cn}},
\end{align*}
which is the bound~\eqref{eq:awgn-step3-pep} consumed by Step~4 in the main
text.
\end{IEEEproof}

\section{Proof of \texorpdfstring{\cref{lem:awgn-interf}}{the Interference Typicality Lemma}}
\label{app:interf-typical}

We prove \cref{lem:awgn-interf}. Throughout, the true channel is
$Y^n=hX^n+S(\theta_*)+Z^n$ with $X^n\sim Q_X^n=\Unif{\sqrt{nP}\,\mS^{n-1}}$
and $Z^n\sim\cN(0,\sigma^2I_n)$ independent, where $(h,\sigma,\theta_*)$ is
an arbitrary point of the SNR-constrained family of
\cref{sec:awgn-interference}; every constant below depends only on the family
constants listed in the lemma, never on $(h,\sigma,\theta_*)$ or on $n$.
Write $V\triangleq\mathrm{span}\{\psi_1,\dots,\psi_d\}$,
$d'\triangleq\dim V\le d$, $P_{V^\perp}$ for the orthogonal projection onto
$V^\perp$, and recall $v_\theta(y)=y-S(\theta)$ with $S(\theta)\in V$ and
$\norm{S(\theta)}\le S_{\max}\sqrt n$ for all $\theta\in\Theta$;
$t_\theta(x,y)$ is the cosine~\eqref{eq:awgn-cosine} and $K$ the (H1)
constant of the family:
$\norm{v_\theta(y)-v_{\theta'}(y)}=\norm{S(\theta')-S(\theta)}
\le K\sqrt n\,\norm{\theta-\theta'}$ for \emph{every} $y$, by the basis
expansion.

We use two elementary tail estimates. First, the $\chi^2$ Chernoff bounds:
for $W\sim\chi^2_m$,
\begin{equation}\label{eq:app-chi2}
  \PR{W\le m/2}\;\le\;\bigl(\tfrac12e^{1/2}\bigr)^{m/2}\;\le\;e^{-m/12},
  \qquad
  \PR{W\ge2m}\;\le\;\bigl(2e^{-1}\bigr)^{m/2}\;\le\;e^{-m/8},
\end{equation}
obtained by optimizing the moment generating function
$\E{e^{\lambda W}}=(1-2\lambda)^{-m/2}$ over $\lambda<0$ (resp.\
$0<\lambda<1/2$). Second, a noncentrality-domination lemma, which the
E\&NN verification of Appendix~\ref{app:isi-typical} also invokes.

\begin{lemma}[Noncentral norms dominate central ones]
\label{lem:awgn-noncentral-dom}
Let $\xi\sim\cN(0,\Sigma)$ on an $m$-dimensional Euclidean space with
$\Sigma\succeq s^2I$, and let $\mu$ be a fixed vector. Then
$\norm{\mu+\xi}^2$ stochastically dominates $s^2\chi^2_m$: for every
$r\ge0$, $\PR{\norm{\mu+\xi}^2\le r}\le\PR{s^2\chi^2_m\le r}$.
\end{lemma}

\begin{IEEEproof}
Work in an orthonormal eigenbasis of $\Sigma$, so the coordinates
$\xi_i\sim\cN(0,\lambda_i)$ are independent with $\lambda_i\ge s^2$, and
$\norm{\mu+\xi}^2=\sum_i(\mu_i+\xi_i)^2$. It suffices to show that
$(\mu_i+\xi_i)^2$ stochastically dominates $\xi_i^2$ for each $i$; then, by
independence, $\norm{\mu+\xi}^2$ dominates
$\sum_i\xi_i^2=\sum_i\lambda_ig_i^2\ge s^2\sum_ig_i^2$ with $g_i$ i.i.d.\
standard normal, which is the claim. Fix a coordinate, drop the index, let
$f$ denote the centered Gaussian density of $\xi$, and take $\mu\ge0$ without
loss of generality (symmetry). For $a>0$,
\begin{align*}
  \PR{(\mu+\xi)^2>a}-\PR{\xi^2>a}
  &\;=\;\int_{\sqrt a-\mu}^{\sqrt a}f(u)\,du
    -\int_{-\sqrt a-\mu}^{-\sqrt a}f(u)\,du\\
  &\;=\;\int_{\sqrt a-\mu}^{\sqrt a}\bigl[f(u)-f(2\sqrt a-u)\bigr]\,du,
\end{align*}
where the second step uses the symmetry $f(u)=f(-u)$ and the substitution
$u\mapsto2\sqrt a-u$. For $u\le\sqrt a$ one has
$\abs{u}\le\abs{2\sqrt a-u}$, hence $f(u)\ge f(2\sqrt a-u)$ since $f$ is
symmetric and unimodal, and the integrand is non-negative.
\end{IEEEproof}

The mechanism is a reflection: shifting a symmetric unimodal density by
$\mu$ moves mass away from the origin, so each coordinate of the shifted
vector dominates its centered counterpart, and independence sums the
dominations. Operationally: noise plus any deterministic offset cannot
concentrate near zero more than the noise alone --- the fact the (H2)
typicality events consume.

\begin{IEEEproof}[Proof of \cref{lem:awgn-interf}]
\emph{Part (i): (H2) by projection.} Since $S(\theta)\in V$ for every
$\theta\in\Theta$, we have $P_{V^\perp}S(\theta)=0$, hence for every
$y\in\mR^n$ and every $\theta$,
\begin{equation}\label{eq:app-interf-proj}
  \norm{v_\theta(y)}\;=\;\norm{y-S(\theta)}
  \;\ge\;\norm{P_{V^\perp}(y-S(\theta))}\;=\;\norm{P_{V^\perp}y}.
\end{equation}
The right-hand side is $\theta$-free, so on
$\cE_n^Y=\{\norm{P_{V^\perp}y}\ge(\sigma_{\min}/2)\sqrt n\}$ the bound
$\norm{v_\theta(y)}\ge(\sigma_{\min}/2)\sqrt n$ holds \emph{simultaneously
for all} $\theta\in\Theta$: (H2) holds with $c_0=\sigma_{\min}/2$. For the
probability of $\cE_n^Y$ under the channel: $P_{V^\perp}S(\theta_*)=0$ as
well, so $P_{V^\perp}Y^n=h\,P_{V^\perp}X^n+P_{V^\perp}Z^n$. Conditionally on
$X^n=x$, this is a Gaussian vector on the $m$-dimensional space $V^\perp$,
$m=n-d'\ge n-d$, with mean $h\,P_{V^\perp}x$ and covariance $\sigma^2I$
(restricted to $V^\perp$). By \cref{lem:awgn-noncentral-dom},
$\norm{P_{V^\perp}Y^n}^2$ conditionally dominates $\sigma^2\chi^2_m$, so
by~\eqref{eq:app-chi2},
$\PR{\norm{P_{V^\perp}Y^n}^2\le\sigma^2m/2\mid X^n=x}\le e^{-m/12}$
uniformly in $x$; averaging over $X^n$ removes the conditioning. For
$n\ge2d$ we have $\sigma^2m/2\ge\sigma_{\min}^2n/4$, whence
\begin{equation}\label{eq:app-interf-h2-prob}
  \PR{Y^n\notin\cE_n^Y}\;\le\;e^{-(n-d)/12}.
\end{equation}

\emph{Part (ii): (H3) by one more projection (upper band) and an $\eta$-net
(lower band).} Define the constant
$C_v\triangleq h_{\max}\sqrt P+2S_{\max}+\sqrt2\,\sigma_{\max}$ and the band
values
\begin{equation}\label{eq:app-interf-bands}
  \beta\;\triangleq\;\frac{\sigma_{\min}^2}{4C_v^2},\qquad
  t_+\;\triangleq\;\sqrt{1-\beta},\qquad
  \zeta\;\triangleq\;\frac{h_{\min}\sqrt P}{4C_v},\qquad
  t_-\;\triangleq\;\zeta,
\end{equation}
and note $t_-^2+\beta\le1/8<1$ (using
$C_v^2\ge h_{\max}^2P+2\sigma_{\max}^2$), so $0<t_-<t_+<1$ as required. Set
$\epsilon\triangleq h_{\min}P/4$. Let $L_t=2K/c_0=4K/\sigma_{\min}$ be the
cosine-Lipschitz constant of \cref{lem:awgn-cosine}, whose hypotheses (H1)
(global, by the basis expansion) and (H2) (part~(i), already established ---
the order matters, cf.\ \cref{rem:awgn-typical}) are in force; set the
\emph{constant} mesh $\eta\triangleq\zeta/L_t$ and let
$\Theta_\eta=\{\theta_1,\dots,\theta_N\}\subset\Theta$ be an $\eta$-net,
$N\le(D\sqrt d/\eta)^d=O(1)$. (If $K=0$ then $v_\theta$ does not depend on
$\theta$ and the extension step below is vacuous.)

Define the typical event \emph{intrinsically}, as the maximal set with the
(H2)/(H3) properties:
\begin{equation}\label{eq:app-interf-cEn}
  \cE_n\;\triangleq\;\BRAs{(x,y):\;\norm{x}^2=nP,\;\;y\in\cE_n^Y,\;\;
  t_\theta(x,y)\in[t_-,t_+]\text{ for all }\theta\in\Theta}.
\end{equation}
Then $\cE_n\subset\cX^n\times\cE_n^Y$ and the band statement of (H3) holds
on $\cE_n$ \emph{by construction}, for all $\theta$ simultaneously; the
entire content is the probability bound
$\PR{(X^n,Y^n)\notin\cE_n}\le e^{-n\delta}$, which we prove by exhibiting
concentration events whose intersection forces $(X^n,Y^n)\in\cE_n$. (The
events below depend on the true $(h,\sigma,\theta_*)$; the \emph{set}
$\cE_n$ does not.) Write $\Delta(\theta)\triangleq S(\theta_*)-S(\theta)\in
V$, so that $v_\theta(Y^n)=hX^n+\Delta(\theta)+Z^n$ and
$\norm{\Delta(\theta)}\le2S_{\max}\sqrt n$; abbreviate
$\Delta_i\triangleq\Delta(\theta_i)$. Let
$W_x\triangleq(V+\mathrm{span}\{x\})^\perp$, of dimension at least $n-d-1$.
Consider
\begin{align*}
  \cA_0&\;\triangleq\;\BRAs{\norm{Z^n}^2\le2n\sigma^2},&
  \cA_1&\;\triangleq\;\BRAs{\abs{\BRAi{X^n,Z^n}}\le\epsilon n},\\
  \cA_2&\;\triangleq\;\BRAs{\norm{P_{W_{X^n}}Z^n}^2\ge\sigma_{\min}^2n/4},&
  \cA_3&\;\triangleq\;\textstyle\bigcap_{i=1}^N
    \BRAs{\abs{\BRAi{X^n,\Delta_i}}\le\epsilon n},
\end{align*}
together with $\cA_4\triangleq\{Y^n\in\cE_n^Y\}$. (In fact
$\cA_2\subset\cA_4$: $W_{X^n}\subset V^\perp$ and $P_{W_{X^n}}$ annihilates
both the signal and the interference, so
$\norm{P_{V^\perp}Y^n}\ge\norm{P_{W_{X^n}}Z^n}$; we keep $\cA_4$ explicit
for the E\&NN parallel of Appendix~\ref{app:isi-typical}.)

\emph{Tails.} $\PR{\cA_0^c}\le e^{-n/8}$ by~\eqref{eq:app-chi2}.
Conditionally on $X^n=x$ (with $\norm{x}=\sqrt{nP}$),
$\BRAi{x,Z^n}\sim\cN(0,nP\sigma^2)$, so
$\PR{\cA_1^c}\le2e^{-\epsilon^2n/(2P\sigma_{\max}^2)}$. Conditionally on
$X^n=x$, $P_{W_x}Z^n$ is a centered Gaussian with covariance $\sigma^2I$ on
$W_x$, of dimension $m'\ge n-d-1$, so by~\eqref{eq:app-chi2} and
$\sigma^2m'/2\ge\sigma_{\min}^2n/4$ for $n\ge2(d+1)$,
$\PR{\cA_2^c}\le e^{-(n-d-1)/12}$, uniformly in $x$. For $\cA_3$: for each
$i$ with $\Delta_i\ne0$, by rotational invariance
$\BRAi{X^n,\Delta_i}/(\sqrt{nP}\,\norm{\Delta_i})$ is distributed as the
angular correlation $\hat\rho$ of a fixed direction, whose closed-form tail
(\cref{lem:awgn-shell-f}, with $I_n(t)\le\pi/2$) gives
$\PR{\abs{\hat\rho}\ge t}=2F(t)\le(1-t^2)^{(n-1)/2}\le e^{-(n-1)t^2/2}$;
with $t=\epsilon n/(\sqrt{nP}\,\norm{\Delta_i})\ge\epsilon/(2\sqrt
PS_{\max})$ (for $t\ge1$ the probability vanishes and the bound is
trivial) this gives
$\PR{\cA_3^c}\le N\,e^{-(n-1)\epsilon^2/(8PS_{\max}^2)}$. Finally
$\PR{\cA_4^c}\le e^{-(n-d)/12}$ by~\eqref{eq:app-interf-h2-prob}. All
exponents are uniform over the family, and $N=O(1)$, so there are $\delta>0$
and $n_0$, depending only on the family constants, with
\[
  \PR{(\cA_0\cap\cA_1\cap\cA_2\cap\cA_3\cap\cA_4)^c}\;\le\;e^{-n\delta}
  \qquad\text{for all }n\ge n_0.
\]

\emph{On the intersection, the band holds for all $\theta$.} Fix a
realization in $\cA_0\cap\cdots\cap\cA_4$ and write $(x,y)$ for
$(X^n,Y^n)$.

\emph{Denominator, all $\theta$ at once:} by the triangle inequality and
$\cA_0$,
\begin{equation}\label{eq:app-interf-denom}
  \norm{v_\theta(y)}\;\le\;h\norm{x}+\norm{\Delta(\theta)}+\norm{Z^n}
  \;\le\;\bigl(h_{\max}\sqrt P+2S_{\max}+\sqrt2\,\sigma_{\max}\bigr)\sqrt n
  \;=\;C_v\sqrt n.
\end{equation}

\emph{Upper band, all $\theta$ at once (projection again):} $W_x\perp x$,
$W_x\perp V$, so $P_{W_x}v_\theta(y)=P_{W_x}Z^n$ for every $\theta$, and
since $W_x\subset x^\perp$,
\[
  1-t_\theta(x,y)^2\;=\;\frac{\norm{P_{x^\perp}v_\theta(y)}^2}
  {\norm{v_\theta(y)}^2}
  \;\ge\;\frac{\norm{P_{W_x}Z^n}^2}{\norm{v_\theta(y)}^2}
  \;\ge\;\frac{\sigma_{\min}^2n/4}{C_v^2\,n}\;=\;\beta
\]
by $\cA_2$ and~\eqref{eq:app-interf-denom}. Hence
$t_\theta(x,y)\le\sqrt{1-\beta}=t_+$ for all $\theta\in\Theta$, with no net.

\emph{Lower band at the net points:} for each $i$,
\[
  \BRAi{x,v_{\theta_i}(y)}\;=\;h\,nP+\BRAi{x,\Delta_i}+\BRAi{x,Z^n}
  \;\ge\;n\bigl(h_{\min}P-2\epsilon\bigr)\;=\;n\,h_{\min}P/2
\]
by $\cA_1$, $\cA_3$ and the choice $\epsilon=h_{\min}P/4$, so
with~\eqref{eq:app-interf-denom},
\[
  t_{\theta_i}(x,y)
  \;=\;\frac{\BRAi{x,v_{\theta_i}(y)}}{\sqrt{nP}\,\norm{v_{\theta_i}(y)}}
  \;\ge\;\frac{n\,h_{\min}P/2}{\sqrt{nP}\cdot C_v\sqrt n}
  \;=\;\frac{h_{\min}\sqrt P}{2C_v}\;=\;2\zeta.
\]

\emph{Extension from the net to all $\theta$:} for $\theta\in\Theta$ pick
$\theta_i$ with $\norm{\theta-\theta_i}\le\eta$. Since $y\in\cE_n^Y$ (event
$\cA_4$) and $x\ne0$, \cref{lem:awgn-cosine} applies and gives
$t_\theta(x,y)\ge t_{\theta_i}(x,y)-L_t\eta\ge2\zeta-\zeta=\zeta=t_-$.

Combining the three displays: $t_\theta(x,y)\in[t_-,t_+]$ for every
$\theta\in\Theta$, and $y\in\cE_n^Y$, so $(x,y)\in\cE_n$
by~\eqref{eq:app-interf-cEn}. This proves
$\PR{(X^n,Y^n)\notin\cE_n}\le e^{-n\delta}$ for $n\ge n_0$, completing
part~(ii) and the lemma.
\end{IEEEproof}

\section{Proof of \texorpdfstring{\cref{lem:awgn-isi}}{the E\&NN Typicality Lemma}}
\label{app:isi-typical}

The proof is that of Appendix~\ref{app:interf-typical} with three
substitutions, which we state; every step not listed is verbatim as there,
with $(\theta,\Theta,v_\theta(y)=y-S(\theta))$ replaced by
$(h,\cH,v_h(y)=G_hy)$ and the cosine
$t_h(x,y)=\BRAi{x,G_hy}/(\norm{x}\,\norm{G_hy})$. The true channel is
$Y^n=H_{h_*}X^n+Z^n$; all constants depend only on the family constants
of hypotheses (a)--(c).

\begin{IEEEproof}[Proof sketch: the three substitutions]
\emph{(H2) by the spectral bound in place of the projection.} Hypothesis
(a) gives $\gamma_G\norm{y}\le\norm{G_hy}\le\Gamma_G\norm{y}$ for every $h$
and $y$, so on the $h$-free event
$\cE_n^Y=\{(\sigma_{\min}/\sqrt2)\sqrt n\le\norm{y}\le C_Y\sqrt n\}$,
$C_Y\triangleq\Gamma_H\sqrt P+\sqrt2\,\sigma_{\max}$, (H2) holds
simultaneously for all $h$ with $c_0=\gamma_G\sigma_{\min}/\sqrt2$, and
the denominators are bracketed by
$\gamma_G(\sigma_{\min}/\sqrt2)\sqrt n\le\norm{G_hy}\le\Gamma_GC_Y\sqrt n$.
The probability bound
$\PR{Y^n\notin\cE_n^Y}\le e^{-n/12}+e^{-n/8}$ follows from
\cref{lem:awgn-noncentral-dom} (the lower cut) and the $\chi^2$
tails~\eqref{eq:app-chi2} on
$\cA_0=\{\norm{Z^n}^2\le2n\sigma^2\}$ (the upper cut).

\emph{The lower band by alignment in place of the mean-offset events.}
Hypothesis (c) bounds the signal part of the score
\emph{deterministically} on the shell:
$\BRAi{x,G_{h_i}H_{h_*}x}
 =x^\top\tfrac12\bigl(G_{h_i}H_{h_*}+(G_{h_i}H_{h_*})^\top\bigr)x
 \ge\kappa nP$. With the net event
$\cB_1=\bigcap_i\{\abs{\BRAi{X^n,G_{h_i}Z^n}}\le\tfrac14\kappa nP\}$
(Gaussian tail, variance $\le\sigma^2\Gamma_G^2nP$ per point), the score at
each net point is at least $\tfrac34\kappa nP$, giving
$t_{h_i}\ge3\kappa\sqrt P/(4\Gamma_GC_Y)=2\zeta$ and the band value
$t_-=\zeta$.

\emph{The upper band by per-point covariance in place of the common
projection.} The conditional covariance of $P_{x^\perp}G_{h_i}Y^n$ on
$x^\perp$ is bounded below by $\sigma^2\gamma_G^2$ (hypothesis (a) on unit
vectors of $x^\perp$), so \cref{lem:awgn-noncentral-dom}
and~\eqref{eq:app-chi2} give
$\norm{P_{x^\perp}G_{h_i}y}^2\ge\gamma_G^2\sigma_{\min}^2n/4$ off an event
of probability $Ne^{-(n-1)/12}$, whence
$1-t_{h_i}^2\ge\beta\triangleq\gamma_G^2\sigma_{\min}^2/(4\Gamma_G^2C_Y^2)$
and $t_+\triangleq\tfrac12(1+\sqrt{1-\beta})$ (and $0<t_-<t_+<1$:
$\kappa\le\Gamma_G\Gamma_H$ gives $\zeta\le3/8<\tfrac12\le t_+$).

The intrinsic definition of $\cE_n$, the constant mesh
$\eta=L_t^{-1}\min\{\zeta,\tfrac12(1-\sqrt{1-\beta})\}$ with
$L_t=2K/c_0$ (hypothesis (b) supplies (H1) on $\cE_n^Y$; (H2) was
established first, so \cref{lem:awgn-cosine} applies --- no circularity),
the $O(1)$ net size, and the extension off the net are exactly as in
Appendix~\ref{app:interf-typical}.
\end{IEEEproof}